\documentclass[11pt,a4paper]{article}
\usepackage[a4paper]{geometry}
\usepackage{amssymb,latexsym,amsmath,amsfonts,amsthm}
\usepackage{graphicx}
\usepackage{epstopdf}
\usepackage{tikz}
\usepackage{pgflibraryshapes}
\usetikzlibrary{arrows,decorations.markings}
\usepackage{subfigure}
\usepackage{overpic}
\usepackage{fullpage}
\usepackage{comment,verbatim}
\usepackage{hyperref}
\usepackage{color}
\usepackage{mathrsfs}
\usepackage[title,titletoc]{appendix}
\usepackage{ulem}

\renewcommand{\Im}{\mathrm{Im}\,}

\newcommand{\ds}{\displaystyle}

\newcommand{\Boh}{\mathcal{O}}

\newtheorem{theorem}{Theorem}[section]
\newtheorem{lemma}[theorem]{Lemma}
\newtheorem{proposition}[theorem]{Proposition}

\newtheorem{corollary}[theorem]{Corollary}

\newtheorem{rhp}[theorem]{RH problem}

\newtheorem{Riemann-Hilbert Problem}{Definition}

\theoremstyle{definition}

\theoremstyle{remark}

\newtheorem{remark}[theorem]{Remark}

\numberwithin{equation}{section}

\begin{document}

\title{Asymptotics of the deformed Fredholm determinant of the confluent hypergeometric kernel}

\author{Dan Dai\footnotemark[1] \ and Yu Zhai\footnotemark[2]}

\renewcommand{\thefootnote}{\fnsymbol{footnote}}
\footnotetext[1]{Department of Mathematics, City University of Hong Kong, Tat Chee
Avenue, Kowloon, Hong Kong. E-mail: \texttt{dandai@cityu.edu.hk}}
\footnotetext[2]{Department of Mathematics, City University of Hong Kong, Tat Chee
Avenue, Kowloon, Hong Kong. E-mail: \texttt{yuzhai4-c@my.cityu.edu.hk}}

\date{\today}

\maketitle

\begin{abstract}
In this paper, we consider the deformed Fredholm determinant of the confluent hypergeometric kernel. This determinant represents the gap probability of the corresponding determinantal point process where each particle is removed independently with probability $1- \gamma$, $0 \leq \gamma <1$. We derive asymptotics of the deformed Fredholm determinant when the gap interval tends to infinity, up to and including the constant term. As an application of our results, we establish a central limit theorem for the eigenvalue counting function and a global rigidity upper bound for its maximum deviation.
\end{abstract}

\tableofcontents

\section{Introduction}

Eigenvalue spacing statistics are one of the fundamental problems in random matrix theory. There are mainly two regimes when we consider properties of the eigenvalue distribution as the matrix size gets large. In the global regime, the limiting eigenvalue distributions are very different in various matrix ensembles. However, in the local regime, the eigenvalue  distributions only rely on symmetry classes and some general characteristics of the ensembles and  display fascinating universal features; for example, see \cite{Erd:Yau2012,Kui11,Meh2004,Pas1992}. Let us focus on a random unitary ensemble with the following joint eigenvalue probability density
\begin{equation} \label{eq:ue-def}
p(x_1, \cdots, x_n) = \frac{1}{Z_n} \prod_{1 \leq i<j \leq n}(x_i-x_j)^2 \prod_{j=1}^n w(x_j),
\end{equation}
where $Z_n$ is a normalization constant and $w(x)$ satisfies certain regularity conditions. Near an interior point where the limiting mean eigenvalue density is positive (after an appropriate scaling if necessary), the limiting eigenvalue correlation kernel is given by the sine kernel. This result has been established for a large class of functions $w(x)$ in \eqref{eq:ue-def} and is well-known as the bulk universality. For instance, Lubinsky \cite{Lubi2009} proved that the bulk universality holds under an extremely weak condition: $w(x)$ is positive and continuous at the  interior point. If there exist singularities inside the bulk, the limiting correlation kernel near the interior singular point changes accordingly. One class of important singularities is the Fisher-Hartwig class, which combines a root-type and a jump-type singularity; see \cite{Fish:Har1968}. Let us illustrate with a unitarily invariant ensemble with quadratic potential and a Fisher-Hartwig singularity at the origin, where the joint probability density is given by \eqref{eq:ue-def} with
\begin{equation}
w(x) = |x|^{2\alpha} \chi_\beta(x) e^{-x^2}, \qquad \alpha > -\frac{1}{2},
\end{equation}
and
\begin{equation} \label{def:chi-beta}
\chi_\beta(x)= \begin{cases}
e^{\pi i\beta}, \quad & x<0,\\
     e^{-\pi i\beta}, \quad & x>0,
\end{cases} \qquad   \beta \in i \mathbb{R}.
\end{equation}
In this case, the limiting correlation kernel near the origin is given by the confluent hypergeometric kernel \cite{Bor:Dei2002,Bor:Ols2001} below
\begin{equation}\label{chgkernel}
  K^{(\alpha,\beta)}(x,y)=\frac{1}{2\pi i}\frac{\Gamma(1+\alpha+\beta)\Gamma(1+\alpha-\beta)}{\Gamma(1+2\alpha)^2}\frac{\mathbb{A}(x)\mathbb{B}(y)-\mathbb{A}(y)\mathbb{B}(x)}{x-y},
\end{equation}
where
\begin{equation}
  \mathbb{A}(x)=\chi_\beta(x)^{\frac{1}{2}}|2x|^{\alpha}e^{-ix}\phi(1+\alpha+\beta,1+2\alpha,2ix), \quad \mathbb{B}(x)=\overline{\mathbb{A}(x)},
\end{equation}
and $\phi(a,b,z)$ is the confluent hypergeometric function (cf. \cite[Chap. 13]{NIST})
\begin{equation}
  \phi(a,b,z)=1+\sum_{k=1}^{\infty}\frac{a(a+1)\cdots(a+k-1)z^k}{b(b+1)\cdots(b+k-1)k!}.
\end{equation}
Besides the unitary ensembles mentioned above, the confluent hypergeometric kernel $K^{(\alpha,\beta)}(x,y)$  \eqref{chgkernel} also appears in the study of circular unitary ensembles with a Fisher-Hartwig singularity \cite{Dei:Kra:Vas2011}, and ergodic measures in infinite random matrices \cite{Bor:Dei2002, Bor:Ols2001}. Moreover, it reduces to the type-I Bessel kernel and sine kernel when the parameters $\alpha$ and $\beta$ equal 0. To be precise, we have
\begin{align}
K^{(\alpha,0)}(x,y) & = K^{\textrm{Bess,1}} (x,y) = \frac{|x|^\alpha |y|^\alpha }{ x^\alpha y^\alpha}  \frac{\sqrt{xy}}{2}
\frac{J_{\alpha  + \frac{1}{2}}(x) J_{\alpha  - \frac{1}{2}}(y) - J_{\alpha  - \frac{1}{2}}(x) J_{\alpha  + \frac{1}{2}}(y) }{x-y} , \label{eq: chg-bessel} \\
K^{(0,0)}(x,y) & = K^{\sin} (x,y) = \frac{\sin (x-y)}{\pi (x-y)}. \label{eq: chg-sine}
\end{align}

With the limiting correlation kernel, one may further investigate the eigenvalue spacing by considering the gap probability, namely, the probability that there is no eigenvalue in a given interval. Since eigenvalues of matrices drawn from unitary ensembles form a determinantal point process, the gap probability can be expressed as a Fredholm determinant. Let $\mathcal{K}^{(\alpha,\beta)}_s$  be the operator acting on $L^2(-s,s)$ with the confluent hypergeometric kernel given in \eqref{chgkernel}. Then, the Fredholm determinant
\begin{equation} \label{eq: F-det}
\det(I-\mathcal{K}^{(\alpha,\beta)}_s)
\end{equation}
gives the probability, in the bulking scaling limit, that there is no eigenvalue in $(-s,s)$ for the unitary ensemble \eqref{eq:ue-def} with a Fisher-Hartwig singularity at the center. Observing the relation between $\det(I-\mathcal{K}^{(\alpha,\beta)}_s)$ and a Toeplitz determinant, Deift et al. \cite{Dei:Kra:Vas2011} derived the large gap asymptotics as $s \to \infty$,  including the constant term. Recently, an integral expression has been established in Xu and Zhao \cite[Theorem 3]{Xu:Zhao2020}:
\begin{equation} \label{eq:F-TW-formula}
\det(I-\mathcal{K}^{(\alpha,\beta)}_s) = \exp\left( \int_0^{-4is} H(\tau; \alpha, \beta) d\tau \right),
\end{equation}
where $H(\tau; \alpha, \beta) $ is the Hamiltonian of a coupled Painlev\'e V system. The above representation is also regarded as a Tracy-Widom type formula due to its similarity to the integral representation of the celebrated Tracy-Widom distribution; see \cite{Tra:Wid1994}.

In the present paper, we will consider the deformed Fredholm determinant
\begin{equation} \label{eq: F-det-deform}
\det(I-\gamma \mathcal{K}^{(\alpha,\beta)}_s), \qquad \gamma \in [0,1),
\end{equation}
which gives the gap probability in $(-s,s)$ when each eigenvalue is removed independently with probability $1- \gamma$. Obviously when $\gamma=1$, this is the gap probability mentioned in \eqref{eq: F-det}. The deformed Fredholm determinant \eqref{eq: F-det-deform} is related to the well-known "thinning" operation in the context of point processes; see \cite{IPSS2008}. It is introduced to the study of random matrices by Bohigas and Pato \cite{Boh:Pato2006,Boh:Pato2004}. In the literature, the undeformed and deformed distributions have been well studied for the classical kernels such as the sine kernel, Bessel kernel and Airy kernel (see e.g. \cite{Bas:Wid1983,Boh:Car:Pato2009,Both:Buc2018,Both:Dei:Kra2015,Bud:Bus1995,Cha2021-Bessel,Cha2021-sine,Dei:Its:Kra2008,Ehr2010}), as well as the non-classical ones, such as the Pearcey kernel and the hard edge Pearcey kernel in \cite{Cha:Mor2021,Dai:Xu:Zhang2022,Dai:Xu:Zhang2022-2}. Although the integral expressions for the Fredholm determinant may look very similar between the undeformed case $\gamma = 1$ and the deformed one $0\leq \gamma <1$, the large gap asymptotics are essentially different. Non-trivial transitions take place when $\gamma \to 1^-$; for example, see \cite{Both2016,Both:Buc2018,Both:Dei:Kra2015}. Moreover, with the additional deforming parameter $\gamma$, the Fredholm determinant of the form \eqref{eq: F-det-deform} can be understood as a moment generating function. Then, one can apply its large gap asymptotics to obtain more information about the associated point process, such as central limit theorems for the counting function and global rigidity upper bounds for the eigenvalue fluctuation; see \cite{Cha2021,Cha2021-Bessel,Cha:Cla2020,Sosh2000}.

In the present work, we will derive the large gap asymptotics of \eqref{eq: F-det-deform} for $0 \leq \gamma < 1$, including the notoriously difficult constant term. Following similar ideas in \cite{Dei:Kra:Vas2011,Xu:Zhao2020}, we will make use of its relation with a Toeplitz determinant possessing Fisher-Hartwig singularities. Then, the Deift-Zhou nonlinear steepest descent method for Riemann-Hilbert (RH) problems (cf. \cite{Dei:Its:Zhou1997,dkmvz1999}) will be applied to derive asymptotics of the Toeplitz determinant. After modifying the model RH problem in \cite{Xu:Zhao2020} to suit the local parametrix construction, we will re-establish the integral representation \eqref{eq:F-TW-formula} for $0 \leq \gamma < 1$, whereas the Hamiltonian $H$ in the integrand is related to a different class of solutions to the  coupled Painlev\'e V system.

\section{Statement of results}

\subsection{Large gap asymptotics}

Our main result for the large gap asymptotics is given in the following theorem.

\begin{theorem}\label{mainresult}
Let $\mathcal{K}^{(\alpha,\beta)}_s$  be the operator acting on $L^2(-s,s)$ with the confluent hypergeometric kernel given in \eqref{chgkernel}.  Then, as $s \to +\infty$, we have
  \begin{equation}\label{fdeterminant}
    \det (I - \gamma \mathcal{K}^{(\alpha,\beta)}_s) =
   e^{2\alpha \pi c}  \left(G(1+ic)G(1-ic)\right)^2 (4s)^{2c^2} e^{-4cs}
    \left[1+O\left(\frac{1}{s}\right)\right], \quad  \gamma \in [0,1),
  \end{equation}
  for $\alpha$, $\beta$ and $\gamma$ in any compact subsets of $(-\frac{1}{2}, \infty)$, $i \mathbb{R}$ and $[0,1)$, respectively. Here,  $G(z)$ is the Barnes G-function and $c$ is a positive real constant defined as
  \begin{equation}\label{ic}
    c:=c(\gamma) = -\frac{1}{2 \pi}\ln (1-\gamma).
  \end{equation}
\end{theorem}

The above asymptotics show that the deformed gap probability, where each eigenvalue is removed independently with probability $1-\gamma$, is exponentially small while the gap gets large. When $\gamma$ increases in $(0,1)$,  the eigenvalues are less likely removed, so that the gap probability gets smaller. This agrees with the asymptotics in \eqref{fdeterminant} where the constant $c$ in the exponential factor is a positive and increasing function in $\gamma$. Recall the large gap asymptotics for $\gamma = 1$ in \cite[Theorem 1]{Dei:Kra:Vas2011}, which is much smaller with the exponent of the leading order in the exponential term doubled. More precisely, the exponential term $e^{-4cs}$ for $\gamma \in [0,1)$ in \eqref{fdeterminant}  is replaced by a ``super-exponential" term $e^{-\frac{s^2}{2} + 2 \alpha s}$ for $\gamma = 1$. This phenomenon seems to be universal for unitarily invariant matrix ensembles, which has been observed for the gap probability associated with classical Airy, sine and Bessel kernels in \cite{Bas:Wid1983,Both:Buc2018,Both:Its:Prok2019,Bud:Bus1995}, as well as the Pearcey kernel in \cite{Dai:Xu:Zhang2022}. Note that, in all of the asymptotic results mentioned, there is a non-trivial leading order transition as $\gamma \to 1^-$ and $s \to \infty$ simultaneously. The transition problem of $ \det (I - \gamma \mathcal{K}^{(\alpha,\beta)}_s)$ will not be addressed in the present work. We refer the interested readers for a full transition picture of the sine process in \cite{Both:Dei:Kra2018} and references therein.

In the derivation of the large gap asymptotics, it is always a challenging problem to determine the constant factor explicitly. In \eqref{fdeterminant}, we have successfully obtained it with the aid of the integral representation \eqref{eq:F-TW-formula}, which will be established for $\gamma \in [0,1)$ in Lemma \ref{lem: Integral-H}. The Hamiltonian $H$ in the integrand is closely related to a tau-function of the coupled Painlev\'e V system in dimension four discovered in \cite{Kawa2018}. Recently, a breakthrough has been achieved in \cite{Both:Its:Prok2019,Its:Lis:Pro2018,Its:Lis:Tyk2015,Its:Pro2016}, which shows that the tau-functions are equal to the classical action differential up to a total differential. Therefore, one can establish crucial differential identities for the $\tau$-function with respect to parameters ($\gamma$ in our case). As our subsequent asymptotic analysis is uniform for $\gamma$ in compact subsets of $[0,1)$, we can finally evaluate the constant explicitly by integrating with respect to $\gamma$; see similar derivations in \cite{Dai:Xu:Zhang2022,Dai:Xu:Zhang2022-2,Xu:Dai2019}.

\begin{remark}
  The parameter $\beta$ does not appear explicitly in the asymptotic formula \eqref{fdeterminant}, but is expected to be included in the $O(\frac{1}{s})$ term. This indicates that, on the level of large gap asymptotics, the jump type singularity $\chi_\beta(x)$ in \eqref{def:chi-beta} is less significant than the algebraic singularity $|x|^{-2\alpha}$ and the thinning operation at least for a symmetric gap $(-s,s)$. It would be interesting to see whether the similar situation still occurs for a non-symmetric gap, say $(-s_0, s_1)$.
\end{remark}

Let us recall the relation between the confluent hypergeometric kernel, the type-I Bessel kernel and sine kernel in \eqref{eq: chg-bessel} and \eqref{eq: chg-sine}. As Theorem \ref{mainresult} holds for $\alpha$ and $\beta$ in compact subsets containing 0, we obtain the following large gap asymptotics related to the  type-I Bessel kernel and sine kernel immediately.

\begin{corollary}
Let $\mathcal{K}^{\textrm{Bess,1}}_s$ and $\mathcal{K}^{\sin}_s$  be the operators acting on $L^2(-s,s)$ with the type-I Bessel kernel and the sine kernel given in  \eqref{eq: chg-bessel} and \eqref{eq: chg-sine}, respectively.  Then, as $s \to +\infty$, we have
  \begin{align}
  & \det (I - \gamma \mathcal{K}^{\textrm{Bess,1}}_s) =
   e^{2\alpha \pi c}  \left(G(1+ic)G(1-ic)\right)^2 (4s)^{2c^2} e^{-4cs}
    \left[1+O\left(\frac{1}{s}\right)\right], \quad  \gamma \in [0,1), \\
     & \det (I - \gamma \mathcal{K}^{\sin}_s) =
   \left(G(1+ic)G(1-ic)\right)^2 (4s)^{2c^2} e^{-4cs}
    \left[1+O\left(\frac{1}{s}\right)\right], \quad  \gamma \in [0,1), \label{eq:sine-kernel-asy}
  \end{align}
  for $\alpha$ and $\gamma$ in any compact subsets of $(-\frac{1}{2}, \infty)$ and $[0,1)$, respectively.
\end{corollary}

The above result for the sine kernel in \eqref{eq:sine-kernel-asy} agrees with those in \cite{Bas:Wid1983,Bud:Bus1995}; see also \cite[Eq. (1.33)]{Both:Its:Prok2019}.

\subsection*{Applications}

Besides the large gap asymptotics, Theorem \ref{mainresult} gives us more information about the determinantal point process characterized by the confluent hypergeometric kernel \eqref{chgkernel}. After choosing special parameter $\gamma$ in the Fredholm determinant \eqref{eq: F-det-deform} and viewing it as a moment generating function, we are able to obtain a central limit theorem for the related point process and establish a global rigidity upper bound for the maximum deviation of the points.

Let $N(s)$ be the random variable that represents the number of points in the process falling into the interval $(-s,s)$. Then, we have the following results.
\begin{corollary}
As $s \to +\infty$, we have
\begin{equation} \label{exp-var-asy}
  \mathbb{E}(N(s))= \mu(s)+O(s^{-1}\ln s),
  \qquad \mathrm{Var}(N(s))= \sigma(s)^2+\frac{1+\gamma_\mathrm{E} +2\ln2}{\pi^2}+O(s^{-1}(\ln s)^2),
 \end{equation}
 where $\gamma_{{\rm E}} = -\Gamma'(1) \approx 0.57721$ is Euler's constant, 
 \begin{equation}\label{musigma}
  \mu(s)=\frac{2s}{\pi}-\alpha, \qquad \sigma(s)^2=\frac{\ln s }{\pi ^2}.
 \end{equation}
The random variable $\frac{N(s) -\frac{2s}{\pi} + \alpha}{\sqrt{\ln s}/\pi  } $ converges in distribution to the normal law $ \mathcal{N}(0,1)$, that is,
 \begin{equation}\label{clt}
  \frac{N(s) -\frac{2s}{\pi} + \alpha}{\sqrt{\ln s}/\pi } \stackrel{d}{\longrightarrow} \mathcal{N}(0,1), \qquad \textrm{as } s \to +\infty.
 \end{equation}
 Moreover, for any $\epsilon > 0$, we have
 \begin{equation}\label{grub}
  \lim_{x \to +\infty}\mathbb{P}\left(\sup_{s>x}\left|
  \frac{N(s)-\frac{2s}{\pi}+\alpha}{\ln s}\right|\leq \frac{\sqrt{2}}{\pi}+\epsilon\right)=1.
 \end{equation}
\end{corollary}

\begin{proof}
First, we recall the well-known joint probability generating function of the occupancy number $N(s)$ below:
 \begin{equation} \label{eintermsoffdeterminant}
  \mathbb{E}(e^{-2 \pi \nu N(s)})=\sum_{k=0}^{\infty}\mathbb{P}(N(s)=k)e^{-2 \pi \nu N(s)} =\det (I-(1-e^{-2 \pi \nu})\mathcal{K}^{(\alpha,\beta)}_s), \qquad \nu \geq 0,
 \end{equation}
 cf. \cite{johan2006,Sosh2000-2}. Based on the above identity, on the one hand, we have
 \begin{equation}\label{expanssionfore}
  \mathbb{E}(e^{-2 \pi \nu N(s)})=1-2\pi  \mathbb{E}(N(s)) \nu + 2\pi^2  \mathbb{E}(N(s)^2) \nu^2+O(\nu^3), \qquad   \nu \to 0.
 \end{equation}
 On the other hand, it follows by setting $c= \nu$ in \eqref{fdeterminant} that
 \begin{equation}\label{expanssionfore1}
    \det (I-(1-e^{-2 \pi \nu})\mathcal{K}^{(\alpha,\beta)}_s)
    = 2^{4 \nu ^2}\left(G(1+i\nu)G(1-i\nu)\right)^2e^{-2\pi \mu(s)\nu +2 \pi^2 \sigma(s)^2 \nu^2}(1+O(s^{-1})), \\
 \end{equation}
 as $s \to + \infty$, where $\mu(s)$ and $\sigma(s)$ are given in \eqref{musigma}. As
 \begin{equation}
G(1+\nu)=1+\frac{\ln(2\pi)-1}{2} \nu+\left(\frac{(\ln(2\pi)-1)^2}{8}-\frac{1+\gamma_{\textrm{E}}}{2}\right) \nu^2+\Boh(\nu^3),\qquad \nu\to 0,
\end{equation}
then we have the following expansion for the leading term in \eqref{expanssionfore1}:
\begin{equation} \label{expanssionfore1-exp}
  \begin{aligned}
    & 2^{4 \nu ^2}\left(G(1+i\nu)G(1-i\nu)\right)^2e^{-2\pi \mu(s)\nu +2 \pi^2 \sigma(s)^2 \nu^2}  \\
     & \quad = 1- 2 \pi \mu(s)\nu +2\pi^2\left(\frac{2\ln2}{\pi^2}+\frac{1+\gamma_\mathrm{E}}{\pi^2}+\mu(s)^2+\sigma(s)^2\right)\nu^2+O(\nu^3) , \quad   \nu \to 0.
  \end{aligned}
 \end{equation}
One can see from \eqref{eintermsoffdeterminant} and \eqref{expanssionfore} that 
\begin{align*}
\mathbb{E}(N(s)) &= -\frac{1}{2 \pi} \frac{\partial}{\partial \nu} \det\left(I-(1-e^{-2\pi \nu})  \mathcal{K}^{(\alpha,\beta)}_s \right) \Big|_{\nu = 0}, \\
{\rm Var}(N(s)) &= \frac{1}{4 \pi^2} \left[ \frac{\partial^2}{\partial \nu^2} \det\left(I-(1-e^{-2\pi \nu}) \mathcal{K}^{(\alpha,\beta)}_s  \right) - \Big( \frac{\partial}{\partial \nu} \det\left(I-(1-e^{-2\pi \nu}) \mathcal{K}^{(\alpha,\beta)}_s \right) \Big)^2 \right] _{\nu = 0}.
\end{align*}
Then, we obtain asymptotics of $\mathbb{E}(N(s))$ and ${\rm Var}(N(s)) $ in \eqref{exp-var-asy} by substituting \eqref{expanssionfore1} and \eqref{expanssionfore1-exp} into the above formulas, where the additional $\ln s$ factor in the error terms comes from the derivatives with respect to $\nu$; see also \eqref{rstoinfinity-diff} below.

Next, it follows from \eqref{eintermsoffdeterminant} and \eqref{expanssionfore1} that
 \begin{equation}
  \mathbb{E}\left(e^{\frac{N(s)-\mu(s)}{\sqrt{\sigma(s)^2}}\cdot t}\right)\to e^{\frac{t^2}{2}}, \qquad s \to +\infty,
 \end{equation}
 which gives us the  central limit theorem in \eqref{clt}.

The proof of \eqref{grub} is based on \cite[Theorem 1.2]{Cha:Cla2020}.  Let $Y=\{\xi_k \}_{k \geq 1}$ be a random point process associated with the confluent hypergeometric kernel \eqref{chgkernel}. Note that the original theorem in \cite{Cha:Cla2020} only applies to point processes which have almost surely a smallest particle. As $Y$ does not have a smallest particle almost surely, we instead consider
\begin{equation}
  X=\{|\xi_k|;\xi_k \in Y\};
 \end{equation}
see a similar treatment for the maximum deviation of the Pearcey process in \cite{Cha2021}. Let $\tilde{N}(s)$ be the number of particles in $X$ which are located in the interval $[0,s]$. Then, we have
\begin{equation}
  \tilde{N}(s)=\sharp(X \cap [0,s])=\sharp(Y \cap [-s,s])=N(s).
\end{equation}
Obviously, the function $\mu(s)$ and $\sigma(s)$ given in \eqref{musigma} are strictly increasing and differentiable with respect to $s$. Moreover, we have
\begin{equation}
  \sqrt{\sigma^2(\mu^{-1}(s))}
   =\frac{\sqrt{\ln s}}{\pi}\left(1+o(1)\right), \qquad \mathrm{as} \ s \to +\infty.
\end{equation}
Therefore, all conditions in \cite[Theorem 1.2]{Cha:Cla2020} are justified, and the upper bound \eqref{grub} follows accordingly.

This finishes the proof of the corollary.
\end{proof}

One may also expect a global rigidity lower bound for the point similar to \eqref{grub}, namely
\begin{equation}
    \lim_{x \to +\infty}\mathbb{P}\left(\sup_{s>x}\left|
    \frac{N(s)-\frac{2s}{\pi}+\alpha}{\ln s}\right|\geq \frac{\sqrt{2}}{\pi}-\epsilon\right)=1.
   \end{equation}
However, this is a much more difficult problem. For example, see the recent progress in \cite{Cla:Fahs:Lam:Webb2021}, where the convergence of the counting function to a  Gaussian multiplicative chaos measure is used to establish the lower bounds.

\subsection{The related coupled Painlev\'e V system and its asymptotics}

As we have mentioned before, the explicit evaluation of the constant term in \eqref{fdeterminant} relies on the integral representation \eqref{eq:F-TW-formula}. For $0 \leq \gamma < 1$, we will show this representation still holds in Lemma \ref{lem: Integral-H} below. The Hamiltonian $H$ in the integrand is related to a coupled Painlev\'e V system in dimension four. In case readers are interested in the related integrable system and properties of special solutions, we list the corresponding results below.

The Hamiltonian  $H(\tau; \alpha, \beta)=H(u_1,u_2,v_1,v_2,\tau;\alpha,\beta, \gamma)$ in  \eqref{eq:F-TW-formula} is given by
 \begin{equation}\label{sh}
 \begin{split}
  H(\tau; \alpha, \beta)=& \ \frac{1}{2}\left( H_V(u_1,v_1,\frac{\tau}{2},\alpha,\beta)- H_V(u_2,v_2,-\frac{\tau}{2},\alpha,\beta) \right) \\
 &  +\frac{1}{\tau} u_1u_2(v_1+v_2)(v_1-1)(v_2-1),
  \end{split}
 \end{equation}
 where $H_V(u,v,\tau,\alpha,\beta)$ is the Hamiltonian for the Painlev\'e V equation as follows
 \begin{equation}\label{hv}
  \tau H_V(u,v,\tau,\alpha,\beta)=u^2v(v-1)^2-\tau uv-\alpha u(v^2-1)-\beta u(v-1)^2.
 \end{equation}
 The Hamiltonian system
 \begin{equation}\label{dh}
  \frac{dv_k}{d\tau}=\frac{\partial H}{\partial u_k},\qquad
  \frac{du_k}{d\tau}=-\frac{\partial H}{\partial v_k}, \qquad k=1,2,
 \end{equation}
gives us the  following coupled Painlev\'e V system of dimension four (cf. \cite{Kawa2018})
\begin{small}
 \begin{equation}\label{p5system}
  \left\{
   \begin{aligned}
     & \tau \frac{d u_1}{d\tau}=\frac{\tau}{2}u_1-u_1^2(v_1-1)(3v_1-1)-u_1u_2(v_2-1)(2v_1+v_2-1)+2(\alpha+\beta)u_1v_1-2\beta u_1, \\
     & \tau\frac{d u_2}{d\tau}=-\frac{\tau}{2}u_2-u_2^2(v_2-1)(3v_2-1)-u_1u_2(v_1-1)(v_1+2v_2-1)+2(\alpha+\beta)u_2v_2-2\beta u_2, \\
     & \tau\frac{d v_1}{d\tau}=-\frac{\tau}{2}v_1+2u_1v_1(v_1-1)^2+u_2(v_1+v_2)(v_1-1)(v_2-1)-\alpha(v_1^2-1)-\beta(v_1-1)^2,      \\
     & \tau\frac{d v_2}{d\tau}=\frac{\tau}{2}v_2+2u_2v_2(v_2-1)^2+u_1(v_1+v_2)(v_1-1)(v_2-1)-\alpha(v_2^2-1)-\beta(v_2-1)^2.
   \end{aligned}
  \right.
 \end{equation}
\end{small}
Then, we have the following asymptotics for the Hamiltonian $H(\tau; \alpha, \beta)$ and the solutions $u_{k}(\tau)$, $v_{k}(\tau)$, $k=1,2$, of the coupled Painlev\'e V system.

\begin{theorem}\label{huvasymptotics}
With $\alpha >-\frac{1}{2}$, $\gamma \in [0,1)$, $\beta \in i \mathbb{R}$ and the constant $c$ given in \eqref{ic}, the Hamiltonian $H(\tau)=H(\tau; \alpha, \beta)$ in \eqref{sh} is pole-free for $\tau \in -i(0,+\infty)$ and satisfies the following asymptotic behavior
  \begin{equation}\label{hinfinity}
    H(\tau) = -ic+\frac{2c^2}{\tau}+O\left(\tau^{-2}\right),
                \qquad i\tau \to +\infty.
  \end{equation}
  Moreover, there exist solutions to the coupled Painlev\'e V system \eqref{p5system} with the following asymptotic behaviors as $i\tau \to +\infty$:
  \begin{eqnarray}
    u_1(\tau) & = & ic\frac{\Gamma(1+\alpha+\beta)\Gamma(1-ic)}{\Gamma(1+\alpha-\beta)\Gamma(1+ic)} 2^{2\beta}e^{\pi i(\alpha-\beta)}e^{-\pi  c}e^{\frac{\tau}{2}} \tau^{2(ic-\beta)}
    \left(1+O(\tau^{-1})\right),\label{u1infinity}    \\
    v_1(\tau) & = & \frac{\Gamma(1+\alpha-\beta)\Gamma(1+ic)}{\Gamma(1+\alpha+\beta) \Gamma(1-ic)}2^{-2\beta}e^{-\pi i(\alpha-\beta)}e^{\pi  c}
    e^{-\frac{\tau}{2}} \tau^{-2(ic-\beta)}\left(1+O(\tau^{-1})\right),\label{v1infinity}     \\
    u_2(\tau) & = & -ic\frac{\Gamma(1+\alpha+\beta)\Gamma(1+ic)}{\Gamma(1+\alpha-\beta)\Gamma(1-ic)}2^{2\beta}e^{-\pi i(\alpha+\beta)}e^{\pi  c} e^{-\frac{\tau}{2}} \tau^{-2(ic+\beta)}\left(1+O(\tau^{-1})\right),\label{u2infinity}     \\
    v_2(\tau) & = & \frac{\Gamma(1+\alpha-\beta)\Gamma(1-ic)}{\Gamma(1+\alpha+\beta)\Gamma(1+ic)}2^{-2\beta}e^{\pi i(\alpha+\beta)}e^{-\pi  c}
    e^{\frac{\tau}{2}} \tau^{2(ic+\beta)}\left(1+O(\tau^{-1})\right).\label{v2infinity}
  \end{eqnarray}
\end{theorem}

\begin{theorem}\label{huvasymptotics1}
  The Hamiltonian and the solutions to the coupled Painlev\'e V system also have the following asymptotic behaviors as $i\tau \to 0^+$:
    \begin{eqnarray}
      H(\tau) & = & \frac{\gamma \,  \Gamma(1+\alpha-\beta)\Gamma(1+\alpha+\beta) \cos (\beta \pi) }{i \pi   2^{2\alpha+1}(2\alpha+1)(\Gamma(1+2\alpha))^2}|\tau|^{2 \alpha}+O\left(\tau^{2\alpha+1}\right),\label{hassto0}   \\
      u_1(\tau) & = & -\frac{\gamma \Gamma(1+\alpha-\beta)\Gamma(1+\alpha+\beta)}{ i \pi 2^{2\alpha+1} (\Gamma(1+2\alpha))^2}e^{-\pi i \beta}    |\tau|^{2\alpha}\left(1+O(\tau^{2\alpha+1})+O(\tau)\right),\label{u1assto0}   \\
      v_1(\tau) & = & 1+O(\tau^{2\alpha+1})+O(\tau),\label{v1assto0}\\
      u_2(\tau) & = & \frac{\gamma \Gamma(1+\alpha-\beta)\Gamma(1+\alpha+\beta)}{i \pi  2^{2\alpha+1} (\Gamma(1+2\alpha))^2}e^{\pi i \beta}
      |\tau|^{2\alpha}\left(1+O(\tau^{2\alpha+1})+O(\tau)\right),\label{u2assto0}\\
      v_2(\tau) & = & 1+O(\tau^{2\alpha+1})+O(\tau).\label{v2assto0}
    \end{eqnarray}
\end{theorem}

\begin{remark}
Comparing the above asymptotics results with those in \cite[Theorem 1]{Xu:Zhao2020} for the case $\gamma = 1$, one can see that the asymptotics as $i\tau \to 0^+$ are similar. Indeed, by taking $\gamma = 1$ in the above formulas  for  $i\tau \to 0^+$, one retrieve the asymptotics in \cite[Theorem 1]{Xu:Zhao2020}. However, the asymptotics as $i\tau \to +\infty$ are significantly different. One can observe similar phenomena for the classical Painlev\'e equations as well; for example, see \cite[Sec. 32.11]{NIST}. From a technical perspective, the proof of Theorem \ref{huvasymptotics} (as $i\tau \to +\infty$) is much harder than that of Theorem \ref{huvasymptotics1} (as $i\tau \to 0^+$).
\end{remark}

The rest of the paper is arranged as follows. With the inspiration of \cite{Dei:Kra:Vas2011}, we first establish a relation between the deformed Fredholm determinant and a Toeplitz determinant in Section \ref{sec: toep}. A RH problem related to the Toeplitz determinant is also constructed. Next, we set up a model RH problem for $\Psi(z; \tau)$ and derive several important differential identities in Section \ref{sec: model-RHP}. Then, we perform a steepest descent analysis for the RH problem associated with the Toeplitz determinant in Section \ref{analysisofy}, where the model problem for $\Psi(z; \tau)$ is adopted in the local parametrix construction. Sections \ref{Sec: large-s} and \ref{Sec: small-s} are devoted to the steepest descent analysis for the RH problem for $\Psi(z; \tau)$ as $i\tau \to +\infty$ and $i\tau \to 0^+$, respectively. The proofs of Theorems \ref{huvasymptotics} and \ref{huvasymptotics1} are also provided at the end of these two sections. Finally, in Section \ref{sec: main-proof}, we establish the integral representation for the deformed  Fredholm determinant and prove the large gap asymptotics in Theorem \ref{mainresult}.

\section{Relation with a Toeplitz determinant} \label{sec: toep}

In this section, we first establish a relation between the Fredholm determinant \eqref{eq: F-det-deform} and a Toeplitz determinant using similar ideas as in \cite{Cha2019,Dei:Kra:Vas2011,Xu:Zhao2020}. Then, we formulate a RH problem for the corresponding orthogonal polynomials on the unit circle, which is related to the Toeplitz determinant via a differential identity.

%
%

\subsection{The Fredholm determinant and the Toeplitz determinant}

Let $C$ be the unit circle and $C_t$ be an arc:
\begin{equation}
  C_t=\{e^{i \theta}; t \leq \theta \leq 2\pi - t \}, \qquad t \in (0, \pi).
\end{equation}
Consider a circular unitary ensemble with the following joint eigenvalue probability density
\begin{equation}
    p_n(\theta_1, \cdots, \theta_n) = \frac{1}{Z_n'} \prod_{j=1}^n |e^{i\theta_j}-1|^{2\alpha}e^{i \beta(\theta_j - \pi)}\prod_{j<k}(e^{i\theta_j}-e^{i\theta_k})^2, \qquad \theta_1, \cdots, \theta_n \in (0,2\pi),
\end{equation}
with $\alpha > -\frac{1}{2}$ and $\beta \in i \mathbb{R}$. Denote $A_k$ as the event that there are exactly $k$ eigenvalues in $C\setminus C_t$. Next, let us remove each eigenvalue independently with probability $1-\gamma$. It is easy to see that
\begin{multline} \label{eq: thinned-prob-1}
      \mathbb{P}(\textrm{all $n$ eigenvalues  in $C_t$ after thinning})
    = \sum_{k=0}^{n} (1-\gamma)^k \mathbb{P}(A_k)  \\
   \qquad \qquad \qquad = \sum_{k=0}^n (1-\gamma)^k
  \int_{A_k}
  \frac{1}{Z_n'} \prod_{j=1}^n |e^{i\theta_j}-1|^{2\alpha}e^{i \beta(\theta_j - \pi)}\prod_{j<k}(e^{i\theta_j}-e^{i\theta_k})^2 \prod_{j=1}^n d\theta_j.
\end{multline}
Define
\begin{equation} \label{eq: f-def}
 f(z;t)=\left\{
 \begin{aligned}
   & 1-\gamma, \qquad & z \in C\setminus C_t,             \\
   & 1, & \mathrm{otherwise}, \\
 \end{aligned}
 \right.
\end{equation}
and note that
\begin{equation}
  [0,2\pi)^n = \bigcup_{k=0}^n A_k,
\end{equation}
it then follows from \eqref{eq: thinned-prob-1} that
\begin{equation} \label{eq: thinned-prob-Toep}
  \begin{aligned}
  &   \mathbb{P}(\textrm{all $n$ eigenvalues  in $C_t$ after thinning}) \\
  & =
  \frac{1}{Z_n'} \int_{[0,2\pi]^n}
   \prod_{j=1}^n |e^{i\theta_j}-1|^{2\alpha}e^{i \beta(\theta_j - \pi)}\prod_{j<k}(e^{i\theta_j}-e^{i\theta_k})^2 \prod_{j=1}^n f(e^{i \theta_j}) d\theta_j = \frac{D_n(t)}{D_n(0)},
 \end{aligned}
\end{equation}
where $D_n(t)$ is a Toeplitz determinant given below:
\begin{equation}\label{dn}
D_n(t) = \det\left(\frac{1}{2\pi}
  \int_0^{2\pi}   |e^{i\theta}-1|^{2\alpha}e^{i \beta(\theta - \pi)} f(e^{i\theta}) \, e^{-i(k-j)\theta}d\theta
  \right)_{j,k=1}^n;
\end{equation}
see a similar derivation in \cite{Cha:Cla2017}. Finally, with an analogous argument as in \cite[Lemma 6]{Dei:Kra:Vas2011}, we have from \eqref{eq: thinned-prob-Toep} that
\begin{equation}\label{fredholmandtoeplitz}
    \det (I - \gamma \mathcal{K}^{(\alpha,\beta)}_s) = \lim_{n \to \infty} \frac{D_n(\frac{2s}{n})}{D_n(0)}.
\end{equation}

\subsection{Orthogonal polynomials on the unit circle}

It is well-known that Toeplitz determinants are related to orthogonal polynomials on the unit circle. Due to the explicit expression of the Toeplitz determinant in \eqref{dn}, we define a weight function as follows:
\begin{equation}\label{w}
 w(z;t)=\tilde{w}(z)f(z;t), \qquad z \in \mathbb{C} \setminus [0,+\infty),
\end{equation}
with $f(z;t)$ defined in \eqref{eq: f-def} and
\begin{equation} \label{def: w-tilde}
 \tilde{w}(z)=(z-1)^{2\alpha}z^{\beta-\alpha}e^{-\pi i (\alpha+\beta)},
\end{equation}
where both branches of $(z-1)^{2\alpha}$ and $z^{\beta-\alpha}$ are taken such that $\arg (\cdot) \in (0,2 \pi)$. Let
\begin{equation}\label{pn}
  p_n(z):=p_n(z;t)=\chi_n(t)z^n+\cdots,
\end{equation}
be polynomials orthonormal with respect to $w(z;t)$ on the unit circle, that is,
\begin{equation}\label{orthogonalrelation}
 \frac{1}{2\pi}\int_0^{2\pi} p_n(e^{i \theta};t) \,  \overline{p_m(e^{i \theta};t)} \, w(e^{i \theta};t)d\theta=\delta_{nm}.
\end{equation}
Then, $p_n(z)$ is related to the following RH problem.

%
%
%
%

\begin{rhp} \label{Y-RHP}
  The $2 \times 2$ matrix-valued function $Y(z)$ satisfies the following properties:
 \begin{itemize}
  \item [(a)] $Y(z)$ is analytic in $\mathbb{C} \setminus C$, where $C$ is taken in the counterclockwise direction.


  \item [(b)] $Y(z)$ satisfies the jump condition
        \begin{equation}
         Y_+(z)=Y_-(z)
         \begin{pmatrix}
          1 & z^{-n}w(z;t) \\ 0 & 1
         \end{pmatrix}, \qquad z \in C.
        \end{equation}
  \item [(c)] As $z \to \infty$, we have
        \begin{equation} \label{eq: Y-large-z}
         Y(z)=\left(I+O\left(\frac{1}{z}\right)\right)
         \begin{pmatrix}
          z^n & 0 \\ 0 & z^{-n}
         \end{pmatrix}.
        \end{equation}
  \item [(d)] Let $z_1=e^{it},\ z_2=e^{i(2\pi -t)}$ be the points where the jump discontinuity in $w(z;t)$ takes place. Then, we have
        \begin{equation}
         Y(z)=
         \begin{pmatrix}
          O(1) & O(\ln (z-z_i)) \\
          O(1) & O(\ln (z-z_i))
         \end{pmatrix}, \qquad \textrm{as } z \to z_i, \ i=1,2.
        \end{equation}
 \end{itemize}
\end{rhp}

The unique solution of the above RH problem is given by
\begin{equation}\label{ydef}
 Y(z)=
 \begin{pmatrix}
  \frac{p_n(z)}{\chi_n(t)} &
  \frac{1}{2\pi i \chi_n(t)}\int_C \frac{p_n(\zeta)w(\zeta)}{\zeta^n(\zeta-z)}d\zeta \vspace{2pt}  \\
  -\chi_{n-1}(t) p_{n-1}^{*}(z) &
  -\frac{\chi_{n-1}(t)}{2\pi i}\int_C \frac{p_{n-1}^{*}(\zeta)w(\zeta)}{\zeta^{n}(\zeta-z)}d\zeta
 \end{pmatrix},
\end{equation}
where $p_{n}^{*}(z)$ is the reverse polynomial of $p_n(z)$:
\begin{equation}
 p_{n}^{*}(z)=z^n \overline{p_n\left(\bar{z} ^{-1}\right)}=z^n \overline{p_n}(z^{-1}).
\end{equation}

\subsection{Differential identity}
At the end of this section, let us present an important differential identity which relates the Toeplitz determinant \eqref{dn} to the above RH problem for $Y$.
\begin{lemma}
 Let $D_n(t)$ be the Toeplitz determinant given in \eqref{dn}, then we have
 \begin{equation}\label{differentialindentityfordn}
  \frac{d}{dt}\ln D_n(t)=- \frac{\gamma}{2\pi} \sum_{i=1}^2
  z_i^{-n+1}\tilde{w}(z_i) \lim_{z \to z_i}\left(Y^{-1}(z) \frac{d}{dz}Y(z)\right)_{21}.
 \end{equation}
\end{lemma}
\begin{proof}
The proof is similar to that of \cite[Lemma 8]{Dei:Kra:Vas2011} and \cite[Lemma 2]{Xu:Zhao2020} where $\gamma = 1$.

Recalling the following relation between a Toeplitz determinant and the leading coefficient of the orthonormal polynomials of $p_n(z)$ in \eqref{pn}:
 \begin{equation}
  D_n(t)=\prod_{k=0}^{n-1}\chi_k^{-2}(t),
 \end{equation}
 then we have
 \begin{equation}\label{differentialidentity1}
  \frac{d}{dt}\ln D_n(t)=-2\sum_{k=0}^{n-1}\chi_k^{-1}(t)\frac{d}{dt}\chi_k(t).
 \end{equation}
Setting $n=m$ in the orthogonal relation \eqref{orthogonalrelation} and differentiating both sides with respect to $t$, we get
\begin{equation}
   -\frac{\gamma}{2\pi}\left(\pi_n(z_1) \overline{\pi_n(z_1)}\tilde{w}(z_1)+\pi_n(z_2) \overline{\pi_n(z_2)}\tilde{w}(z_2)\right)=-2\chi_n^{-3}(t)\frac{d}{dt}\chi_n(t),
 \end{equation}
where $\pi_n(z) = \frac{1}{\chi_n(t)} p_n(z)$ is the monic polynomial. It then follows from the above two formulas that
 \begin{equation}\label{differentialidentityd}
   \frac{d}{dt}\ln D_n(t)
    =-\frac{\gamma}{2 \pi}\sum_{i=1}^2\sum_{k=0}^{n-1}p_k(z_i)\overline{p_k}(z_i^{-1})\tilde{w}(z_i).
 \end{equation}
Finally, we obtain the desired differential identity \eqref{differentialindentityfordn} with the aid of the Christoffel-Darboux formula,  the three-term recurrence relation for orthogonal polynomials on the unit circle (cf. \cite{Sze1975}), and the explicit expression of $Y(z)$ in \eqref{ydef}.
 \end{proof}

\begin{remark}
One may simply set $\gamma =1$ in \eqref{differentialindentityfordn} to retrieve the differential identity for the Toeplitz determinant in the undeformed case, namely, Eq. (52) in \cite{Dei:Kra:Vas2011} or Eq. (5.7) in \cite{Xu:Zhao2020}.
\end{remark}

Before we conduct a steepest descent analysis for the RH problem for $Y$, let us consider a model RH problem, which will be used in the local parametrix construction near $z = 1$.

\section{A model RH problem and its Lax pair} \label{sec: model-RHP}

In this section, we study a model RH problem for a function $\Psi(z;\tau)$ and an associated Lax pair. This RH problem is similar to the one in Xu and Zhao \cite[Sec. 2.1]{Xu:Zhao2020}, with an additional deforming parameter $\gamma$ in the problem. When $\gamma =1$, these two problems are identical.

\subsection{Model RH problem}
The $2 \times 2$ matrix-valued function $\Psi(z) = \Psi(z;\tau)$ satisfies a RH problem as follows.
\begin{rhp}\label{modelrhp}
 \hfill
 \begin{itemize}
  \item [(a)] $\Psi(z;\tau)$ is analytic for $z \in \mathbb{C} \setminus \{\cup_{i=1}^7 \Sigma_i \}$, where the oriented contours are defined as
        \begin{equation*}
         \begin{aligned}
           & \Sigma_1=1+e^{\frac{\pi i}{4}}\mathbb{R}^{+}, \qquad \Sigma_2=-1+e^{\frac{3\pi i}{4}}\mathbb{R}^{+}, \qquad
          \Sigma_3=-1+e^{-\frac{3\pi i}{4}}\mathbb{R}^{+}, \qquad     \\
           & \Sigma_4=e^{-\frac{\pi i}{2}}\mathbb{R}^{+}, \qquad\Sigma_5=1+e^{-\frac{\pi i}{4}}\mathbb{R}^{+}, \qquad \Sigma_6=(0,1), \qquad \Sigma_7=(-1,0);
         \end{aligned}
       \end{equation*}
        see Figure \ref{figure1}.
       \begin{figure}[h]
         \includegraphics[width=5.2in]{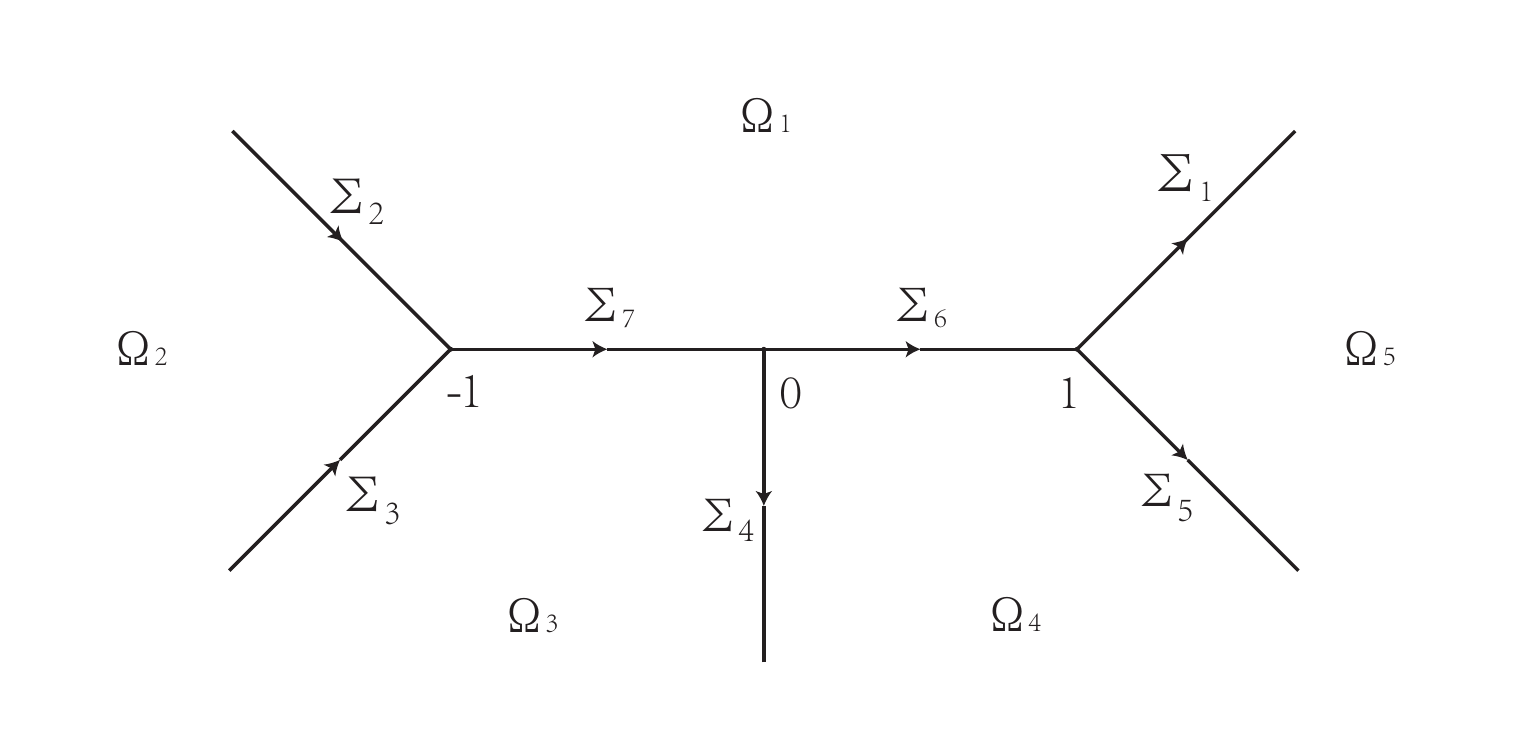}
         \centering
         \caption{Contours for the model RH problem. Regions $\Omega_i$, $i =1, \cdots, 5,$ are also depicted.}
         \label{figure1}
        \end{figure}

  \item [(b)] $\Psi$ has limiting values $\Psi_{\pm}(z;\tau)$ for $z \in \cup_{i=1}^7 \Sigma_i$, where $\Psi_\pm$ denotes the values of $\Psi$ taken from the left/right side of $\Sigma_i$. Moreover, they satisfy the following jump conditions
  \begin{equation}\label{jumpforpsi}
    \Psi_+(z;\tau)=\Psi_-(z;\tau)J_i(z), \qquad z \in \Sigma_i,
  \end{equation}
  where
  \begin{equation}
    \begin{aligned}
      & J_1(z) = \begin{pmatrix} 1 & 0 \\ e^{-\pi i(\alpha-\beta)} & 1 \end{pmatrix}, \quad
      J_2(z) = \begin{pmatrix} 1 & 0 \\ e^{\pi i(\alpha-\beta)} & 1 \end{pmatrix}, \quad
      J_3(z) = \begin{pmatrix} 1 & -e^{-\pi i(\alpha-\beta)} \\ 0 & 1 \end{pmatrix}, \\
      & J_4(z) = e^{2\pi i \beta \sigma_3}, \quad
      J_5(z) = \begin{pmatrix} 1 & -e^{\pi i(\alpha-\beta)} \\ 0 & 1 \end{pmatrix}, \quad
      J_6(z) = \begin{pmatrix} 0 & -e^{\pi i(\alpha-\beta)} \\ e^{-\pi i(\alpha-\beta)} & 1-\gamma \end{pmatrix}, \\
      & J_7(z) = \begin{pmatrix} 0 & -e^{-\pi i(\alpha-\beta)} \\ e^{\pi i(\alpha-\beta)} & 1-\gamma \end{pmatrix}.
    \end{aligned}
  \end{equation}

  \item [(c)] When $z \to \infty$, we have
        \begin{equation}\label{infinitybehaviorforpsi}
         \Psi(z;\tau)=\left(I + \frac{\Psi_1(\tau)}{z} + \frac{\Psi_2(\tau)}{z^2} + O\left(\frac{1}{z^3}\right)\right)z^{-\beta\sigma_3}e^{\frac{\tau z}{4}\sigma_3},
        \end{equation}
        where the branch cut of $z^\beta$ is taken along  the negative imaginary axis such that $\arg z \in (-\frac{\pi}{2}, \frac{3\pi}{2})$. Here, $\sigma_3$ is the third Pauli matrix $\begin{pmatrix} 1 & 0 \\ 0 & -1 \end{pmatrix}$.

  \item [(d)] When $z \to 0$ and  $z \in \Omega_i, \, i=1,3,4$, we have
  \begin{equation}\label{localbehaviorforpsinear0}
         \Psi(z;\tau)=\Psi^{(0)}(z;\tau)z^{\alpha \sigma_3} \begin{cases}
         \begin{pmatrix} 1 & (1-\gamma) \frac{\sin(\alpha + \beta)\pi}{\sin 2\alpha \pi}) \\ 0 & 1 \end{pmatrix}C_i^{(0)}, & \textrm{if } 2\alpha \notin \mathbb{N}, \vspace{5pt} \\
         \begin{pmatrix} 1 & \frac{(-1)^{2\alpha}(1-\gamma)}{\pi} \sin(\alpha + \beta)\pi \ln z \\ 0 & 1 \end{pmatrix}C_i^{(0)}, & \textrm{if } 2\alpha \in \mathbb{N},
         \end{cases}
        \end{equation}
        where both $z^\alpha$ and $\ln z$ take the principal branch,  and the constant matrices $C_i^{(0)}$ are given by
        \begin{equation}
         C_1^{(0)}=I, \qquad  C_3^{(0)}=J_6^{-1}J_4^{-1}, \qquad  C_4^{(0)}=J_6^{-1}
        \end{equation}
  with jump matrices $J_i$ given in \eqref{jumpforpsi}. Here, $\Psi^{(0)}(z;\tau)=\Psi^{(0)}_0(\tau)\left(I+\Psi^{(0)}_1(\tau)z+O(z^2)\right)$ is analytic at $z=0$.

  \item [(e)] When $z \to 1$ and $z \in \Omega_i, \, i=1,4,5$, we have
        \begin{equation}\label{localbehaviorforpsinear1}
         \Psi(z;\tau)=\Psi^{(1)}(z;\tau)
         \begin{pmatrix} 1 & -\frac{\gamma e^{\pi i (\alpha-\beta)}}{2\pi i}\ln(z-1) \\ 0 & 1 \end{pmatrix} C_i^{(1)},
        \end{equation}
        where $\ln (z-1)$ takes the principal branch, and the constant matrices $C_i^{(1)}$ are given by
        \begin{equation}
         C_1^{(1)}=I, \qquad  C_4^{(1)}=J_1^{-1}J_5^{-1}, \qquad  C_5^{(1)}=J_1^{-1}.
        \end{equation}
        Here, $\Psi^{(1)}(z;\tau)=\Psi^{(1)}_0(\tau)\left(I+\Psi^{(1)}_1(\tau)(z-1)+O(z-1)^2)\right)$ is analytic at $z=1$.

  \item [(f)] When $z \to -1$ and $z \in \Omega_i, \, i=1,2,3$, we have
        \begin{equation}\label{localbehaviorforpsinear-1}
         \Psi(z;\tau)=\Psi^{(-1)}(z;\tau)
         \begin{pmatrix} 1 & \frac{\gamma e^{-\pi i (\alpha-\beta)}}{2\pi i}\ln(z+1) \\ 0 & 1 \end{pmatrix} C_i^{(-1)},
        \end{equation}
        where the branch cut of $\ln (z+1)$ is taken along $(-1, \infty)$ such that $\arg (z+1) \in(0,2\pi)$, and the constant matrices $C_i^{(-1)}$ are given by
        \begin{equation}
         C_1^{(-1)}=I, \qquad  C_2^{(-1)}=J_2^{-1}, \qquad  C_3^{(-1)}=J_2^{-1}J_3^{-1}.
        \end{equation}
        Here, $\Psi^{(-1)}(z;\tau)=\Psi^{(-1)}_0(\tau)\left(I+\Psi^{(-1)}_1(\tau)(z+1)+O(z+1)^2)\right)$ is analytic at $z=-1$.
 \end{itemize}
\end{rhp}

\subsection{Lax pair}

As the jump matrices in \eqref{jumpforpsi} are independent of variables $z$ and $\tau$, one may derive a Lax pair from the above model RH problem. Comparing with the model RH problem in Xu and Zhao \cite{Xu:Zhao2020}, the additional deforming parameter $\gamma$ appears in the $(2,2)$-entry of the jump matrices $J_6$ and $J_7$ in \eqref{jumpforpsi}, as well as the endpoint conditions near $0$ and $\pm 1$ in \eqref{localbehaviorforpsinear0}, \eqref{localbehaviorforpsinear1} and \eqref{localbehaviorforpsinear-1}. A straightforward computation shows that this additional parameter does not change the associated Lax pair. Therefore, we reproduce the following results.

\begin{proposition}\label{laxpairprop}
 For the function $\Psi(z;\tau)$ satisfying RH problem \ref{modelrhp}, we have
 \begin{equation}\label{laxpair}
  \frac{\partial}{\partial z}\Psi(z;\tau)=L(z;\tau)\Psi(z;\tau), \qquad
  \frac{\partial}{\partial \tau}\Psi(z;\tau)=U(z;\tau)\Psi(z;\tau),
 \end{equation}
 where
 \begin{equation}\label{l}
  L(z;\tau)=\frac{\tau}{4}\sigma_3+\frac{A_0(\tau)}{z}+\frac{A_1(\tau)}{z-1}+\frac{A_2(\tau)}{z+1},
 \end{equation}
 and
 \begin{equation}\label{u}
  U(z;\tau)=\frac{z}{4}\sigma_3+\mathcal{B}(\tau).
 \end{equation}
 The above coefficient matrices are given by
\begin{equation}\label{a0phi0a0}
  A_0(\tau)=\begin{cases}
  \alpha \Psi_0^{(0)}(\tau) \sigma_3 \left(\Psi_0^{(0)}(\tau)\right)^{-1}, & \alpha \neq 0,  \vspace{2mm}\\
  \Psi_0^{(0)}(\tau) \begin{pmatrix} 0 & \frac{1-\gamma}{\pi} \sin \beta \pi \\ 0 & 0 \end{pmatrix} \left(\Psi_0^{(0)}(\tau)\right)^{-1}, & \alpha = 0,
  \end{cases}
 \end{equation}
 \begin{eqnarray}
  A_1(\tau) & = & -\frac{\gamma e^{\pi i (\alpha-\beta)}}{2\pi i} \Psi_0^{(1)}(\tau)
  \begin{pmatrix} 0 & 1 \\ 0 & 0 \end{pmatrix} \left(\Psi_0^{(1)}(\tau)\right)^{-1}, \label{a1intermsofpsi} \\
  A_2(\tau) & = & \frac{\gamma e^{-\pi i (\alpha-\beta)}}{2\pi i} \Psi_0^{(-1)}(\tau) \begin{pmatrix} 0 & 1 \\ 0 & 0 \end{pmatrix} \left(\Psi_0^{(-1)}(\tau)\right)^{-1}\label{a2intermsofpsi}
\end{eqnarray}
and 
\begin{equation} \label{b-Psi}
  \mathcal{B}(\tau)=  \frac{1}{2} \begin{pmatrix}
  0 & (\Psi_1(\tau))_{12} \\
  (\Psi_1(\tau))_{21} & 0 
  \end{pmatrix},
\end{equation}
where  $\Psi_1(\tau)$ is given in \eqref{infinitybehaviorforpsi},  $\Psi_0^{(k)}(\tau)$, $k = 0, \pm 1$ are the functions defined in \eqref{localbehaviorforpsinear0}, \eqref{localbehaviorforpsinear1} and  \eqref{localbehaviorforpsinear-1}, respectively. 

Moreover, the Hamiltonian $H(\tau; \alpha, \beta)$ in \eqref{sh}
 is related to $\Psi(z;\tau)$ as follows
 \begin{equation}\label{h}
  H(\tau; \alpha, \beta)=-\frac{1}{2}(\Psi_1(\tau))_{11}- \frac{\alpha^2-\beta^2}{\tau}.
 \end{equation}
\end{proposition}

\begin{proof}
 The proof is similar to that of Proposition 1 in \cite{Xu:Zhao2020}. 
\end{proof}

From \eqref{a0phi0a0}-\eqref{b-Psi}, we also have the following explicit expressions for the coefficient matrices
 \begin{small}
 \begin{equation}\label{a0}
  A_0(\tau)=
   \begin{pmatrix}
    u_1(\tau)v_1(\tau)+u_2(\tau)v_2(\tau)-\beta               & -(u_1(\tau)v_1(\tau)+u_2(\tau)v_2(\tau)-\alpha-\beta)y(\tau) \\
    (u_1(\tau)v_1(\tau)+u_2(\tau)v_2(\tau)+\alpha-\beta)/y(\tau) & -u_1(\tau)v_1(\tau)-u_2(\tau)v_2(\tau)+\beta
   \end{pmatrix},
 \end{equation}
 \end{small}
 \begin{equation}\label{a1}
  A_1(\tau)=
  \begin{pmatrix}
   -u_1(\tau)v_1(\tau)        & u_1(\tau)y(\tau)   \\
   -u_1(\tau)v_1^2(\tau)/y(\tau) & u_1(\tau)v_1(\tau)
  \end{pmatrix},
 \end{equation}
 \begin{equation}\label{a2}
  A_2(\tau)=
  \begin{pmatrix}
   -u_2(\tau)v_2(\tau)        & u_2(\tau)y(\tau)   \\
   -u_2(\tau)v_2^2(\tau)/y(\tau) & u_2(\tau)v_2(\tau)
  \end{pmatrix},
 \end{equation}
 and
 \begin{equation}\label{b}
  \mathcal{B}(\tau)=\frac{1}{\tau}
  \left(
  \begin{array}{ccc}
    0 \qquad  \qquad  \qquad (-u_1(\tau)(v_1(\tau)-1)-u_2(\tau)(v_2(\tau)-1)+\alpha+\beta)y(\tau) \\
    (-u_1(\tau)v_1(\tau)(v_1(\tau)-1)-u_2(\tau)v_2(\tau)(v_2(\tau)-1)+\alpha - \beta))/y(\tau) \qquad 0
   \end{array}
  \right).
 \end{equation}
 The above formulas, together with the compatibility condition of \eqref{laxpair}, give us the coupled Painlev\'e V system in \eqref{p5system}.

\begin{remark} \label{remark-lambda}
The above result shows that the associated coupled Painlev\'e V system is the same for $0 \leq \gamma <  1$ and $\gamma = 1$. Although the parameter $\gamma$ does not appear explicitly in the differential system \eqref{p5system}, properties of the solutions $u_k(\tau)$ and $v_k(\tau)$, $k = 1, 2$, depend on the parameter $\gamma$, such as the pole distribution and asymptotic behaviors. This is a general phenomenon for the Painlev\'e equations. For example, one may refer to \cite[Sec. 1.1]{Dai:Hu2017} for a discussion of the well-known homogeneous Painlev\'e II equation.
\end{remark}

Besides the Lax pair above, the following relations play an important role in our future analysis.
\begin{proposition}\label{ydprop}
 Let $\Psi_0^{(k)}(\tau)$, $k = 0, \pm 1$, be the functions defined in \eqref{localbehaviorforpsinear0}, \eqref{localbehaviorforpsinear1} and  \eqref{localbehaviorforpsinear-1}, respectively. Then, we have the following relations among $\Psi_0^{(k)}(\tau)$, the Hamiltonian $H(\tau)$ defined in \eqref{h}, and the function $y(\tau)$ appearing in the coefficient matrices \eqref{a0}-\eqref{a2}:
\begin{equation} \label{ypsi0}
  y(\tau)=\frac{\left(\Psi_0^{(0)}(\tau)\right)_{11}}{\left(\Psi_0^{(0)}(\tau)\right)_{21}},
\end{equation}
\begin{equation}\label{sH+ic}
  \frac{d}{d \tau}\left(\tau H(\tau)\right)=-\frac{\gamma }{2\pi i} \left( e^{\pi i (\alpha-\beta)}
  \frac{d}{d \tau}\left(\Psi_1^{(1)}(\tau)\right)_{21}+e^{-\pi i (\alpha-\beta)}
  \frac{d}{d \tau}\left(\Psi_1^{(-1)}(\tau)\right)_{21} \right).
 \end{equation}
In addition, let us define
\begin{equation} \label{dpsi0}
d(\tau) := 2\alpha\left(\Psi_0^{(0)}(\tau)\right)_{11}\left(\Psi_0^{(0)}(\tau)\right)_{21}.
\end{equation}
Then, the following differential relations hold
\begin{eqnarray}
\tau\frac{d}{d \tau}\ln y(\tau)&=&u_1(v_1-1)^2+u_2(v_2-1)^2+2\beta, \label{yuv} \\
\tau\frac{d}{d \tau}\ln d(\tau)&=&-u_1(v_1^2-1)-u_2(v_2^2-1)+2\alpha.  \label{duv}
\end{eqnarray}
\end{proposition}

\begin{proof}

It is obvious from the explicit expression of $A_0(\tau)$ in \eqref{a0} that
  \begin{equation}
  y(\tau)=\frac{(A_0(\tau))_{11}+\alpha}{(A_0(\tau))_{21}}.
 \end{equation}
Since $A_0(\tau)$ is related to $\Psi_0^{(0)}(\tau)$ via the equation \eqref{a0phi0a0}, we get \eqref{ypsi0} from the above formula.

To derive \eqref{sH+ic}, we first make use of the second equation in the Lax pair \eqref{laxpair} to obtain
\begin{eqnarray}
\frac{d}{d \tau}\Psi_0^{(0)}(\tau) & = & \mathcal{B}(\tau)\Psi_0^{(0)}(\tau), \label{psi0b} \\
  \frac{d}{d \tau}\Psi_1^{(k)}(\tau) & = & \frac{1}{4}\left(\Psi_0^{(k)}(\tau)\right)^{-1}\sigma_3\Psi_0^{(k)}(\tau), \qquad k=1,-1.
  \label{psi1intermsofpsi0}
\end{eqnarray}
Next, with the differential equations in \eqref{p5system} and the Hamiltonian system \eqref{dh}, a straightforward computation gives us
\begin{eqnarray}
  \frac{d}{d \tau}(\tau H(\tau)) &= & -\frac{1}{2}(u_1(\tau)v_1(\tau)-u_2(\tau)v_2(\tau)) \\
  & = & \frac{1}{2}(A_1(\tau))_{11} - \frac{1}{2} (A_2(\tau))_{11},
\end{eqnarray}
where the explicit expressions of $A_1$ and $A_2$ in \eqref{a1} and \eqref{a2} are also used in the last equation.
We note the following identities from \eqref{a1intermsofpsi}, \eqref{a2intermsofpsi} and \eqref{psi1intermsofpsi0}:
 \begin{equation}
  \begin{aligned}
   (A_1(\tau))_{11}=-\frac{\gamma e^{\pi i (\alpha-\beta)}}{\pi i}
   \frac{d}{d \tau}\left(\Psi_1^{(1)}(\tau)\right)_{21}, \quad
   (A_2(\tau))_{11}=\frac{\gamma e^{-\pi i (\alpha-\beta)}}{\pi i}
   \frac{d}{d \tau}\left(\Psi_1^{(-1)}(\tau)\right)_{21}.
  \end{aligned}
 \end{equation}
 Combining the above two formulas, we have \eqref{sH+ic}.

 The derivation of the differential system \eqref{yuv} and \eqref{duv} is the same as that of \cite[Eq. (2.28)\&(2.30)]{Xu:Zhao2020}. This completes the proof of the proposition.
\end{proof}

The following differential identities of the Hamiltonian play a crucial role in the future derivation of the large gap asymptotics, especially the constant term.
\begin{proposition} \label{prop: diff-iden}
 The  Hamiltonian $H=H(u_1,u_2,v_1,v_2,\tau;\alpha,\beta, \gamma)$ is  equivalent  to  the  classical action up to a total differential
 \begin{equation}\label{inth}
  H=\left( u_1\frac{dv_1}{d \tau}+u_2\frac{dv_2}{d \tau}-H\right)
  +\frac{d}{d \tau}(\tau H+\alpha \ln d(\tau)-\beta \ln y(\tau)-2(\alpha^2-\beta^2)\ln \tau).
 \end{equation}
The  Hamiltonian  system implies the following  differential formula for the classical action
 with respect to $\gamma$:
 \begin{equation}\label{differentialidentitylambda}
  \frac{\partial}{\partial \gamma}\left(u_1\frac{dv_1}{d \tau}+u_2\frac{dv_2}{d \tau}-H\right)=
  \frac{\partial}{\partial \tau}\left(u_1\frac{\partial v_1}{\partial \gamma}+u_2\frac{\partial v_2}{\partial \gamma}\right)
 \end{equation}
\end{proposition}

\begin{proof}
The equation \eqref{inth} is the same as that in \cite[Proposition 3]{Xu:Zhao2020}. The differential identity \eqref{differentialidentitylambda} is a straightforward computation with the aid of \eqref{dh}. This completes the proof.
\end{proof}

\subsection{Vanishing lemma}

As mentioned in Remark \ref{remark-lambda}, although the deforming parameter does not change the differential system and the Lax pair, the existence and asymptotics of solutions to the model RH problem is not trivial at all. We will first establish the existence of the solution via proving the vanishing lemma below, then investigate their detailed asymptotics in Sections \ref{Sec: large-s} and \ref{Sec: small-s}.

\begin{lemma}
 For $\alpha > -\frac{1}{2}$, $\beta \in i \mathbb{R}$ and $\tau \in -i(0, +\infty)$, let $\hat{\Psi}(z;\tau)$ be a function  satisfying the same problem as $\Psi(z;\tau)$ in RH problem \ref{modelrhp}, except that the behavior at infinity is altered to be
 \begin{equation}
  \hat{\Psi}(z;\tau)=O\left(\frac{1}{z}\right)z^{-\beta \sigma_3}e^{\frac{\tau z}{4}\sigma_3}, \qquad \textrm{as } z \to \infty.
 \end{equation}
 Then, we have $\hat{\Psi}(z;\tau)\equiv 0$.
\end{lemma}
\begin{proof}
We first normalize the behavior of $\hat{\Psi}(z;\tau)$ at infinity and reduce the jump contour to the real line. For this purpose, we introduce a transformation as follows:
\begin{equation}
  \hat{\Psi}^{(1)}(z;\tau)=\hat{\Psi}(z;\tau)e^{-\frac{1}{2}\pi i \beta \sigma_3}\varphi(z)^{\beta \sigma_3}e^{-\frac{\tau z}{4}\sigma_3}
  \left\{
  \begin{aligned}
    & \begin{pmatrix} 1 & 0 \\ e^{-\pi i\alpha}\varphi(z)^{2\beta}e^{-\frac{1}{2}\tau z} & 1 \end{pmatrix}, & \Im z>0, \,  z \in \Omega_5, \\
    & \begin{pmatrix} 1 & e^{\pi i\alpha}\varphi(z)^{-2\beta}e^{\frac{1}{2}\tau z} \\ 0 & 1 \end{pmatrix}, & \Im z<0, \,  z \in \Omega_5, \\
    & \begin{pmatrix} 1 & 0 \\ e^{\pi i(\alpha-2\beta)}\varphi(z)^{2\beta}e^{-\frac{1}{2}\tau z} & 1 \end{pmatrix}, & \Im z>0, \,  z \in \Omega_2, \\
    & \begin{pmatrix} 1 & e^{-\pi i(\alpha-2\beta)}\varphi(z)^{-2\beta}e^{\frac{1}{2}\tau z} \\ 0 & 1 \end{pmatrix}, & \Im z<0, \,  z \in \Omega_2, \\
    & I,                          & \mathrm{otherwise},
  \end{aligned}
  \right.
 \end{equation}
where $\varphi(z)$ is defined as
 \begin{equation}
  \varphi(z):=z+\sqrt{z^2-1}, \qquad z \in \mathbb{C} \setminus [-1,1],
 \end{equation}
and $\sqrt{z^2-1}$ takes the principal branch. For $x\in(-1,1)$,  we have $\varphi_+(x)\varphi_-(x)=1$. Note that the branch of $z^\beta$ is taken as in \eqref{infinitybehaviorforpsi} such that $\arg z \in (-\frac{\pi}{2}, \frac{3\pi}{2})$.
Then, we have a RH problem for $\hat{\Psi}^{(1)}(z;\tau)$ as follows.

 \begin{rhp}
  \hfill
  \begin{itemize}
   \item [(a)] $\hat{\Psi}^{(1)}(z) = \hat{\Psi}^{(1)}(z;\tau)$ is analytic for $z \in \mathbb{C} \setminus \mathbb{R}$.
   \item [(b)] For $x \in \mathbb{R}$, we have
         \begin{equation} \label{jumpforpsi1}
          \hat{\Psi}_+^{(1)}(x)=\hat{\Psi}_-^{(1)}(x)\hat{J}(x)
         \end{equation}
         with
         \begin{equation}
          \hat{J}(x)=
          \begin{cases}
            \begin{pmatrix} 0 & -e^{\pi i\alpha}\varphi(x)^{-2\beta}e^{\frac{1}{2}\tau x} \\ e^{-\pi i\alpha}\varphi(x)^{2\beta}e^{-\frac{1}{2}\tau x} & 1 \end{pmatrix}, & x>1,    \\
            \begin{pmatrix} 0 & -e^{\pi i\alpha}e^{\frac{1}{2}\tau x} \\ e^{-\pi i\alpha}e^{-\frac{1}{2}\tau x} & (1-\gamma)\left(\frac{\varphi_+(x)}{\varphi_-(x)}\right)^{-\beta} \end{pmatrix}, & 0<x<1,  \\
            \begin{pmatrix} 0 & -e^{-\pi i\alpha}e^{\frac{1}{2}\tau x} \\ e^{\pi i\alpha}e^{-\frac{1}{2}\tau x} & (1-\gamma)\left(\frac{\varphi_+(x)}{\varphi_-(x)}\right)^{-\beta} \end{pmatrix}, & -1<x<0, \\
            \begin{pmatrix} 0 & -e^{-\pi i\alpha}|\varphi(x)|^{-2\beta}e^{\frac{1}{2}\tau x} \\ e^{\pi i\alpha}|\varphi(x)|^{2\beta}e^{-\frac{1}{2}\tau x} & 1 \end{pmatrix}, & x<-1,   \\
          \end{cases}
         \end{equation}
         where  $\alpha > -\frac{1}{2}$, $\beta \in i \mathbb{R}$ and $\tau \in -i(0, +\infty)$.

   \item [(c)] When $z \to \infty$, we have
         \begin{equation}\label{infinitybehaviorforpsi1}
          \hat{\Psi}^{(1)}(z)=O\left(\frac{1}{z}\right).
         \end{equation}

   \item [(d)] When $z \to 0$ and $\Im z > 0$,   we have
         \begin{equation}\label{localbehaviorforpsi1near0}
          \hat{\Psi}^{(1)}(z)=O\left(1\right) z^{\alpha \sigma_3}
          \begin{cases}
          \begin{pmatrix} 1 & (1-\gamma) \frac{\sin(\alpha + \beta)\pi}{\sin 2\alpha \pi}) \\ 0 & 1 \end{pmatrix}, & 2\alpha \notin \mathbb{N}, \\
          \begin{pmatrix} 1 & \frac{(-1)^{2\alpha}(1-\gamma)}{\pi} \sin(\alpha + \beta)\pi \ln z \\ 0 & 1 \end{pmatrix}, & 2\alpha \in \mathbb{N}.
          \end{cases}
         \end{equation}
The behavior for $\Im z < 0$ can be obtained via the jump condition \eqref{jumpforpsi1}.
   \item [(e)] When $z \to \pm1$, we have
         \begin{equation}\label{localbehaviorforpsi1near+-1}
          \hat{\Psi}^{(1)}(z)=O\left(\ln (z\pm 1)\right).
         \end{equation}
  \end{itemize}
 \end{rhp}

Next, let us introduce
 \begin{equation}
  Q(z):=\hat{\Psi}^{(1)}(z)\left(\hat{\Psi}^{(1)}(\bar{z})\right)^{*}, \qquad z \in \mathbb{C} \setminus \mathbb{R},
 \end{equation}
 where $(\cdot)^{*}$ represents the Hermitian conjugate. Consider the following integral
 \begin{equation}\label{q+q*}
  \int_{\mathbb{R}} \left[ Q_+(x)+(Q_+(x))^* \right] dx=
  \int_{\mathbb{R}}\left[ \hat{\Psi}_{+}^{(1)}(x)\left(\left(\hat{J}^{-1}(x)\right)^{*}+\hat{J}^{-1}(x)\right)\left(\hat{\Psi}_{+}^{(1)}(x)\right)^{*} \right] dx.
 \end{equation}
Since  $\varphi_+(x)=\overline{\varphi_-(x)} = e^{i \theta(x)}$ for $x \in (-1,1)$ with $\theta (x)=\arg \varphi_+(x)$, then the second term in the above integrand reduces to
 \begin{equation}\label{j+j*}
  \left(\hat{J}^{-1}(x)\right)^{*}+\hat{J}^{-1}(x)=
  \left\{
  \begin{aligned}
    & \begin{pmatrix} 2 & 0 \\ 0 & 0 \end{pmatrix}, & x<-1, \ x>1, \\
    & \begin{pmatrix} 2(1-\gamma)e^{-2\theta(x)i \beta} & 0 \\ 0 & 0 \end{pmatrix}, & -1<x<1.      \\
  \end{aligned}
  \right.
 \end{equation}
We will show that the integral in \eqref{q+q*} actually equals 0. Let us first consider the possible singularity at $x =\pm 1 $ and $x = 0$. It is easy to see from \eqref{localbehaviorforpsi1near+-1} that the integral is convergent at $x = \pm 1$. Near $x = 0$, as $\theta(0)$ is bounded, it is readily seen from  \eqref{localbehaviorforpsi1near0} and the above formula that, as $x \to 0$,
 \begin{equation*}
 \hat{\Psi}_{+}^{(1)}(x)\left(\left(\hat{J}^{-1}(x)\right)^{*}+\hat{J}^{-1}(x)\right)\left(\hat{\Psi}_{+}^{(1)}(x)\right)^{*}  =
 \begin{cases}
 O(x^{2\alpha}), & 2\alpha \notin \mathbb{N}, \\
 O(x^{2\alpha} \ln^2 x), & 2\alpha \in \mathbb{N},
 \end{cases}
 \end{equation*}
 Since $\alpha > -\frac{1}{2}$, the above estimation implies that the integral \eqref{q+q*} converges at $x = 0$ as well. Last, due to the large-$z$ behavior in \eqref{infinitybehaviorforpsi1}, the integrand is $O(1/x^2)$ as $x \to \pm \infty$. Therefore, we have
 \begin{equation}
  \int_{\mathbb{R}} \left[ \hat{\Psi}_{+}^{(1)}(x)\left(\left(\hat{J}^{-1}(x)\right)^{*}+\hat{J}^{-1}(x)\right)\left(\hat{\Psi}_{+}^{(1)}(x)\right)^{*} \right] dx=0.
 \end{equation}

  Next, we substitute \eqref{j+j*} into the above formula to get
{\small
\begin{align*}
    &  2\int_{1}^{\infty}
     \begin{pmatrix} \left|\hat{\Psi}_{11,+}^{(1)}(x)\right|^2 &
     \hat{\Psi}_{11,+}^{(1)}(x) \overline{\hat{\Psi}_{21,+}^{(1)}(x)} \\ \overline{\hat{\Psi}_{11,+}^{(1)}(x)} \hat{\Psi}_{21,+}^{(1)}(x) & \left|\hat{\Psi}_{21,+}^{(1)}(x)\right|^2\end{pmatrix} dx
    +
     2\int_{-\infty}^{-1}
     \begin{pmatrix} \left|\hat{\Psi}_{11,+}^{(1)}(x)\right|^2
     & \hat{\Psi}_{11,+}^{(1)}(x) \overline{\hat{\Psi}_{21,+}^{(1)}(x)} \\
     \overline{\hat{\Psi}_{11,+}^{(1)}(x)} \hat{\Psi}_{21,+}^{(1)}(x) &
     \left|\hat{\Psi}_{21,+}^{(1)}(x)\right|^2\end{pmatrix} dx    \\
    & \qquad \qquad
     +2(1-\gamma)\int_{-1}^{1} e^{-2\theta(x)i \beta}
     \begin{pmatrix} \left|\hat{\Psi}_{11,+}^{(1)}(x)\right|^2 &
     \hat{\Psi}_{11,+}^{(1)}(x) \overline{\hat{\Psi}_{21,+}^{(1)}(x)}  \\ \overline{\hat{\Psi}_{11,+}^{(1)}(x)} \hat{\Psi}_{21,+}^{(1)}(x) & \left|\hat{\Psi}_{21,+}^{(1)}(x)\right|^2\end{pmatrix} dx = 0.
\end{align*}}
Recall that $\beta \in i \mathbb{R}$, and the factor $e^{-2\theta(x)i \beta}$ is always positive. Then, we have from the above formula that $\hat{\Psi}_{11,+}^{(1)}(x)  = \hat{\Psi}_{21,+}^{(1)}(x) = 0$ for $x \in \mathbb{R}$. As a consequence, we conclude that the first column $\hat{\Psi}_{11}^{(1)}(z)=\hat{\Psi}_{21}^{(1)}(z)=0$ for $\Im z >0$ by analytic continuation.
With the jump conditions \eqref{jumpforpsi1}, it is easy to see that entries of the second column $\hat{\Psi}_{12}^{(1)}(z)$ and $\hat{\Psi}_{22}^{(1)}(z)$ also vanish for $\Im z <0$ via a similar argument.

To show that the other entries also vanish, we introduce the following auxiliary scalar functions
 \begin{equation}
  g_k(z)= \left\{
  \begin{aligned}
    & -\left(\hat{\Psi}^{(1)}(z)\right)_{k2}e^{-\frac{\tau z}{4}}\varphi(z)^{2\beta}e^{-\pi i \alpha}, \qquad & \Im z>0, \\
    & \left(\hat{\Psi}^{(1)}(z)\right)_{k1}e^{\frac{\tau z}{4}}, \qquad  & \Im z<0,
  \end{aligned}
  \right.
 \end{equation}
for $k=1,2$. One can see from  \eqref{jumpforpsi1} that $g_k(z)$ is analytic for $z \notin (-\infty, 1]$. Moreover, based on the large-$z$ behavior \eqref{infinitybehaviorforpsi1}, we have $g_k(z)=O(e^{-\frac{|\tau z|}{4}})$  for pure imaginary $z$. According to Carlson's theorem, we conclude that $g_k(z)=0$ for $k =1,2$.
This guarantees that the first and second column of $\hat{\Psi}^{(1)}(z)$ vanishes for $\Im z < 0$ and $\Im z > 0 $, respectively.

This finishes the proof of the vanishing lemma.
\end{proof}

By a standard analysis \cite{dkmvz1999,Fokas:Mu:Zhou1992,Fokas:Zhou1992}, the following proposition is an immediate consequence of the vanishing lemma, which gives us the existence of the unique solution for the model RH problem for $\Psi(z;\tau)$ and the associated Hamiltonian $H(\tau)$ in \eqref{h}.

\begin{proposition} \label{Prop:Psi-exist}
For $\alpha > -\frac{1}{2}$, $\beta \in i \mathbb{R}$ and $\tau \in -i(0, +\infty)$, there exists a unique solution to the RH problem \ref{modelrhp} for $\Psi(z;\tau)$. Moreover, the Hamiltonian $H(\tau)$ in \eqref{h} is pole-free for $\tau \in -i(0, +\infty).$
\end{proposition}

\section{Asymptotic analysis of the RH problem for $Y(z)$}\label{analysisofy}

Now, we are ready to perform a Deift-Zhou nonlinear steepest descent analysis to derive asymptotics of $Y(z)$ as $n \to \infty$ and  $nt$ is bounded.  The first three steps in the analysis are the same as those in \cite[Sec. 4]{Dei:Its:Kra2011}, while the fourth step involves a local parametrix construction in terms of the function $\Psi(z;\tau)$ given in the previous section.

\subsection{Normalization}
To normalize the large-$z$ asymptotics of $Y$ in \eqref{eq: Y-large-z}, we simply introduce the first transformation as follows:
\begin{equation}\label{t}
 T(z)=
 \begin{cases}
   Y(z)z^{-n\sigma_3}, \qquad & |z|>1, \\
   Y(z), \qquad               & |z|<1.
 \end{cases}
\end{equation}
It is straightforward to check that $T(z)$ satisfies the following RH problem.
\begin{rhp}
 \quad 
 \begin{itemize}
  \item [(a)] $T(z)$ is analytic in $\mathbb{C} \setminus C$.
  \item [(b)] $T(z)$ satisfies the jump condition
        \begin{equation}\label{jumpfort}
         T_+(z)=T_-(z)
         \begin{pmatrix}
          z^n & w(z;t) \\ 0 & z^{-n}
         \end{pmatrix}, \qquad z \in  C.
        \end{equation}
  \item [(c)] As $z \to \infty$, we have
        \begin{equation}
         T(z)=I+O\left(\frac{1}{z}\right).
        \end{equation}
  \item [(d)] As $z \to z_i, \ i=1,2$, we have
        \begin{equation}
         T(z)=
         \begin{pmatrix}
          O(1) & O(\ln (z-z_i)) \\
          O(1) & O(\ln (z-z_i))
         \end{pmatrix}.
        \end{equation}
 \end{itemize}
\end{rhp}

\subsection{Opening of the lens}
For $z = e^{i\theta} \in C$, the diagonal entries in \eqref{jumpfort} are highly oscillatory as $n \to \infty$, which is not suitable for us to derive the large-$n$ asymptotics. A standard approach is to conduct contour deformation in the complex plane, such that the oscillatory jump matrix is transformed to new ones, which are either $n$-independent or tending to the identity matrix as $n \to \infty$. In our case, based on the factorization
\begin{equation}
 \begin{pmatrix}
  z^n & w(z) \\ 0 & z^{-n}
 \end{pmatrix}=
 \begin{pmatrix}
  1 & 0 \\ z^{-n}w(z)^{-1} & 1
 \end{pmatrix}
 \begin{pmatrix}
  0 & w(z) \\ -w(z)^{-1} & 0
 \end{pmatrix}
 \begin{pmatrix}
  1 & 0 \\ z^{n}w(z)^{-1} & 1
 \end{pmatrix},
\end{equation}
we define the second transformation as
\begin{equation}\label{s}
 S(z)=
 \begin{cases}
   T(z)\begin{pmatrix}  1 & 0 \\ z^{-n}w(z)^{-1} & 1  \end{pmatrix}, \qquad & z \in \Omega_E,     \\
   T(z)\begin{pmatrix}  1 & 0 \\ -z^{n}w(z)^{-1} & 1  \end{pmatrix}, \qquad & z \in \Omega_I,     \\
   T(z),  & \mathrm{otherwise},
 \end{cases}
\end{equation}
where $\Omega_E$ and $\Omega_I$ denote  regions outside and inside unit circle, respectively; see Figure \ref{figure10} for an illustration. Then, $S(z)$ satisfies the following RH problem.

\begin{figure}
 \centering
 \includegraphics[width=3.8in]{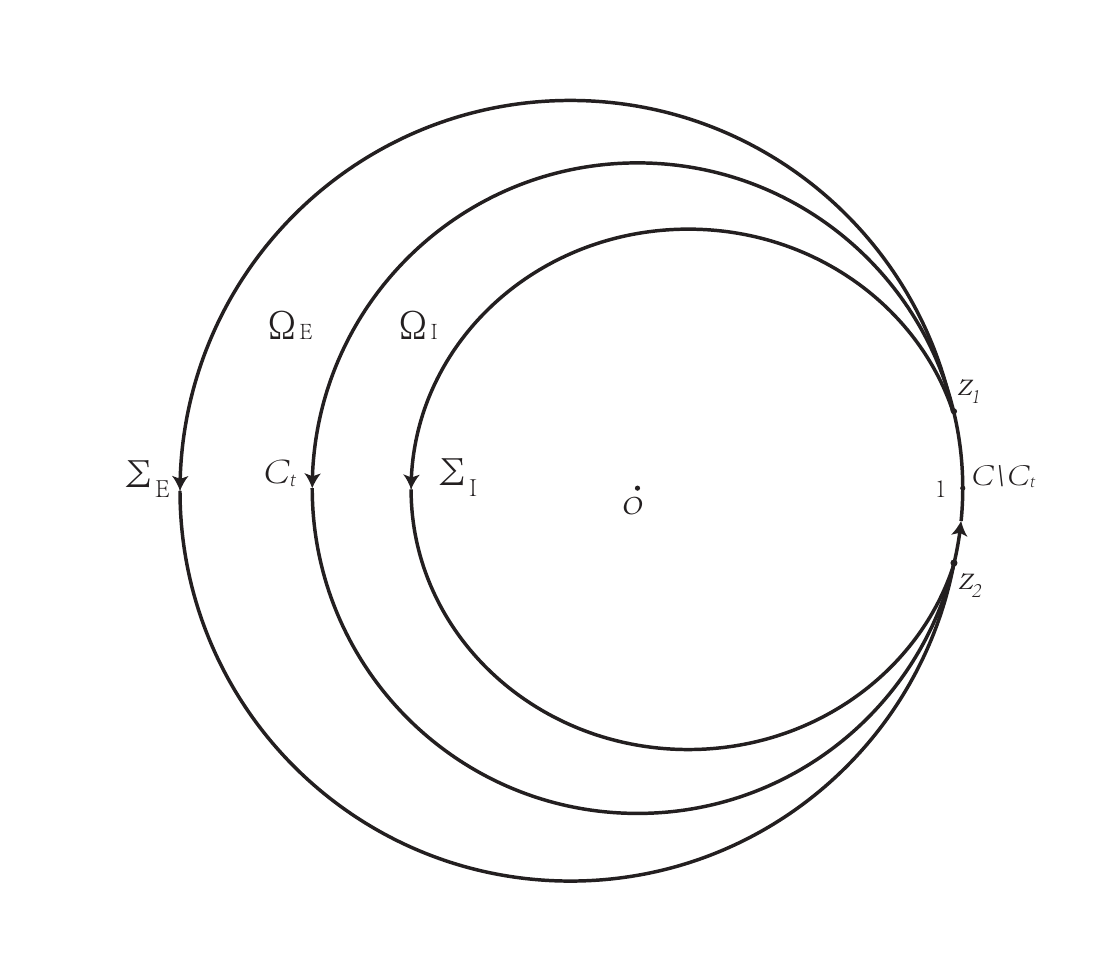}
 \caption{Contours for the RH problem for $S$ and the regions $\Omega_E$ and $\Omega_I$}
 \label{figure10}
\end{figure}

\begin{rhp}
 \quad 
 \begin{itemize}
  \item [(a)] $S(z)$ is analytic in $\mathbb{C} \setminus \{ C \cup \Sigma_E \cup \Sigma_I \} $; see Figure \ref{figure10}.
  \item [(b)] $S(z)$ satisfies the jump condition
        \begin{equation}
         S_+(z)=S_-(z)J_S(z), \qquad z \in C \cup \Sigma_E \cup \Sigma_I,
        \end{equation}
        where
        \begin{equation}\label{jumpfors}
         J_S(z)=\begin{cases}
         \begin{pmatrix}  1 & 0 \\ z^{-n}\tilde{w}(z)^{-1} & 1 \end{pmatrix}, \qquad & z \in \Sigma_E,    \\
         \begin{pmatrix}  1 & 0 \\ z^{n}\tilde{w}(z)^{-1} & 1 \end{pmatrix}, \qquad & z \in \Sigma_I,   \\
         \begin{pmatrix}  0 & \tilde{w}(z) \\ -\tilde{w}(z)^{-1} & 0 \end{pmatrix}, \qquad & z \in C_t,   \\
         \begin{pmatrix}  z^n & (1-\gamma)\tilde{w}(z) \\ 0 & z^{-n} \end{pmatrix}, \qquad & z \in C \setminus C_t,
         \end{cases}
        \end{equation}
        and $\tilde{w}(z)$ is defined in \eqref{def: w-tilde}. For convenience, we denote the jump matrices on $\Sigma_E, \Sigma_I$ as $J_{S,E}, J_{S,I}$, respectively.
  \item [(c)] As $z \to \infty$, we have
        \begin{equation}
         S(z)=I+O\left(\frac{1}{z}\right).
        \end{equation}
  \item [(d)] As $z \to z_i, \ i=1,2$, we have
        \begin{equation}
         S(z)=
         \begin{pmatrix}
          O(1) & O(\ln (z-z_i)) \\
          O(1) & O(\ln (z-z_i))
         \end{pmatrix}.
        \end{equation}
 \end{itemize}
\end{rhp}

\subsection{Global parametrix}

As $|z| > 1$ on $\Sigma_E$ and $|z| < 1$ on $\Sigma_I$, it is obvious that both $J_{S,E}(z)$ and $J_{S,I}(z) $ tend to the identity matrix when $n \to \infty$. While the endpoints $z_1 = e^{it}$ and $z_2=e^{i(2\pi -t)}$ shrink to 1 when $nt$ is bounded, we are led to a global parametrix for $P^{(\infty)}(z)$ which only possesses jump on the unit circle.

\begin{rhp}
 \quad 
 \begin{itemize}
  \item [(a)] $P^{(\infty)}(z)$ is analytic in $\mathbb{C} \setminus C$.
  \item [(b)] $P^{(\infty)}(z)$ satisfies the jump condition
        \begin{equation}
         P^{(\infty)}_+(z)=P^{(\infty)}_-(z)
         \begin{pmatrix}  0 & \tilde{w}(z) \\ -\tilde{w}(z)^{-1} & 0 \end{pmatrix}, \qquad z \in C,
        \end{equation}
        with $\tilde{w}(z)$ given in \eqref{def: w-tilde}.
  \item [(c)] As $z \to \infty$, we have
        \begin{equation}
         P^{(\infty)}(z)=I+O\left(\frac{1}{z}\right).
        \end{equation}
 \end{itemize}
\end{rhp}

The above RH problem can be solved explicitly with the solution given by (cf. \cite[Eq. (6.12)]{Xu:Zhao2020})
\begin{equation}\label{Pinfinity}
 P^{(\infty)}(z)
 =\left\{
 \begin{aligned}
    & \left(\frac{z-1}{z}\right)^{(\beta - \alpha)\sigma_3}, \qquad   & |z|>1, \\
    & (z-1)^{(\beta+\alpha)\sigma_3}e^{-\pi i (\alpha+\beta)\sigma_3}
  \begin{pmatrix}0 & 1 \\ -1 & 0\end{pmatrix}, \qquad & |z|<1,
 \end{aligned}
 \right.
\end{equation}
where the branch of $z^{\alpha-\beta}$ is taken along $[0,+\infty)$ such that $\arg z \in (0,2\pi)$, and the branches of $(z-1)^{\alpha+\beta}$ and $(z-1)^{\beta-\alpha}$ are taken along $[1,+\infty)$ such that $\arg (z-1) \in (0,2\pi)$.

\subsection{Local parametrix}

Since the jump of $S(z)$ on $ C \setminus C_t$ has been ignored in the global parametrix construction, $P^{(\infty)}(z)$ clearly cannot approximate $S(z)$ when $z$ is close to the endpoints $z_1$ and $z_2$. Note that both $z_1$ and $z_2$ belong to a neighborhood $U(1,\delta) =\{ z: |z-1|<\delta \}$ of $z = 1$, with $\delta$ small but fixed. Therefore, we look for a function $P^{(1)}(z)$ satisfying the following RH problem in $U(1,\delta)$.

\begin{rhp}\label{rhpforP1}
 \quad 
 \begin{itemize}
  \item [(a)] $P^{(1)}(z)$ is analytic in $U(1,\delta) \setminus \{ C \cup \Sigma_E \cup \Sigma_I \}$; see Figure \ref{figure10} for the contours.

  \item [(b)] $P^{(1)}(z)$ satisfies the jump condition
   \begin{equation}
     P_+^{(1)}(z)=P_-^{(1)}(z) J_S(z), \qquad z \in U(1,\delta) \cap \{ C \cup \Sigma_E \cup \Sigma_I \},
    \end{equation}
    where $J_S(z)$ is given in \eqref{jumpfors}.
   \item [(c)] As $n \to \infty$, we have the matching condition
    \begin{equation}\label{matchconditionforP1}
      P^{(1)}(z)=\left(I+O\left(\frac{1}{n}\right)\right)P^{(\infty)}(z), \qquad  z \in \partial U(1,\delta).
    \end{equation}
 \end{itemize}
\end{rhp}

We solve the above RH problem explicitly in terms of the model RH problem \ref{modelrhp} for $\Psi$. First, let us introduce the following conformal mapping near  $z=1$:
\begin{equation}\label{conformalmapping}
 \zeta(z):=\frac{e^{-\frac{\pi i}{2}}}{t}\ln z, \qquad z \in U(1,\delta),
\end{equation}
where the principal branch of $\ln z$ is taken. Clearly, we have $\zeta(z_1)=1$ and $\zeta(z_2)=-1$.

\begin{lemma}
Let $\Psi(z;\tau)$ be the solution to RH problem \ref{modelrhp} with properties specified in Proposition \ref{Prop:Psi-exist}. Then, the solution to  RH problem \ref{rhpforP1} is given by
 \begin{equation}\label{P1}
  P^{(1)}(z)=E(z)\Psi(\zeta(z);-2int)z^{\frac{n}{2}\sigma_3}
  \begin{cases}
   e^{\frac{1}{2}(\alpha-\beta) \pi i\sigma_3}\tilde{w}(z)^{\frac{\sigma_3}{2}}
   \begin{pmatrix} 0 & -1 \\ 1 & 0 \end{pmatrix},
   &  |z|>1, \Im z>0,    \\
   e^{\frac{1}{2}(\alpha-\beta) \pi i\sigma_3}\tilde{w}(z)^{-\frac{\sigma_3}{2}},
   &   |z|<1, \Im z>0,    \\
   e^{-\frac{1}{2}(\alpha-\beta) \pi i\sigma_3}\tilde{w}(z)^{-\frac{\sigma_3}{2}},
   &   |z|<1, \Im z<0,    \\
   e^{-\frac{1}{2}(\alpha-\beta) \pi i\sigma_3}\tilde{w}(z)^{\frac{\sigma_3}{2}}
   \begin{pmatrix} 0 & -1 \\ 1 & 0 \end{pmatrix},
   &  |z|>1, \Im z<0,
  \end{cases}
 \end{equation}
where the principal branch of $z^{\frac{n}{2}}$ is taken. The prefactor $E(z)$ is an analytic function in $U(1,\delta)$ and defined as
 \begin{equation}\label{E}
  E(z)=z^{\frac{\alpha-\beta}{2}\sigma_3}\left(\frac{\zeta(z)}{z-1}\right)^{-\beta \sigma_3} \begin{cases}
    e^{-\beta \pi i  \sigma_3}\begin{pmatrix} 0 & 1 \\ -1 & 0 \end{pmatrix},  & \qquad \Im z>0, \\
    e^{-\alpha \pi i  \sigma_3}\begin{pmatrix} 0 & 1 \\ -1 & 0 \end{pmatrix}, & \qquad \Im z<0,
  \end{cases}
 \end{equation}
where the branch of $z^{\frac{\alpha-\beta}{2}}$ is taken along $[0,\infty)$ such that $\arg z \in (0, 2\pi)$, and the principal branch of $(\cdot)^{-\beta}$ is taken.
\end{lemma}

\begin{proof}
The analyticity of $E(z)$ has been proved in \cite[Proposition 6]{Xu:Zhao2020}. Next, let us verify the jump condition of $P^{(1)}(z)$. For $z \in C \setminus C_t$ and $\Im z > 0$, it follows from in \eqref{conformalmapping} that $0<\zeta(z) < 1$. Note that $\tilde{w}(z)$ is analytic on $C$; see its definition in \eqref{def: w-tilde}. Then, we have from \eqref{jumpforpsi} and \eqref{P1} that
 \begin{eqnarray*}
    && \left(P_-^{(1)}(z)\right)^{-1}P_+^{(1)}(z)  \\
    & & =  \begin{pmatrix} 0 & 1 \\ -1 & 0 \end{pmatrix}\tilde{w}(z)^{-\frac{\sigma_3}{2}}e^{-\frac{1}{2}(\alpha-\beta) \pi i\sigma_3}z^{-\frac{n}{2}\sigma_3}
    \begin{pmatrix} 0 & -e^{\pi i(\alpha-\beta)} \\ e^{-\pi i(\alpha-\beta)} & 1-\gamma \end{pmatrix}
    z^{\frac{n}{2}\sigma_3}
    e^{\frac{1}{2}(\alpha-\beta) \pi i\sigma_3}\tilde{w}(z)^{\frac{\sigma_3}{2}}
     \\
    &  & =  \begin{pmatrix}  z^n & (1-\gamma) \tilde{w}(z) \\ 0 & z^{-n} \end{pmatrix},
 \end{eqnarray*}
 which matches the jump of $S(z)$ for $z \in C \setminus C_t$ in \eqref{jumpfors}. When $z$ belongs to the remaining contours of $U(1,\delta) \cap \{\Sigma_E \cup \Sigma_I \cup C\}$, the jump condition can be verified in a similar way.

Finally, let us verify the matching condition \eqref{matchconditionforP1}. As $nt$ is bounded when $n \to \infty$, one can see $\zeta(z) = O(n)$ for $z \in \partial U(1,\delta)$. When $|z| > 1$ and $\Im z > 0$, it follows from \eqref{infinitybehaviorforpsi}, \eqref{P1} and \eqref{E} that
 \begin{equation}
    P^{(1)}(z)\left(P^{(\infty)}(z)\right)^{-1}
  = E(z)\left(I+O\left(\frac{1}{n}\right)\right)\begin{pmatrix} 0 & -1 \\ 1 & 0 \end{pmatrix}\zeta^{\beta \sigma_3}e^{-\frac{\pi i}{2}(\alpha - \beta)\sigma_3}\tilde{w}(z)^{-\frac{\sigma_3}{2}}
  \left(P^{(\infty)}(z)\right)^{-1}.
\end{equation}
With the definition of $P^{(\infty)}(z)$ in \eqref{Pinfinity}, we have
 \begin{equation}
   \begin{pmatrix} 0 & -1 \\ 1 & 0 \end{pmatrix}\zeta^{\beta \sigma_3}e^{-\frac{\pi i}{2}(\alpha - \beta)\sigma_3}\tilde{w}(z)^{-\frac{\sigma_3}{2}}
   \left(P^{(\infty)}(z)\right)^{-1}
   = E^{-1}(z).
 \end{equation}
Although $E(z)$ depends on $t$, one can see $E(z)=O(1)$ as $n \to \infty$ and $nt$ bounded since $\beta$ is a pure imaginary parameter. This justifies the matching condition \eqref{matchconditionforP1} when $|z| > 1$ and $\Im z > 0$. For $z$ in other sectors, the verification is similar. This completes our proof.
\end{proof}

\subsection{Final transformation}

Let us define the final transformation as
\begin{equation}\label{finalr}
 R(z)=
  \begin{cases}
   S(z)\left(P^{(\infty)}(z)\right)^{-1}, & \qquad  |z-1|>\delta,  \\
   S(z)\left(P^{(1)}(z)\right)^{-1},      & \qquad  |z-1|<\delta.
  \end{cases}
\end{equation}
It is easy to see that $R(z)$ satisfies the following RH problem.
\begin{rhp}\label{rhpforR}
 \quad 
 \begin{itemize}
  \item [(a)] $R(z)$ is analytic in $\mathbb{C} \setminus \Sigma_R$, where $\Sigma_R$ is illustrated in Figure \ref{figure12}.

  \begin{figure}
   \centering
   \includegraphics[width=3in]{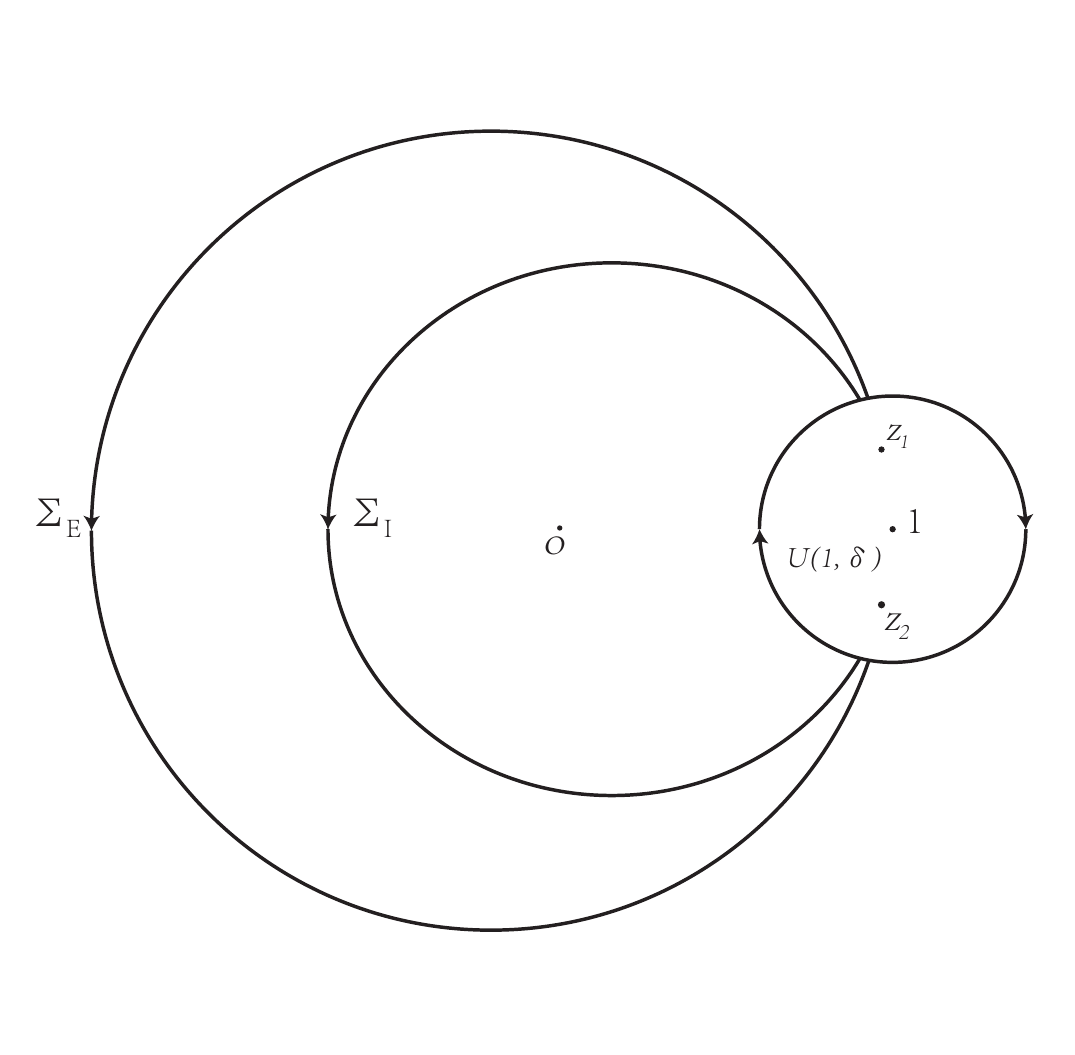}
   \caption{Contours for the RH problem for $R(z)$}
   \label{figure12}
  \end{figure}

  \item [(b)] $R(z)$ satisfies the jump condition
        \begin{equation}\label{jumpforR}
         R_+(z)=R_-(z)J_R(z), \qquad z \in \Sigma_R,
        \end{equation}
        where
        \begin{equation} \label{eq: JR-form}
          J_R(z)=
          \begin{cases}
            P^{(1)}(z)\left(P^{(\infty)}(z)\right)^{-1}, \qquad  & |z-1|=\delta, \\
            P^{(\infty)}(z)J_{S,E}(z)\left(P^{(\infty)}(z)\right)^{-1}, \qquad & z \in \Sigma_E, |z-1|>\delta,   \\
            P^{(\infty)}(z)J_{S,I}(z)\left(P^{(\infty)}(z)\right)^{-1}, \qquad & z \in \Sigma_I, |z-1|>\delta.
         \end{cases}
        \end{equation}
  \item [(c)] As $z \to \infty$, we have
        \begin{equation}
         R(z)=I+O\left(\frac{1}{z}\right).
        \end{equation}
 \end{itemize}
\end{rhp}

For $z \in \Sigma_{E,I}$ and $|z-1| > \delta$, one can readily see from \eqref{jumpfors}, \eqref{Pinfinity} and \eqref{eq: JR-form} that, there exists a positive constant $c > 0$ such that
\begin{equation}
 J_R(z)= I + O \left( e^{-cn}\right), \qquad n \to \infty.
\end{equation}
For $|z-1| = \delta$, we have from the matching condition \eqref{matchconditionforP1} that
\begin{equation}\label{JR}
 J_R(z)=I + O \left(\frac{1}{n}\right).
\end{equation}
Therefore, RH problem \ref{rhpforR} is a small norm problem when $n$ is large. By a standard argument for RH problems \cite{Deift1999book},
we conclude that $R$ exists for sufficiently large $n$ and has the following asymptotics
\begin{equation}\label{Rinfinity}
 R(z)=I+O\left(\frac{1}{n}\right), \qquad \frac{d}{dz}R(z)=O\left(\frac{1}{n}\right),
\end{equation}
as $n \to \infty$, uniformly for $z \in \mathbb{C} \setminus \Sigma_R $.

This finishes the steepest descent analysis for RH problem \ref{Y-RHP} for $Y$.

\subsection{The Toeplitz determinant and the Hamiltonian}

Let us conclude this section by establishing a relation between the logarithmic derivative of the Toeplitz determinant $D_n(t)$ and the Hamiltonian $H(\tau)=H(\tau; \alpha, \beta)$ in \eqref{sh} when $n$ is large enough.

Tracing back the transformations $Y(z) \mapsto T(z) \mapsto S(z) \mapsto R(z)$ in \eqref{t}, \eqref{s} and \eqref{finalr}, we have
\begin{equation}
 Y(z)=R(z)E(z)\Psi(\zeta(z); -2int)  e^{\pm\frac{1}{2}(\alpha-\beta)\pi i \sigma_3} z^{\frac{n}{2}\sigma_3}  \tilde{w}(z)^{-\frac{\sigma_3}{2}}, \quad \pm \Im z >0
\end{equation}
for $|z| < 1$ and $z \in U(1,\delta)$. With the explicit formula for $E(z)$ in \eqref{E} and the approximation in \eqref{Rinfinity}, as $n \to \infty$, we get
\begin{equation}
   \left(Y^{-1}(z)\frac{d}{dz}Y(z)\right)_{21}=\frac{z^{n-1}}{it\tilde{w}(z)}e^{\pm(\alpha-\beta)\pi i}\left(\Psi^{-1}(\zeta)\frac{d}{d\zeta}\Psi(\zeta)\right)_{21}+  O(1) , \quad \pm \Im z>0.
\end{equation}
Substituting above equation into the differential identity \eqref{differentialindentityfordn}, we have, as $n \to \infty$,
\begin{eqnarray}
  \frac{d}{dt}\ln D_n(t)
   & = & \notag - \frac{\gamma e^{(\alpha-\beta)\pi i}}{2 \pi i t}
   \lim_{\zeta \to 1}\left((\Psi_+(\zeta))^{-1}\frac{d}{d\zeta}\Psi_+(\zeta)\right)_{21}       \\
   &  & \notag -\frac{\gamma e^{-(\alpha-\beta)\pi i}}{2 \pi i t}
   \lim_{\zeta \to -1}\left((\Psi_+(\zeta))^{-1}\frac{d}{d\zeta}\Psi_+(\zeta)\right)_{21}+O(1)   \\
   & = & -\frac{\gamma}{2\pi it}\left(e^{(\alpha-\beta)\pi i}\Psi_1^{(1)}(-2int)+e^{-(\alpha-\beta)\pi i}\Psi_1^{(-1)}(-2int)\right)_{21}+O(1),\label{integral2}
\end{eqnarray}
where $\Psi_1^{(1)}(\cdot)$ and $\Psi_1^{(-1)}(\cdot)$ are functions given by the local behaviors of $\Psi(z;\tau)$ as $z \to \pm 1$; see \eqref{localbehaviorforpsinear1} and \eqref{localbehaviorforpsinear-1}. As $\tau = -2int$, recalling the relation among $H(\tau)$, $\Psi_1^{(1)}(\tau)$ and $\Psi_1^{(-1)}(\tau)$ in \eqref{sH+ic}, we have from the above formula
\begin{equation}\label{differentialfordn}
 \frac{d}{dt}\ln D_n(t) = -2in H(\tau) + n \mathcal{L} +  O(1) , \qquad n \to \infty,
\end{equation}
uniformly for $nt$ bounded, with the constant $\mathcal{L}$ given by
\begin{equation}\label{constantl}
\mathcal{L}:= 2i  \lim_{i \tau \to 0^+} \frac{1}{\tau}\left[  \tau H(\tau) + \frac{\gamma}{2 \pi i }  \left( e^{\pi i (\alpha-\beta)}
 \left(\Psi_1^{(1)}(\tau)\right)_{21}+e^{-\pi i (\alpha-\beta)}\left(\Psi_1^{(-1)}(\tau)\right)_{21} \right) \right].
\end{equation}
The existence of the above limit will be established in Section \ref{sec: main-proof}.

\section{Asymptotic analysis of the model RH problem as $i \tau \to +\infty$} \label{Sec: large-s}

To obtain our final asymptotics for the deformed Fredholm determinant, we need asymptotics of $H(\tau)$ as $i \tau \to +\infty$ and $i \tau \to 0^+$; see the relation among $\det (I - \gamma \mathcal{K}^{(\alpha,\beta)}_s)$, $D_n(t)$ and $H(\tau)$ in \eqref{fredholmandtoeplitz} and \eqref{differentialfordn}. Based on \eqref{h}, this can be done by performing the steepest descent analysis for the model RH problem \ref{modelrhp} for $\Psi(z;\tau)$ as $i \tau \to +\infty$ and $i \tau \to 0^+$. Besides, to determine the constant term, asymptotics of the functions $u_k(\tau)$, $v_k(\tau)$, $k=1,2$ are also needed. In the present section, we will consider the case when $i \tau \to +\infty$.

\subsection{Normalization}
We first normalize the behavior of $\Psi(z;\tau)$ at infinity with the transformation below
\begin{equation}\label{a}
 A(z)=\Psi(z;\tau)e^{-\frac{\tau z}{4}\sigma_3}.
\end{equation}
Then, $A(z)$ satisfies the following RH problem.
\begin{rhp}
\quad 
 \begin{itemize}
  \item [(a)] $A(z)$ is analytic on $z \in \mathbb{C} \setminus \{\cup_{i=1}^7 \Sigma_i \}$, where the oriented contours are the same as those shown in Figure \ref{figure1}.

  \item [(b)] $A(z)$ satisfies the jump condition
              \begin{equation}\label{jumpfora}
               A_+(z)=A_-(z)J_A(z), \qquad z \in \cup_{i=1}^7 \Sigma_i,
              \end{equation}
              where
              \begin{equation}
               J_A(z) =
               \left\{
               \begin{aligned}
                 & \begin{pmatrix} 1 & 0 \\ e^{-\pi i(\alpha-\beta)}e^{-\frac{\tau z}{2}} & 1 \end{pmatrix},  & z \in \Sigma_1, \\
                 & \begin{pmatrix} 1 & 0 \\ e^{\pi i(\alpha-\beta)}e^{-\frac{\tau z}{2}} & 1 \end{pmatrix},  & z \in \Sigma_2, \\
                 & \begin{pmatrix} 1 & -e^{-\pi i(\alpha-\beta)}e^{\frac{\tau z}{2}} \\ 0 & 1 \end{pmatrix},  & z \in \Sigma_3, \\
                 & e^{2\pi i \beta \sigma_3},   & z \in \Sigma_4, \\
                 & \begin{pmatrix} 1 & -e^{\pi i(\alpha-\beta)}e^{\frac{\tau z}{2}} \\ 0 & 1 \end{pmatrix}, & z \in \Sigma_5, \\
                 & \begin{pmatrix} 0 & -e^{\pi i(\alpha-\beta)}e^{\frac{\tau z}{2}} \\ e^{-\pi i(\alpha-\beta)}e^{-\frac{\tau z}{2}} & 1-\gamma \end{pmatrix}, & z \in \Sigma_6, \\
                 & \begin{pmatrix} 0 & -e^{-\pi i(\alpha-\beta)}e^{\frac{\tau z}{2}} \\ e^{\pi i(\alpha-\beta)}e^{-\frac{\tau z}{2}} & 1-\gamma \end{pmatrix}, & z \in \Sigma_7. \\
               \end{aligned}
               \right.
              \end{equation}
              For convenience, we denote the jump matrices as $J_{A,i}$ on $\Sigma_i$.

  \item [(c)] As $z \to \infty$, we have
        \begin{equation}\label{infinitybehaviorfora}
         A(z)=\left(I + O\left(\frac{1}{z}\right)\right)z^{-\beta\sigma_3},
        \end{equation}
        where the branch of $z^{\beta}$ is taken as \eqref{infinitybehaviorforpsi} such that $\arg z \in (-\frac{\pi}{2}, \frac{3\pi}{2})$.

  \item [(d)] The local behaviors of $A(z)$ near $\pm1, \ 0$ are the same as $\Psi(z;\tau)$; cf. \eqref{localbehaviorforpsinear0}, \eqref{localbehaviorforpsinear1} and \eqref{localbehaviorforpsinear-1}.
 \end{itemize}
\end{rhp}

\subsection{Opening of the lens}
As $i \tau \to +\infty$, the factors $e^{-\frac{\tau z}{2}}$ and $e^{\frac{\tau z}{2}}$ in $ J_{A,6}(z)$ and $ J_{A,7}(z)$ are highly oscillatory for $z \in (-1,1)$, namely $z \in \Sigma_6 \cup \Sigma_7$. To eliminate the highly oscillatory terms, we deform contours near $(-1,1)$ by opening lens. This is based on the observation that both $J_{A,6}(z)$ and $J_{A,7}(z)$ can be factorized as
\begin{eqnarray}
 J_{A,6}(z)
  &=& \begin{pmatrix} 1 & -\frac{1}{1-\gamma}e^{\pi i(\alpha-\beta)}e^{\frac{\tau z}{2}} \\ 0 & 1 \end{pmatrix}
  \begin{pmatrix} \frac{1}{1-\gamma} & 0 \\ 0 & 1-\gamma \end{pmatrix}
  \begin{pmatrix} 1 & 0 \\ \frac{1}{1-\gamma}e^{-\pi i(\alpha-\beta)}e^{-\frac{\tau z}{2}}  & 1 \end{pmatrix}, \\
   J_{A,7}(z)
  &= & \begin{pmatrix} 1 & -\frac{1}{1-\gamma}e^{-\pi i(\alpha-\beta)}e^{\frac{\tau z}{2}} \\ 0 & 1 \end{pmatrix}
  \begin{pmatrix} \frac{1}{1-\gamma} & 0 \\ 0 & 1-\gamma \end{pmatrix}
  \begin{pmatrix} 1 & 0 \\ \frac{1}{1-\gamma}e^{\pi i(\alpha-\beta)}e^{-\frac{\tau z}{2}}  & 1 \end{pmatrix}.
\end{eqnarray}
From the above decomposition, the second transformation is defined as
\begin{equation}\label{openlens}
 B(z)=A(z) \left\{
 \begin{aligned}
   & \begin{pmatrix} 1 & 0 \\ -\frac{1}{1-\gamma}e^{-\pi i(\alpha-\beta)}e^{-\frac{\tau z}{2}}  & 1 \end{pmatrix}, & z \in \Omega_6,     \\
   & \begin{pmatrix} 1 & -\frac{1}{1-\gamma}e^{\pi i(\alpha-\beta)}e^{\frac{\tau z}{2}} \\ 0 & 1 \end{pmatrix}, & z \in \Omega_7,     \\
   & \begin{pmatrix} 1 & 0 \\ -\frac{1}{1-\gamma}e^{\pi i(\alpha-\beta)}e^{-\frac{\tau z}{2}}  & 1 \end{pmatrix}, & z \in \Omega_8,     \\
   & \begin{pmatrix} 1 & -\frac{1}{1-\gamma}e^{-\pi i(\alpha-\beta)}e^{\frac{\tau z}{2}} \\ 0 & 1 \end{pmatrix}, & z \in \Omega_9,     \\
   & I ,                          & \mathrm{otherwise},
 \end{aligned}
 \right.
\end{equation}
where the regions $\Omega_i,\ i = 6,7,8,9$, are depicted in Figure \ref{figure2}.

\begin{figure}[h]
 \centering
 \includegraphics[width=5.2in]{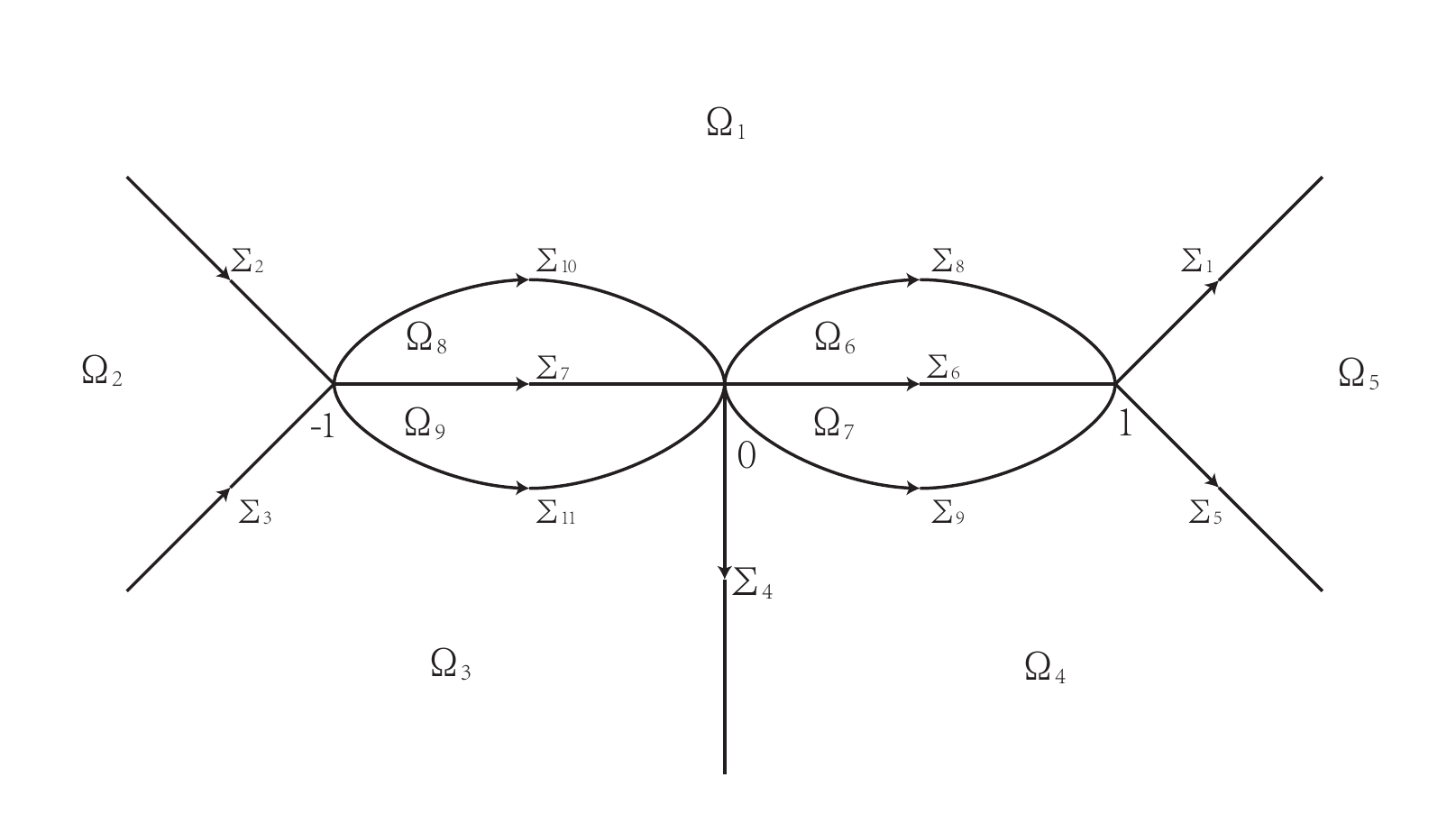}
 \caption{Contours for the RH problem for $B(z)$ and the regions $\Omega_k$, $k=1,\cdots, 9$.}
 \label{figure2}
\end{figure}

It is straightforward to verify that $B(z)$ satisfies the following RH problem.
\begin{rhp}
 \hfill
 \begin{itemize}
  \item [(a)] $B(z)$ is analytic for $z \in \mathbb{C} \setminus \{\cup_{i=1}^{11} \Sigma_i \}$, where the oriented contours are shown in Figure \ref{figure2}.

  \item [(b)] $B(z)$ satisfies the following jump condition
        \begin{equation}\label{jumpforb}
         B_+(z)=B_-(z)J_B(z), \qquad z \in \cup_{i=1}^{11} \Sigma_i,
        \end{equation}
        where
        \begin{equation} \label{jumpforb-formula}
         J_B(z) = \begin{cases}
           \begin{pmatrix} 1 & 0 \\ \frac{1}{1-\gamma}e^{-\pi i(\alpha-\beta)}e^{-\frac{\tau z}{2}}  & 1 \end{pmatrix}, & z \in \Sigma_8,               \\
           \begin{pmatrix} 1 & -\frac{1}{1-\gamma}e^{\pi i(\alpha-\beta)}e^{\frac{\tau z}{2}} \\ 0 & 1 \end{pmatrix}, & z \in \Sigma_9,               \\
           \begin{pmatrix} 1 & 0 \\ \frac{1}{1-\gamma}e^{\pi i(\alpha-\beta)}e^{-\frac{\tau z}{2}}  & 1 \end{pmatrix}, & z \in \Sigma_{10},            \\
           \begin{pmatrix} 1 & -\frac{1}{1-\gamma}e^{-\pi i(\alpha-\beta)}e^{\frac{\tau z}{2}} \\ 0 & 1 \end{pmatrix}, & z \in \Sigma_{11},            \\
           \begin{pmatrix} \frac{1}{1-\gamma} & 0 \\ 0 & 1-\gamma \end{pmatrix}, & z \in \Sigma_6 \cup \Sigma_7 \\
           J_A(z), & z \in \cup_{i=1}^5 \Sigma_i. \\
         \end{cases}
        \end{equation}

  \item [(c)] As $z \to \infty$, we have
        \begin{equation}\label{infinitybehaviorforb}
         B(z)=\left(I + O\left(\frac{1}{z}\right)\right)z^{-\beta\sigma_3}.
        \end{equation}

  \item [(d)] The local behaviors of $B(z)$ near $\pm1, \ 0$ are the same as $\Psi(z;\tau)$; cf. \eqref{localbehaviorforpsinear0},\eqref{localbehaviorforpsinear1} and \eqref{localbehaviorforpsinear-1}.
 \end{itemize}
\end{rhp}

\subsection{Global parametrix}

As $i \tau \to +\infty$, the factors $e^{-\frac{\tau z}{2}}$ and $e^{\frac{\tau z}{2}}$ are exponentially small when $z$ in upper and lower half plane, respectively. Therefore, all the jump matrices of $B(z)$ decay to the identity matrix exponentially as $i \tau \to +\infty$, except the ones on $\Sigma_4, \Sigma_6, \Sigma_7$. So we focus on the remaining jumps on $\Sigma_4 \cup \Sigma_6 \cup \Sigma_7$ and consider the following RH problem for $B^{(\infty)}(z)$.

\begin{rhp}
 \quad 
 \begin{itemize}
  \item [(a)] $B^{(\infty)}(z)$ is analytic for $z \in \mathbb{C} \setminus \{(-1,1) \cup \Sigma_4 \}$.
  \item [(b)] $B^{(\infty)}(z)$ satisfies the jump condition
        \begin{equation}\label{jumpforpinfinity}
         B^{(\infty)}_+(z)=B^{(\infty)}_-(z)
         \begin{cases}
           \begin{pmatrix} \frac{1}{1-\gamma} & 0 \\ 0 & 1-\gamma \end{pmatrix},
           & z \in (-1,1),      \\
           e^{2\pi i \beta \sigma_3},   & z \in \Sigma_4.
         \end{cases}
        \end{equation}
  \item [(c)] When $z \to \infty$, we have
        \begin{equation}\label{infinitybehaviorforpinfinity}
         B^{(\infty)}(z)=\left(I + O\left(\frac{1}{z}\right)\right)z^{-\beta\sigma_3},
        \end{equation}
        where the branch of $z^{\beta}$ is taken along the negative imaginary axis such that $\arg z \in (-\frac{\pi}{2}, \frac{3\pi}{2})$.
 \end{itemize}
\end{rhp}

\begin{lemma}
With the constant $c$ defined in \eqref{ic}, the solution to above RH problem is explicitly given as
  \begin{equation}\label{pinfinity}
   B^{(\infty)}(z)=\left(\frac{z+1}{z-1}\right)^{ic \sigma_3}z^{-\beta \sigma_3},
  \end{equation}
  where the branch for $z^{\beta}$ is taken such that $\arg z \in (-\frac{\pi}{2}, \frac{3\pi}{2})$ and the principal branch of $\left(\frac{z+1}{z-1}\right)^{ic}$ is taken.
\end{lemma}
\begin{proof}
  Based on the branches chosen above, it is easy to see that, for $z \in (-1,1)$,
  \begin{equation}
    \left(B^{(\infty)}_-(z)\right)^{-1}B^{(\infty)}_+(z) = \left(\frac{z+1}{z-1}\right)_-^{-ic \sigma_3}\left(\frac{z+1}{z-1}\right)_+^{ic \sigma_3} = e^{2\pi c \sigma_3} = \begin{pmatrix} \frac{1}{1-\gamma} & 0 \\ 0 & 1-\gamma \end{pmatrix},
  \end{equation}
  and
  \begin{equation}
    \left(B^{(\infty)}_-(z)\right)^{-1}B^{(\infty)}_+(z) = z_-^{\beta \sigma_3} z_+^{-\beta \sigma_3} = e^{2\pi i \beta \sigma_3}, \qquad z \in \Sigma_4.
  \end{equation}
  Therefore, the jump condition $(b)$ is verified. Conditions $(a),(c)$ directly follow from the definition of $ B^{(\infty)}(z)$. This completes our proof.
\end{proof}

The global parametrix $B^{(\infty)}(z)$ fails to approximate $B(z)$ near the endpoints $\pm 1$ and $0$. In the coming three subsections, we will construct the local parametrices near these endpoints separately. It turns out that all of them are given in terms of the confluent hypergeometric parametrix $\Phi(\zeta; a,b)$ in Appendix \ref{sec:Whittaker}, with different parameters $a$ and $b$.

\subsection{Local parametrix at 1}
Let $U(z_0, \delta)$ be an open disc centered at $z_0$ with fixed radius $\delta$. We look for a function $B^{(1)} (z)$ satisfying a RH problem in $U(1,\delta)$ as follows.

\begin{rhp}\label{rhpforp1}
 \quad 
 \begin{itemize}
  \item [(a)] $B^{(1)}(z)$ is analytic in $U(1, \delta) \setminus \{\cup \Sigma_i, \ i=1,5,6,8,9 \}$; see Figure \ref{figure2} for the contours.

  \item [(b)] $B^{(1)}(z)$ satisfies the following jump condition
        \begin{equation}\label{jumpforp1}
         B^{(1)}_+(z)=B^{(1)}_-(z)J_B(z), \qquad z \in U(1,\delta) \cap \{ \cup \Sigma_i , \ i=1,5,6,8,9 \},
        \end{equation}
        where $J_B(z)$ is given in \eqref{jumpforb-formula}.

  \item [(c)] As $i \tau \to +\infty$,  we have the matching condition
        \begin{equation}\label{matchforp1}
         B^{(1)}(z)=\left(I + O\left(\frac{1}{\tau}\right)\right)B^{(\infty)}(z), \qquad  z \in \partial U(1,\delta).
        \end{equation}
 \end{itemize}
\end{rhp}


Let us  introduce a conformal mapping $\zeta^{(1)}(z)$ near $z = 1$ as
\begin{equation}\label{conformalmappingat1}
 \zeta^{(1)}(z)=e^{\pi i}\frac{\tau(z-1)}{2}, \qquad z \in U(1,\delta).
\end{equation}
Since we take $i \tau \in \mathbb{R}^+$, it is obvious that $\zeta^{(1)}(z)$ simply maps the neighborhood near 1 on $z$-plane to the neighborhood near 0 on $\zeta^{(1)}$-plane with a rotation of $\frac{\pi}{2}$ counterclockwise. Then, we have the following result.

\begin{lemma}\label{lemmap1}
 Let $\Phi(\zeta; a,b)$ be the confluent hypergeometric parametrix given in Appendix \ref{sec:Whittaker}. Then, the solution to  RH problem \ref{rhpforp1}  is given by
 \begin{equation}\label{p1}
   B^{(1)}(z)=E^{(1)}(z)\Phi(\zeta^{(1)}(z);0,ic)(1-\gamma)^{-\frac{\sigma_3}{4}}
   \begin{cases}
    e^{-\frac{1}{2}\pi i (\alpha-\beta)\sigma_3}e^{-\frac{\tau z}{4}\sigma_3},  & \Im z>0, \\
    \begin{pmatrix} 0 & 1 \\ -1 & 0 \end{pmatrix}
    e^{-\frac{1}{2}\pi i (\alpha-\beta)\sigma_3}e^{-\frac{\tau z}{4}\sigma_3},  & \Im z<0.
   \end{cases}
 \end{equation}
The prefactor $E^{(1)}(z)$ is an analytic function in $U(1,\delta)$ and defined as
 \begin{equation}\label{e1}
  E^{(1)}(z)=\left(\frac{z+1}{2} \right)^{ic \sigma_3} z^{-\beta \sigma_3} e^{\frac{\tau}{4}\sigma_3} \tau^{ic\sigma_3} e^{\frac{1}{2}\pi i (\alpha-\beta) \sigma_3} (1-\gamma)^{\frac{1}{4}\sigma_3},
 \end{equation}
 where the branch for $z^{\beta}$ is taken such that $\arg z \in (-\frac{\pi}{2}, \frac{3\pi}{2})$ and the principal branch of $(z+1)^{ic}$ is taken.
\end{lemma}

\begin{proof}
  Obviously, $E^{(1)}(z)$  is analytic in $U(1,\delta)$. Next, let us verify the jump condition of $ B^{(1)}(z)$. For $z \in \Sigma_1$, it is readily seen from  \eqref{conformalmappingat1} that $\zeta^{(1)} \in e^{\frac{3\pi i}{4}}\mathbb{R}^{+}$. Then, we have from \eqref{jumpforphi} that
 \begin{equation}
  \begin{aligned}
    & \left(B_-^{(1)}(z)\right)^{-1}B_+^{(1)}(z)    \\
    & = e^{\frac{\tau z}{4}\sigma_3}e^{\frac{1}{2}\pi i (\alpha-\beta)\sigma_3}(1-\gamma)^{\frac{\sigma_3}{4}}\begin{pmatrix} 1 & 0 \\ e^{-\pi c} & 1 \end{pmatrix}(1-\gamma)^{-\frac{\sigma_3}{4}}e^{-\frac{1}{2}\pi i (\alpha-\beta)\sigma_3}e^{-\frac{\tau z}{4}\sigma_3}    \\
    & = \begin{pmatrix} 1 & 0 \\ e^{-\pi i(\alpha-\beta)}e^{-\frac{\tau z}{2}} & 1 \end{pmatrix},
  \end{aligned}
 \end{equation}
where we use the relation $e^{-2 \pi c}=1-\gamma$ in the last step. When $z$ belongs to the remaining contours, the jump condition can be verified similarly.

For the matching condition \eqref{matchforp1}, notice that $\zeta^{(1)}(z) = O(\tau)$ when $|z-1|=\delta$ and $i \tau \to +\infty$. It then follows from the asymptotic behavior of $\Phi$ at infinity in \eqref{infinitybehaviorforphi} that
 \begin{equation}
   \begin{aligned}
    & B^{(1)}(z)\left( B^{(\infty)}(z) \right)^{-1} \\
    & = E^{(1)}(z) \left(I+O\left(\frac{1}{\tau}\right)\right) \left(\frac{\tau(z-1)}{2}\right)^{-ic\sigma_3}e^{-\frac{\tau}{4}\sigma_3} (1-\gamma)^{-\frac{\sigma_3}{4}}e^{-\frac{1}{2}\pi i (\alpha-\beta)\sigma_3}\left( B^{(\infty)}(z) \right)^{-1}.
    \end{aligned}
 \end{equation}
With the definition of $B^{(\infty)}(z)$ in \eqref{pinfinity}, we have
 \begin{equation}
   \left(\frac{\tau(z-1)}{2}\right)^{-ic\sigma_3} e^{-\frac{\tau}{4}\sigma_3} (1-\gamma)^{-\frac{\sigma_3}{4}}e^{-\frac{1}{2}\pi i (\alpha-\beta)\sigma_3}\left( B^{(\infty)}(z) \right)^{-1} = \left(E^{(1)}(z)\right)^{-1}.
 \end{equation}
As both $\tau$ and  $ic$ are pure imaginary, it is easy to see from \eqref{e1} that  $E^{(1)} (z)= O(1)$ as $i \tau \to +\infty$. Therefore, we have established  \eqref{matchforp1} and completed the proof of the lemma.
\end{proof}

\subsection{Local parametrix at $-1$}
Similar to the scenario near $z=1$, we seek a function $B^{(-1)}(z)$ satisfying the following RH problem inside $U(-1,\delta)$.
\begin{rhp}\label{rhpforp-1}
 \quad
 \begin{itemize}
  \item [(a)] $B^{(-1)}(z)$ is analytic in $U(-1, \delta) \setminus \{\cup \Sigma_i, \ i=2,3,7,10,11 \}$; see Figure \ref{figure2} for the contours.


  \item [(b)] $B^{(-1)}(z)$ satisfies the following jump condition
        \begin{equation}\label{jumpforp-1}
         B^{(-1)}_+(z)=B^{(-1)}_-(z)J_B(z), \qquad z \in U(-1, \delta) \cap \{\cup \Sigma_i, \ i=2,3,7,10,11 \},
        \end{equation}
        where $J_B(z)$ is given in \eqref{jumpforb-formula}.
        
  \item [(c)] As $ i \tau \to +\infty$, we have the matching condition
        \begin{equation}\label{matchforp-1}
         B^{(-1)}(z)=\left(I + O\left(\frac{1}{\tau}\right)\right)B^{(\infty)}(z), \qquad z \in \partial U(-1,\delta).
        \end{equation}
 \end{itemize}
\end{rhp}

The construction of the local parametrix in $U(-1,\delta)$ is similar to that in $U(1,\delta)$.  The conformal mapping $\zeta^{(-1)}(z)$ near $z = -1 $ is defined as
\begin{equation}\label{conformalmappingat-1}
 \zeta^{(-1)}(z)=e^{\pi i}\frac{\tau(z+1)}{2}, \qquad z \in U(-1,\delta),
\end{equation}
which maps the neighborhood near $-1$ on $z$-plane to the neighborhood near 0 on $\zeta^{(-1)}$-plane with a rotation of $\frac{\pi}{2}$ counterclockwise. Then, the solution to the above RH problem is given in the following lemma.

\begin{lemma}
 Let $\Phi(\zeta;a,b)$ be the confluent hypergeometric parametrix given in Appendix \ref{sec:Whittaker}. Then, the solution to  RH problem \ref{rhpforp-1} is given by
 \begin{equation}\label{p-1}
   B^{(-1)}(z)=E^{(-1)}(z)\Phi(\zeta^{(-1)}(z);0,-ic)(1-\gamma)^{-\frac{\sigma_3}{4}}
   \begin{cases}
    e^{\frac{1}{2}\pi i (\alpha-\beta)\sigma_3}e^{-\frac{\tau z}{4}\sigma_3},& \Im z>0, \\
    \begin{pmatrix} 0 & 1 \\ -1 & 0 \end{pmatrix}
    e^{\frac{1}{2}\pi i (\alpha-\beta)\sigma_3}e^{-\frac{\tau z}{4}\sigma_3},  & \Im z<0.
   \end{cases}
 \end{equation}
 The prefactor $E^{(-1)}(z)$ is analytic in $U(-1,\delta)$ and defined as
 \begin{equation}\label{e-1}
  E^{(-1)}(z)=\left(\frac{z-1}{2}\right)^{-ic \sigma_3}z^{-\beta \sigma_3}  e^{-\frac{\tau}{4}\sigma_3} \tau^{-ic \sigma_3} e^{-\frac{1}{2}\pi i (\alpha-\beta) \sigma_3}
  \begin{cases}
    (1-\gamma)^{\frac{1}{4}\sigma_3}, \qquad & \Im z>0,  \\
    (1-\gamma)^{-\frac{3}{4}\sigma_3}, \qquad & \Im z<0,
  \end{cases}
 \end{equation}
 where the branch for $z^{\beta}$ is taken such that $\arg z \in (-\frac{\pi}{2}, \frac{3\pi}{2})$ and the principal branch of $(z-1)^{-ic}$ is taken.
\end{lemma}

\begin{proof}
We just establish the  analyticity of $E^{(-1)}(z)$ in $U(-1,\delta)$. The remaining proof is similar to that of Lemma \ref{lemmap1}. Due to the choice of the branches in \eqref{e-1}, the only possible jump for $E^{(-1)}(z)$ occurs on  $(-1 -\delta, -1+\delta)$, where we have
\begin{equation}
\Big( E^{(-1)}_-(z) \Big)^{-1} E^{(-1)}_+(z) = (1-\gamma)^{\frac{3}{4}\sigma_3} \left( \frac{(z-1)_-}{(z-1)_+} \right)^{ i c \sigma_3} (1-\gamma)^{\frac{1}{4}\sigma_3} = (1-\gamma)^{ \sigma_3} e^{2c\pi \sigma_3} = I.
\end{equation}
In the last identity above, the definition of the constant $c$ in \eqref{ic} is used. Obviously, $E^{(-1)}(z)$ is bounded as $z \to -1$. Therefore, $E^{(-1)}(z)$ is indeed analytic for $z \in U(-1,\delta)$.

 This completes the proof of the lemma.
\end{proof}

\subsection{Local parametrix at 0}
It remains to construct the local parametrix near $z = 0$. We look for a function $B^{(0)}(z)$ satisfying a RH problem in $U(0,\delta)$ as follows.

\begin{rhp}\label{rhpforp0}
 \quad
 \begin{itemize}
  \item [(a)] $B^{(0)}(z)$ is analytic in $U(0,\delta) \setminus \{\cup \Sigma_i, \ i=4,6,7,8,9,10,11 \}$; see Figure \ref{figure2} for the contours.


  \item [(b)] $B^{(0)}(z)$ satisfies the following jump condition
        \begin{equation}\label{jumpforp0}
         B^{(0)}_+(z)=B^{(0)}_-(z)J_B(z), \qquad z \in U(0,\delta) \cap \{\cup \Sigma_i, \ i=4,6,7,8,9,10,11 \},
        \end{equation}
        where $J_B(z)$ is given in \eqref{jumpforb-formula}.
        
  \item [(c)] As $i \tau \to +\infty$, we have the matching condition
        \begin{equation}\label{matchforp0}
         B^{(0)}(z)=\left(I + O\left(\frac{1}{\tau}\right)\right)B^{(\infty)}(z), \qquad z \in \partial U(0,\delta).
        \end{equation}
 \end{itemize}
\end{rhp}


The conformal mapping in $U(0,\delta)$  is defined by
\begin{equation}\label{conformalmappingat0}
 \zeta^{(0)}(z)=e^{\pi i}\frac{\tau z}{2},
\end{equation}
which maps the neighborhood near 0 on $z$-plane to the neighborhood near 0 on $\zeta^{(0)}$-plane with a rotation of $\frac{\pi}{2}$ counterclockwise. Then, we have the following solution for the above RH problem.

\begin{lemma}
 Let $\Phi(\zeta;a,b)$ be the confluent hypergeometric parametrix given in Appendix \ref{sec:Whittaker}. Then, the solution to  RH problem \ref{rhpforp0} for $B^{(0)}(z)$ is given by
 \begin{equation}\label{p0}
  B^{(0)}(z) = E^{(0)}(z)\Phi(\zeta^{(0)}(z);\alpha,\beta)(1-\gamma)^{-\frac{1}{2}\sigma_3}
  \begin{cases}
    e^{\frac{\pi i}{2} \alpha\sigma_3}e^{-\frac{\tau z}{4}\sigma_3},  & \arg z \in (0,\frac{\pi}{2}),      \\
    e^{-\frac{\pi i}{2} \alpha\sigma_3}e^{-\frac{\tau z}{4}\sigma_3},   & \arg z \in (\frac{\pi}{2},\pi),    \\
    e^{\frac{\pi i}{2}(\alpha + 2\beta)\sigma_3}\begin{pmatrix} 0 & 1 \\ -1 & 0 \end{pmatrix}e^{-\frac{\tau z}{4}\sigma_3},  & \arg z \in (-\pi, -\frac{\pi}{2}), \\
    e^{-\frac{\pi i}{2}(\alpha + 2\beta)\sigma_3}\begin{pmatrix} 0 & 1 \\ -1 & 0 \end{pmatrix}e^{-\frac{\tau z}{4}\sigma_3},  & \arg z \in (-\frac{\pi}{2},0).
  \end{cases}
 \end{equation}
The prefactor $E^{(0)}(z)$ is analytic in the neighborhood in $U(0,\delta)$ and defined as
 \begin{equation}\label{e0}
  E^{(0)}(z) = \left(\frac{z+1}{z-1}\right)^{ic \sigma_3} \left(\frac{\tau}{2}\right)^{\beta \sigma_3}(1-\gamma)^{\pm \frac{1}{2}\sigma_3}, \qquad \pm \Im z > 0,
 \end{equation}
 where the principal branch of $\left(\frac{z+1}{z-1}\right)^{ic}$ is taken.
\end{lemma}

\begin{proof}
First, let us check the analyticity of $E^{(0)}(z)$ in $U(0,\delta)$. According to the definition \eqref{e0}, the possible jump of $E^{(0)}(z)$ may occur on $(-\delta, \delta)$. For $z \in (-\delta, \delta)$, we have
 \begin{equation}
     \left(E_-^{(0)}(z)\right)^{-1}E_+^{(0)}(z)
   = (1-\gamma)^{\frac{1}{2}\sigma_3} \left( \frac{(z-1)_-}{(z-1)_+} \right)^{ i c \sigma_3} (1-\gamma)^{\frac{1}{2}\sigma_3} = (1-\gamma)^{ \sigma_3} e^{2c\pi \sigma_3} = I.
 \end{equation}
Obviously, $E^{(0)}(z)$ is bounded as $z \to 0$. Therefore, $E^{(0)}(z)$ is analytic for $z \in U(0,\delta)$.

 Now we verify the jump condtion. For $z \in \Sigma_6$, it is readily seen from \eqref{conformalmappingat0} that $\zeta^{(0)} \in e^{\frac{\pi i}{2}}\mathbb{R}^{+}$. Hence, it follows from \eqref{jumpforphi} that
 \begin{equation*}
  \begin{aligned}
      & \left(B_-^{(0)}(z)\right)^{-1}B_+^{(0)}(z)  \\
    & = e^{\frac{\tau z}{4}\sigma_3} \begin{pmatrix} 0 & -1 \\ 1 & 0 \end{pmatrix} (1-\gamma)^{\frac{1}{2}\sigma_3}e^{\frac{1}{2}\pi i (\alpha+2\beta) \sigma_3}\begin{pmatrix} 0 & e^{-\pi i \beta} \\ -e^{\pi i \beta} & 0 \end{pmatrix}(1-\gamma)^{-\frac{1}{2}\sigma_3}e^{\frac{1}{2}\pi i \alpha \sigma_3}e^{-\frac{\tau z}{4}\sigma_3}  \\
    & = e^{\frac{\tau z}{4}\sigma_3} (1-\gamma)^{-\frac{1}{2}\sigma_3} e^{-\frac{1}{2}\pi i (\alpha+2\beta) \sigma_3} \begin{pmatrix} e^{\pi i \beta} & 0 \\ 0 & e^{-\pi i \beta} \end{pmatrix} (1-\gamma)^{-\frac{1}{2}\sigma_3} e^{\frac{1}{2}\pi i \alpha \sigma_3} e^{-\frac{\tau z}{4}\sigma_3} = (1-\gamma)^{-\sigma_3}.
  \end{aligned}
 \end{equation*}
When $z$ belongs to the remaining contours, the jump condition can be verified similarly.

 For the matching condition \eqref{matchforp0}, we first note that $\zeta^{(0)}(z)= O(\tau)$ when $|z|= \delta$ and $i \tau \to +\infty$.
It  then follows from \eqref{pinfinity}, \eqref{p0}, \eqref{e0} and \eqref{infinitybehaviorforphi} that
 \begin{equation}
      B^{(0)}(z)\left(B^{(\infty)}(z) \right)^{-1}
    = E^{(0)}(z)\left(I+O\left(\frac{1}{\tau}\right)\right) \left(E^{(0)}(z)\right)^{-1}.
 \end{equation}
 As $\beta$ is purely imaginary, $E^{(0)}(z) = O(1)$ as $i \tau \to +\infty$ for $|z|= \delta$. Thus, the matching condition is justified, which completes our proof.
\end{proof}

\subsection{Final transformation}
With all of the parametrices constructed, the final transformation is defined as
\begin{equation}\label{r}
 \mathcal{R}(z) = B(z)
 \begin{cases}
   \left(B^{(0)}(z)\right)^{-1}, \qquad      & z \in U(0,\delta),  \\
   \left(B^{(1)}(z)\right)^{-1}, \qquad      & z \in U(1,\delta),  \\
   \left(B^{(-1)}(z)\right)^{-1}, \qquad     & z \in U(-1,\delta), \\
   \left(B^{(\infty)}(z)\right)^{-1}, \qquad & \mathrm{otherwise}. \\
 \end{cases}
\end{equation}
It is straightforward to verify that $\mathcal{R}(z)$ satisfies the following RH problem.
\begin{rhp} \label{RHP: R-large-s}
 \quad 
 \begin{itemize}
  \item [(a)] $\mathcal{R}(z)$ is analytic in $\mathbb{C} \setminus \Sigma_{\mathcal{R}}$,
         where $\Sigma_{\mathcal{R}}$ is illustrated in Figure \ref{figureforr}.

        \begin{figure}[h]
         \centering
         \includegraphics[width=5.4in]{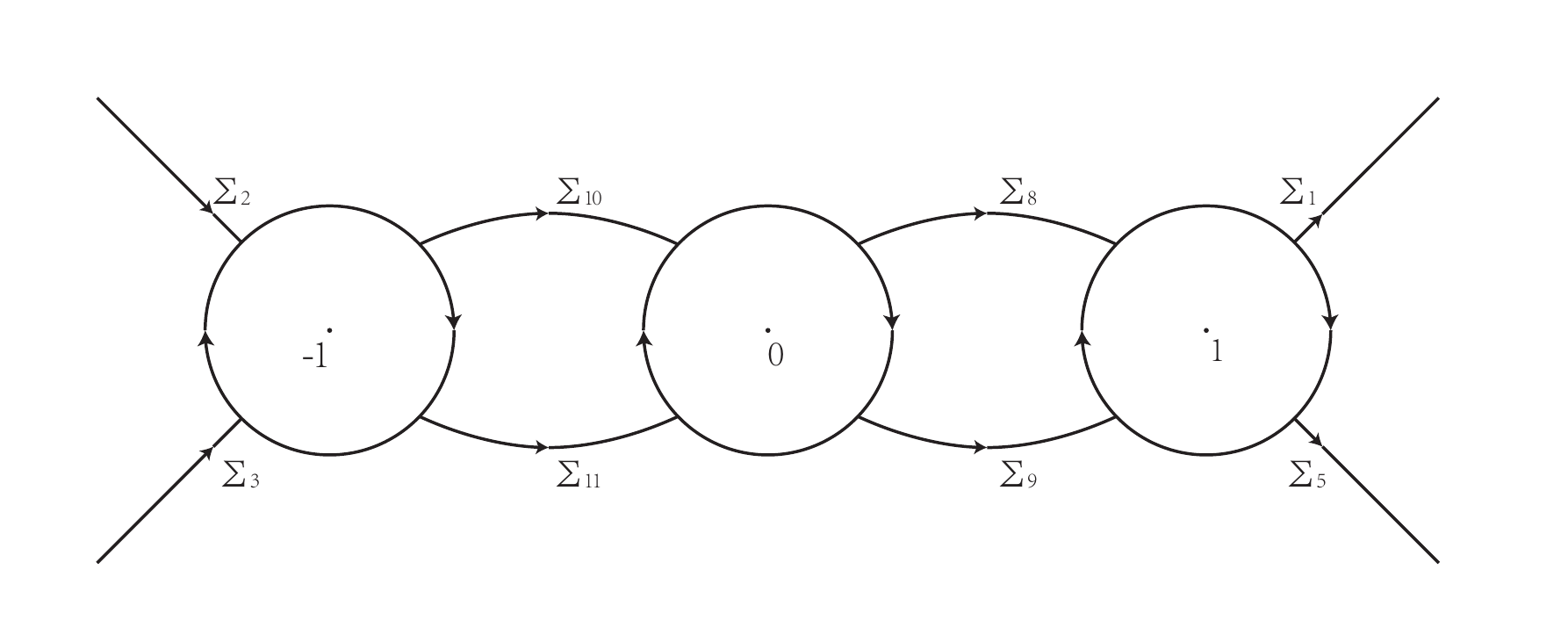}
         \caption{Contours for the RH problem for  $\mathcal{R}(z)$}
         \label{figureforr}
        \end{figure}

  \item [(b)] $\mathcal{R}(z)$ satisfies the jump condition
        \begin{equation}\label{jumpforr}
         \mathcal{R}_+(z)=\mathcal{R}_-(z)J_{\mathcal{R}}(z), \qquad z \in \Sigma_{\mathcal{R}},
        \end{equation}
        where
        \begin{equation}\label{jr}
         J_{\mathcal{R}}(z)=
         \begin{cases}
           B^{(1)}(z)\left(B^{(\infty)}(z)\right)^{-1}, \qquad  & z \in \partial U(1,\delta),     \\
           B^{(-1)}(z)\left(B^{(\infty)}(z)\right)^{-1}, \qquad  & z \in \partial U(-1,\delta),   \\
           B^{(0)}(z)\left(B^{(\infty)}(z)\right)^{-1}, \qquad & z \in \partial U(0,\delta),    \\
           B^{(\infty)}(z)J_B(z)\left(B^{(\infty)}(z)\right)^{-1}, \qquad & \mathrm{otherwise}.
         \end{cases}
        \end{equation}
  \item [(c)] As $z \to \infty$, we have
        \begin{equation}\label{rztoinfinity}
         \mathcal{R}(z) =I+O\left(\frac{1}{z}\right).
        \end{equation}
 \end{itemize}
\end{rhp}

It is easy to see from \eqref{jumpforb-formula} and \eqref{pinfinity} that $B^{(\infty)}(z)J_B(z)\left(B^{(\infty)}(z)\right)^{-1}$ tends to the identity matrix exponentially as $i \tau \to +\infty$, uniformly for $z \in \Sigma_{\mathcal{R}} \setminus \{ \partial U(-1,\delta) \cup \partial U(0,\delta) \cup \partial U(0,\delta)  \}$.  For $z \in \partial U(-1,\delta) \cup \partial U(0,\delta) \cup \partial U(0,\delta)$, it follows from \eqref{matchforp1}, \eqref{matchforp-1} and \eqref{matchforp0} that
\begin{equation}\label{jrasstoinfinity}
 J_{\mathcal{R}}(z) =I+\frac{J_{\mathcal{R},1}(z)}{\tau}+O\left(\frac{1}{\tau^2}\right), \qquad i \tau \to +\infty.
\end{equation}
This shows that RH problem \ref{RHP: R-large-s} is a small norm problem when $i \tau \to +\infty$, which gives us the following approximation of $\mathcal{R}$.

\begin{proposition}
We have
\begin{equation}\label{rstoinfinity}
  \mathcal{R}(z) = I + \frac{\mathcal{R}_1(z)}{\tau} + O\left(\frac{1}{\tau^2}\right), \qquad i \tau \to +\infty,
\end{equation}
uniformly for $z \in \mathbb{C} \setminus \Sigma_{\mathcal{R}}$, where
\begin{equation}\label{r111}
 \left(\mathcal{R}_1(z)\right)_{11} = -\left(\mathcal{R}_1(z)\right)_{22} = \begin{cases}
   -\frac{2c^2}{z+1}+\frac{2(\beta^2-\alpha^2)}{z},   & z \in U(1, \delta),  \\
   -\frac{2c^2}{z+1}-\frac{2c^2}{z-1},   & z \in U(0, \delta),  \\
   -\frac{2c^2}{z-1}+\frac{2(\beta^2-\alpha^2)}{z},   & z \in U(-1, \delta), \\
   -\frac{2c^2}{z+1}-\frac{2c^2}{z-1}+\frac{2(\beta^2-\alpha^2)}{z},  & \mathrm{otherwise}.
 \end{cases}
\end{equation}
\end{proposition}

\begin{proof}
The expansion \eqref{rstoinfinity} follows  from \eqref{jrasstoinfinity} and a standard argument for a small norm RH problem.

To get explicit expressions for the diagonal entries of $\mathcal{R}_1(z)$, we need more information about $J_{\mathcal{R},1}(z)$ in \eqref{jrasstoinfinity}. First we consider $|z-1|=\delta$. When $\Im z > 0$, it follows from \eqref{pinfinity}, \eqref{conformalmappingat1}-\eqref{e1} and \eqref{infinitybehaviorforphi} that
\begin{equation}\label{jr1near1}
  J_{\mathcal{R},1}(z)
   = -\frac{2c^2}{z-1} E^{(1)}(z) \begin{pmatrix}
   1 & \frac{\Gamma(-ic)}{\Gamma(1+ ic )} \\
   -\frac{\Gamma(ic)}{\Gamma(1-ic)} & -1
   \end{pmatrix}
   E^{(1)}(z)^{-1}.
\end{equation}
Note that $E^{(1)}(z)$ is an analytic function near $z = 1$, which implies that the above formula also holds for $\Im z < 0$ as well. By a similar argument, we have from \eqref{pinfinity}, \eqref{conformalmappingat-1}-\eqref{e-1}, \eqref{conformalmappingat0}-\eqref{e0} and \eqref{infinitybehaviorforphi},
\begin{equation}\label{jr1near-1}
  J_{\mathcal{R},1}(z)
   = -\frac{2c^2}{z+1}  E^{(1)}(z) \begin{pmatrix}
   1 & \frac{\Gamma(ic)}{\Gamma(1- ic )} \\
   -\frac{\Gamma(-ic)}{\Gamma(1+ic)} & -1
   \end{pmatrix}
   E^{(-1)}(z)^{-1}, \qquad |z+1|=\delta .
\end{equation}
and
\begin{equation}\label{jr1near0}
  J_{\mathcal{R},1}(z)
   =  - \frac{2(\alpha^2-\beta^2)}{z} \begin{pmatrix}
   1 & \frac{\Gamma(\alpha-\beta)}{\Gamma(\alpha+\beta+1)} \\
   -\frac{\Gamma(\alpha+\beta)}{\Gamma(\alpha-\beta+1)} & -1
   \end{pmatrix}
   E^{(-1)}(z)^{-1}, \qquad |z|=\delta .
\end{equation}

Now, we are ready to derive the explicit expression of $\mathcal{R}_1(z)$. A combination of \eqref{jumpforr}, \eqref{jrasstoinfinity} and \eqref{rstoinfinity} shows that $\mathcal{R}_1(z)$ satisfies the following RH problem.
\begin{rhp}
  \quad
  \begin{itemize}
    \item [(a)] $\mathcal{R}_1(z)$ is analytic in $\mathbb{C} \setminus \{\partial U(1,\delta) \cup \partial U(-1,\delta) \cup \partial U(0,\delta)\}$.
    \item [(b)] $\mathcal{R}_1(z)$ satisfies the jump condition
    \begin{equation}
     \mathcal{R}_{1,+}(z)=\mathcal{R}_{1,-}(z) +J_{\mathcal{R},1}(z), \qquad z \in \partial U(1,\delta) \cup \partial U(-1,\delta) \cup \partial U(0,\delta).
    \end{equation}
    \item [(c)] As $z \to \infty$, we have $\mathcal{R}_1(z) = O\left(\frac{1}{z}\right)$.
  \end{itemize}
\end{rhp}
By Plemelj's formula, we have
\begin{equation}\label{r1}
 \mathcal{R}_1(z)=
 \frac{1}{2\pi i} \int_{\partial U(0,\delta)} \frac{J_{\mathcal{R},1}(\zeta)}{\zeta-z}d\zeta
 +\frac{1}{2\pi i} \int_{\partial U(1,\delta)} \frac{J_{\mathcal{R},1}(\zeta)}{\zeta-z}d\zeta
 +\frac{1}{2\pi i} \int_{\partial U(-1,\delta)} \frac{J_{\mathcal{R},1}(\zeta)}{\zeta-z}d\zeta,
\end{equation}
As only the diagonal entries of $\Psi(z;\tau)$ are involved in the subsequent asymptotic derivation, we just focus on the diagonal entries of $\mathcal{R}_1(z)$.
Substituting the expressions of $J_{\mathcal{R},1}(z)$ in \eqref{jr1near1}-\eqref{jr1near0} into \eqref{r1}, we obtain \eqref{r111} by a direct residue computation and the observation $(J_{\mathcal{R},1}(z))_{11}=-(J_{\mathcal{R},1}(z))_{22}$.

This finishes the proof of the proposition.
\end{proof}

As an immediate consequence of \eqref{r111},  let us list the following estimation for $\mathcal{R}_1(z)$:
\begin{equation}\label{r1infinity}
 \mathcal{R}_1(z)=\frac{1}{z}\hat{\mathcal{R}}_1+O\left(\frac{1}{z^2}\right), \qquad z \to \infty,
\end{equation}
where the $(1,1)$-entry of $\hat{\mathcal{R}_1}$ is given by
\begin{equation}\label{r111infinity}
 \left(\hat{\mathcal{R}}_1\right)_{11}=-4c^2+2(\beta^2-\alpha^2).
\end{equation}
Moreover, one also gets asymptotics for the derivatives of $\mathcal{R}$ with respect to the parameter $\gamma$ from the expansion \eqref{rstoinfinity}:
\begin{equation} \label{rstoinfinity-diff}
\frac{\partial^k}{\partial \gamma^k}  \mathcal{R}(z) =  \frac{1}{\tau} \frac{\partial^k}{\partial \gamma^k} \mathcal{R}_1(z) + O\left(\frac{(\ln |\tau|)^k}{\tau^2}\right), \qquad i \tau \to +\infty, \quad k = 1,2,\cdots,
\end{equation}
uniformly for $z \in \mathbb{C} \setminus \Sigma_{\mathcal{R}}$; see similar analysis in \cite[Sec. 3.5]{Cha:Lene2021}.

\subsection{Proof of Theorem \ref{huvasymptotics}}

 With Propositions \ref{laxpairprop} and \ref{ydprop} and the steepest descent analysis above, we are ready to derive the large-$\tau$ asymptotics in Theorem \ref{huvasymptotics}.

First, let us consider the Hamiltonian $H(\tau)$, which is given in terms of the coefficient of $z^{-1}$-term in the large-$z$ expansion of $\Psi(z;\tau)$ in \eqref{infinitybehaviorforpsi}; see \eqref{h}. Tracing back the transformations $\Psi \mapsto A \mapsto B \mapsto \mathcal{R}$ in \eqref{a}, \eqref{openlens} and \eqref{r}, we have
\begin{equation}
\Psi(z;\tau)  = \mathcal{R}(z) B^{(\infty)}(z) e^{\frac{\tau z}{4}\sigma_3}
\end{equation}
for $z \in \mathbb{C} \setminus \{  U(-1,\delta) \cup   U(0,\delta) \cup   U(0,\delta) \}$. It then follows \eqref{pinfinity} and \eqref{rstoinfinity} that
\begin{equation}
  \Psi(z;\tau)
     = \left(I+\frac{\mathcal{R}_1(z)}{\tau}+O\left(\frac{1}{\tau^2}\right)\right)\left(\frac{z+1}{z-1}\right)^{ic \sigma_3}z^{-\beta \sigma_3}e^{\frac{\tau z}{4}\sigma_3}, \qquad i \tau \to +\infty.
\end{equation}
It is easy to see from \eqref{infinitybehaviorforpsi} that
\begin{equation}
\Psi_1(\tau) = \lim_{z \to \infty} z \left(\Psi(z;\tau) z^{\beta \sigma_3}e^{-\frac{\tau z}{4}\sigma_3} -I \right).
\end{equation}
Recalling \eqref{r1infinity}, the above two formulas give us
\begin{equation}
\Psi_1(\tau) = 2ic \sigma_3 + \frac{\hat{\mathcal{R}}_1}{\tau} +  O\left(\frac{1}{\tau^2}\right).
\end{equation}
Thus, we obtain the asymptotics of $H(\tau)$ in \eqref{hinfinity} from \eqref{h} and the above formula.

To tackle the asymptotics of the solutions to the coupled Painlev\'e V system,  we need asymptotics of the functions $y(\tau)$ and $d(\tau)$ given in \eqref{ypsi0} and \eqref{dpsi0}, which are related to the behavior of $\Psi(z;\tau)$ near $z = 0$;  see \eqref{localbehaviorforpsinear0}. From  \eqref{a}, \eqref{openlens} and \eqref{r}, we have
\begin{equation}
\Psi(z;\tau)  = \mathcal{R}(z) B^{(0)}(z) e^{\frac{\tau z}{4}\sigma_3}, \qquad z \in U(0, \delta).
\end{equation}
Recalling the definition of  $B^{(0)}(z)$ in \eqref{p0}, we get
\begin{equation}
  \Psi(z;\tau)
   = \mathcal{R}(z)\left(\frac{z+1}{z-1}\right)^{ic \sigma_3}  \left(\frac{\tau}{2}\right)^{\beta \sigma_3}(1-\gamma)^{\frac{1}{2}\sigma_3}\Phi(\zeta^{(0)}(z);\alpha,\beta)(1-\gamma)^{-\frac{1}{2}\sigma_3}
   e^{\frac{1}{2}\pi i \alpha\sigma_3}
\end{equation}
for $z \in U(0, \delta) $ and $\arg z \in (\frac{\pi}{4},\frac{\pi}{2})$. The above formula, together with the local behavior of $\Phi$ in \eqref{localbehaviorforphinear0}, yields
\begin{equation}
 \begin{aligned}
  \Psi(z;\tau)
   = \mathcal{R}(z)\left(I+O(z)\right)\left(\frac{\tau}{2}\right)^{\beta \sigma_3}
  \begin{pmatrix}
   \frac{\Gamma(1+\alpha-\beta)}{\Gamma(1+2\alpha)} & *\\
   \frac{\Gamma(1+\alpha+\beta)}{\Gamma(1+2\alpha)} & *
  \end{pmatrix}
  \left(I+O(z)\right)
  \left(\frac{\tau z}{2}\right)^{\alpha \sigma_3}(1-\gamma)^{-\frac{1}{2}\sigma_3},
 \end{aligned}
\end{equation}
as $z \to 0$, where * denotes certain unimportant entries. Substituting \eqref{rstoinfinity} into the above formula,  a direct comparison with \eqref{localbehaviorforpsinear0} shows that
\begin{eqnarray}
 \begin{aligned}
   \left(\Psi_0^{(0)}(\tau)\right)_{11} & = & \frac{\Gamma(1+\alpha-\beta)}{\Gamma(1+2\alpha)}\left(\frac{\tau}{2}\right)^{\alpha+\beta}(1-\gamma)^{-\frac{1}{2}}\left(1+O\left(\frac{1}{\tau}\right)\right), \\
   \left(\Psi_0^{(0)}(\tau)\right)_{21} & = & \frac{\Gamma(1+\alpha+\beta)}{\Gamma(1+2\alpha)}\left(\frac{\tau}{2}\right)^{\alpha-\beta}(1-\gamma)^{-\frac{1}{2}}\left(1+O\left(\frac{1}{\tau}\right)\right),
 \end{aligned}
\end{eqnarray}
as $i \tau \to +\infty$. Therefore, with \eqref{ypsi0} and \eqref{dpsi0}, we have, as $i \tau \to +\infty$,
\begin{eqnarray}
 y(\tau) & = & \frac{\Gamma(1+\alpha-\beta)}{\Gamma(1+\alpha+\beta)}\left(\frac{\tau}{2}\right)^{2\beta}\left(1+O\left(\frac{1}{\tau}\right)\right), \label{y}     \\
 d(\tau) & = & \frac{2\alpha}{1-\gamma}\frac{\Gamma(1+\alpha-\beta)\Gamma(1+\alpha+\beta)}{(\Gamma(1+2\alpha))^2}\left(\frac{\tau}{2}\right)^{2\alpha}\left(1+O\left(\frac{1}{\tau}\right)\right).\label{d}
\end{eqnarray}

Next, we move to derive asymptotics of $u_1(\tau)$ and $v_1(\tau)$. Recalling the coefficient matrix $A_1(\tau)$  in \eqref{a1}, we have
\begin{equation} \label{eq: v1-u1-relation}
v_1(\tau) = \frac{(A_1)_{21}}{(A_1)_{11}} y(\tau), \qquad u_1(\tau) = -\frac{(A_1)_{11}}{v_1(\tau)}.
\end{equation}
As $A_1(\tau)$ is the coefficient matrix of the $(z-1)^{-1}$ term in \eqref{l}, we need the asymptotics of $\Psi(z;\tau)$ when $z \to 1$. Let us focus on the sector $z \in U(1,\delta)$ and  $\arg (z-1) \in (\frac{\pi}{4},\frac{\pi}{2})$.
A combination of \eqref{a}, \eqref{conformalmappingat1}, \eqref{p1}, \eqref{r} and \eqref{differentialrelationofphi} yields
\begin{eqnarray*}
  && L(z)  = \frac{\partial}{\partial z}\Psi(z;\tau) \cdot \Psi(z;\tau)^{-1}
   =  \mathcal{R}^\prime(z)\mathcal{R}^{-1}(z)+\mathcal{R}(z)
   \left(E^{(1)}(z)\right)^\prime\left(E^{(1)}(z)\right)^{-1}\mathcal{R}^{-1}(z)       \\
   & & \qquad+\mathcal{R}(z)E^{(1)}(z)\left(-\frac{1}{2}\sigma_3 + \frac{2}{e^{\pi i}\tau(z-1)}
  \begin{pmatrix} -ic & \frac{\Gamma(1-ic)}{\Gamma(ic)} \\ \frac{\Gamma(1+ic)}{\Gamma(-ic)} & ic \end{pmatrix} \right)\left(E^{(1)}(z)\right)^{-1}\mathcal{R}^{-1}(z)(\zeta^{(1)}(z)) ^\prime,
\end{eqnarray*}
where $(\cdot)^\prime$ represents the derivative with respect to $z$. Note that both $\mathcal{R}(z)$ and $E^{(1)}(z)$ are analytic function for $z \in U(1,\delta)$. It then follows from \eqref{e1} and \eqref{rstoinfinity} that
\begin{equation}
  A_1(\tau)=
  \begin{pmatrix}
  -ic & * \\
  -ic\frac{\Gamma(1+ic)}{\Gamma(1-ic)}\tau^{-2ic}e^{-\frac{\tau}{2}}e^{-\pi i (\alpha-\beta)}(1-\gamma)^{-\frac{1}{2}} & *
  \end{pmatrix}
  \left(I+O\left(\frac{1}{\tau}\right)\right), \quad i \tau \to +\infty.
\end{equation}
Therefore, the asymptotics for $u_1(\tau)$ and $v_1(\tau)$ in \eqref{u1infinity} and \eqref{v1infinity} follow from \eqref{y}, \eqref{eq: v1-u1-relation} and the above formula.

The derivation of asymptotics of $u_2(\tau)$ and $v_2(\tau)$ in \eqref{u2infinity}  and \eqref{v2infinity} is analogous to that of $u_1(\tau)$ and $v_1(\tau)$. From \eqref{a}, \eqref{conformalmappingat-1}, \eqref{p-1}, \eqref{r} and \eqref{differentialrelationofphi}, we have
\begin{eqnarray*}
  && L(z)  = \frac{\partial}{\partial z}\Psi(z;\tau) \cdot \Psi(z;\tau)^{-1}  =  \mathcal{R}^\prime(z)\mathcal{R}^{-1}(z)+\mathcal{R}(z)\left(E^{(-1)}(z)\right)^\prime
             \left(E^{(-1)}(z)\right)^{-1}\mathcal{R}^{-1}(z)     \\
       & & \qquad  +\mathcal{R}(z)E^{(-1)}(z)\left(-\frac{1}{2}\sigma_3 + \frac{2}{e^{\pi i}\tau(z+1)}
  \begin{pmatrix} ic & \frac{\Gamma(1+ic)}{\Gamma(-ic)} \\ \frac{\Gamma(1-ic)}{\Gamma(ic)} & -ic \end{pmatrix} \right)\left(E^{(-1)}(z)\right)^{-1}\mathcal{R}^{-1}(z)(\zeta^{(-1)}(z))^\prime
\end{eqnarray*}
for $z \in U(-1, \delta)$ and $\arg (z+1) \in (\frac{\pi}{4},\frac{\pi}{2})$. Since the remaining computations are similar,  we omit the details.

This completes the proof of Theorem \ref{huvasymptotics}.  \hfill \qed

\section{Asymptotic analysis of the model RH problem as $i \tau \to 0^+$} \label{Sec: small-s}

In this section, we carry out a nonlinear steepest descent analysis for $\Psi(z;\tau)$ as $i \tau \to 0^+$. The analysis is  relatively simpler than the large-$\tau$ case in the previous section. Moreover, it is similar to the case $\gamma = 1$ in \cite[Sec. 3]{Xu:Zhao2020}.

\subsection{Global parametrix}

We first introduce a rescaling as
\begin{equation}\label{x}
 X(z;\tau)=\left(\frac{|\tau|}{2}\right)^{-\beta \sigma_3}\Psi\left(\frac{2z}{|\tau|};\tau\right).
\end{equation}
Under this transformation, the interval $[-1,1]$ is mapped to $\left[-\frac{|\tau|}{2},\frac{|\tau|}{2}\right]$, which is contained in  $U(0, \delta)$ for an arbitrary small but fixed positive constant $\delta>0$ as $i \tau \to 0^+$. Recalling the model RH problem \ref{modelrhp} for $\Psi(z;\tau)$, the parameter $\gamma$ only appears in the jump on $(-1,1)$ and the local behaviours near the endpoints $\pm 1$ and 0. As a consequence, the global parametrix for $z \in  \mathbb{C} \setminus U(0, \delta)$ is exactly the same as that in \cite[Sec. 3.1]{Xu:Zhao2020}. Here, we decide to skip the details and  just list the results for brevity.

The global parametrix is constructed in terms of the confluent hypergeometric functions as follows
\begin{equation}\label{pinfinityassto0}
 X^{(\infty)}(z)= e^{\frac{1}{2}\pi i \beta \sigma_3}\sigma_3 M(e^{\frac{\pi i}{2}}z) \sigma_3, \qquad \arg z \in (-\frac{\pi}{2},\frac{3}{2}\pi),
\end{equation}
where $M(z)$ is the $2 \times 2$ matrix-valued function given in \cite[Sec. 4.2.1]{Cla:Its:Kra2011}. More precisely, we have
\begin{equation}
 M(\zeta)=\zeta^{\alpha \sigma_3}
 \begin{pmatrix}
  \begin{smallmatrix}
   U(\alpha+\beta, 1+2\alpha, \zeta) & \frac{\Gamma(1+\alpha-\beta)}{\alpha+\beta}U(1+\alpha-\beta, 1+2\alpha, e^{-\pi i}\zeta)e^{-\pi i \alpha}e^{-\pi i \beta}\\
   \frac{\Gamma(1+\alpha+\beta)}{\alpha-\beta}U(1-\alpha+\beta, 1-2\alpha, \zeta)e^{2\pi i \beta} &
   U(-\alpha-\beta, 1-2\alpha, e^{-\pi i}\zeta)e^{\pi i \alpha}e^{\pi i \beta}
  \end{smallmatrix}
 \end{pmatrix}
 e^{-\frac{\zeta}{2}\sigma_3},
\end{equation}
for $\arg \zeta \in (\frac{\pi}{4}, \frac{3 \pi}{4})$, where $U(a,b,z)$ is the Kummer function; see \cite[Sec. 13.2]{NIST}. In addition, with the explicit expression of $X^{(\infty)}(z)$, we are able to obtain the  large-$z$ expansion for later use:
\begin{equation}\label{infinitybehaviorforpinfinityassto0}
 X^{(\infty)}(z)=\left(I+\frac{X_1^{(\infty)}(\tau)}{z}+O\left(\frac{1}{z^2}\right)\right)z^{-\beta \sigma_3}e^{-\frac{i}{2}z\sigma_3} \qquad \textrm{as }  z \to \infty,
\end{equation}
with
\begin{equation}\label{pinfinity1}
 X_1^{(\infty)}(\tau)=
 \begin{pmatrix}
  -(\alpha^2-\beta^2)i & ie^{-\pi i\beta}\frac{\Gamma(1+\alpha-\beta)}{\Gamma(\alpha+\beta)} \\
  -ie^{\pi i\beta}\frac{\Gamma(1+\alpha+\beta)}{\Gamma(\alpha-\beta)} & (\alpha^2-\beta^2)i
 \end{pmatrix}.
\end{equation}

\subsection{Local parametrix}

In $U(0,\delta)$, we seek a function $X^{(0)}(z)$ satisfying a RH problem as follows. Note that, the parameter $\gamma$ will take effect in this problem.

\begin{rhp}\label{rhpforp0assto0}
 \quad
 \begin{itemize}
  \item [(a)] $X^{(0)}(z)$ is analytic in $U(0,\delta) \setminus \Sigma^X$, where $\Sigma^X$ represents the jump contours for $X(z;\tau)$.
  \item [(b)] $X^{(0)}(z)$ satisfies the same jump condition as $X(z;\tau)$ in $U(0,\delta) \cap \Sigma_X$.
  \item [(c)] As $i \tau \to 0^+$, we have the following matching condition
              \begin{equation}\label{matchingconditionforp0assto0}
                X^{(0)}(z)=(I+O(\tau^{2\alpha+1}))X^{(\infty)}(z), \qquad z \in \partial U(0,\delta).
              \end{equation}
 \end{itemize}
\end{rhp}

To construction the solution to the above RH problem, let us define the following functions
\begin{equation}\label{k}
 k(\zeta):=-\frac{\gamma e^{\pi i (\alpha-\beta)}}{2 \pi i}\int_0^\frac{1}{2} \frac{x^{2\alpha}}{x-\zeta}dx
 -\frac{\gamma e^{\pi i (\alpha+\beta)}}{2 \pi i}\int_{-\frac{1}{2}}^0 \frac{|x|^{2\alpha}}{x-\zeta}dx,
\end{equation}
and
\begin{small}
\begin{eqnarray}
  &&X^{(\infty,0)}(z)
   :=e^{-\frac{iz}{2}}   \label{pinfinity0}        \\
  && \times \begin{pmatrix}
   \frac{\Gamma(1+\alpha-\beta)}{\Gamma(1+2\alpha)}\mathrm{M}(\alpha+\beta,1+2\alpha,iz)e^{-\frac{1}{2}\pi i (\alpha+\beta)}   &
   -\frac{\Gamma(2\alpha)}{\Gamma(\alpha+\beta)}\mathrm{M}(-\alpha+\beta,1-2\alpha,iz)e^{\frac{1}{2}\pi i (\alpha-\beta)}        \\
   \frac{\Gamma(1+\alpha+\beta)}{\Gamma(1+2\alpha)}\mathrm{M}(1+\alpha+\beta,1+2\alpha,iz)e^{-\frac{1}{2}\pi i (\alpha-\beta)} &
   \frac{\Gamma(2\alpha)}{\Gamma(\alpha-\beta)}\mathrm{M}(1-\alpha+\beta,1-2\alpha,iz)e^{\frac{1}{2}\pi i (\alpha+\beta)}
  \end{pmatrix}, \nonumber
\end{eqnarray}
\end{small}
where $\mathrm{M}(a,b,z)$ is the Kummer function; see \cite[Sec. 13.2]{NIST}. Then, we have the following result.

\begin{proposition} \label{prop: small-s-para}
Let $k(\zeta)$ and $X^{(\infty,0)}(z)$ be given in \eqref{k} and \eqref{pinfinity0}. Then, the solution to RH problem \ref{rhpforp0assto0} is given by, when $2\alpha \notin \mathbb{N}$,
  \begin{equation}\label{p0assto0}
    X^{(0)}(z)=X^{(\infty,0)}(z)
    \begin{pmatrix} 1 & |\tau|^{2 \alpha}k(\frac{z}{|\tau|}) \\ 0 & 1 \end{pmatrix}
      z^{\alpha \sigma_3}
      \begin{pmatrix} 1 & \frac{\sin \pi (\alpha+\beta)}{\sin 2\pi \alpha} \\ 0 & 1
    \end{pmatrix}C_i, \qquad   z \in \Omega^X_i, \\
  \end{equation}
where $C_i$ are the constant matrices given by the jump matrices in \eqref{jumpforpsi}, more precisely,
\begin{equation}\label{ci}
  C_1=I,\ C_2=J_2^{-1},\ C_3=J_2^{-1}J_3^{-1},\ C_4=J_1^{-1}J_5^{-1},\ C_5=J_1^{-1}.
\end{equation}
\end{proposition}

\begin{proof}
In the definition of $X^{(0)}(z)$ in \eqref{p0assto0}, the only $\gamma$-dependent factor is $k(\zeta)$ given in \eqref{k}.
When $\gamma  =1$,  $k(\zeta)$ is the same as the function in \cite[Eq. (3.12)]{Xu:Zhao2020}. Therefore, $X^{(0)}(z)$ equals to that in \cite[Eq. (3.12)]{Xu:Zhao2020} when $\gamma = 1$.

Let us just focus on the $\gamma$-dependent jump on $(- \frac{|\tau|}{2},\frac{|\tau|}{2})$. It is easy to see from \eqref{k} that, $k(\zeta)$ is analytic for $\zeta \in \mathbb{C} \setminus [-\frac{1}{2}, \frac{1}{2}]$ with the following jump condition
\begin{equation}\label{jumpfork}
     k_+(\zeta)-k_-(\zeta)=\begin{cases}
       -\gamma e^{\pi i (\alpha-\beta)}\zeta^{2\alpha}, \qquad   & \zeta \in \left(0, \frac{1}{2}\right), \\
       -\gamma e^{\pi i (\alpha+\beta)}|\zeta|^{2\alpha}, \qquad & \zeta \in \left(-\frac{1}{2},0\right). \\
     \end{cases}
   \end{equation}
This give us
  \begin{equation}
    \left(X_-^{(0)}(z)\right)^{-1}X_+^{(0)}(z) = J_5(z) J_1(z)
    \begin{pmatrix}
      1 & -\gamma e^{\pi i (\alpha-\beta)} \\
      0 & 1
    \end{pmatrix}
     = J_6(z), \qquad z \in \left(0,\frac{|\tau|}{2}\right).
  \end{equation}
  The jump on $(-\frac{|\tau|}{2},0 )$ can be verified similarly.

  Next we verify the matching condition. It follows from \eqref{k} that
  \begin{equation}\label{kinfinity}
     k(\zeta) = \frac{\gamma e^{\pi i \alpha}\cos \beta \pi}{\pi i 2^{2\alpha+1}(2\alpha+1)}\cdot\frac{1}{\zeta}+ O\left(\frac{1}{\zeta ^2}\right), \qquad \zeta \to \infty.
   \end{equation}
  Then, as $i \tau \to 0^+$,  $|z| = \delta$ and $z \in \Omega_1^X$, we have  from \eqref{pinfinityassto0}, \eqref{p0assto0} and the above formula that
  \begin{equation}\label{jrassto0}
   \begin{aligned}
       & X^{(0)}(z)\left(X^{(\infty)}(z) \right)^{-1}       \\
     & = I + |\tau|^{2 \alpha}k\left(\frac{z}{|\tau|}\right) X^{(\infty,0)}(z)
     \begin{pmatrix} 0 & 1 \\ 0 & 0 \end{pmatrix}\left(X^{(\infty,0)}(z)\right)^{-1}    \\
     & = I + \frac{\gamma e^{\pi i \alpha}\cos \beta \pi}{\pi i 2^{2\alpha+1}(2\alpha+1)}\cdot\frac{|\tau|^{2\alpha+1}}{z}
     X^{(\infty,0)}(z)\begin{pmatrix} 0 & 1 \\ 0 & 0 \end{pmatrix}\left(X^{(\infty,0)}(z)\right)^{-1}
     +O\left(\tau^{2\alpha+2}\right).
   \end{aligned}
  \end{equation}
  The matching condition in other sectors can be verified similarly. This completes our proof.
\end{proof}
\begin{remark} \label{rmk: small-s-para}
When $2 \alpha \in \mathbb{N}$, there is a logarithmic singularity in $\Psi(z;\tau)$ near $z = 0$; see \eqref{localbehaviorforpsinear0}. Nevertheless, the solution to RH problem \ref{rhpforp0assto0} can be constructed in a similar manner; see \cite[Sec. 3.2]{Xu:Zhao2020} for more details.
\end{remark}

\subsection{Final transformation}

Now we define the final transformation as
\begin{equation}\label{r_}
 Z(z) = X(z;\tau) \begin{cases}
   \left(X^{(\infty)}(z)\right)^{-1}, \qquad & |z|>\delta, \\
   \left(X^{(0)}(z)\right)^{-1}, \qquad & |z|<\delta.      \\
 \end{cases}
\end{equation}
It is easily seen that $Z(z)$ satisfies the following RH problem.
\begin{rhp}
  The $2 \times 2$ matrix-valued function $Z(z)$ defined in \eqref{r_} has the following properties:
  \begin{itemize}
    \item [(a)] $Z(z)$ is analytic in $\mathbb{C} \setminus \partial U(0,\delta)$.
    \item [(b)] $Z(z)$ satisfies the jump condition
    \begin{equation}
      Z_+(z)=Z_-(z)J_{Z}(z), \qquad |z| = \delta,
    \end{equation}
    where the jump contour is taken clockwise and
    \begin{equation}
      J_{Z}(z)=X^{(0)}(z)\left(X^{(\infty)}(z)\right)^{-1}.
    \end{equation}
    \item [(c)] As $z \to \infty$, we have
    \begin{equation}\label{infinitybehaviorforr_}
     Z(z)=I+\frac{Z_1(\tau)}{z}+O\left(\frac{1}{z^2}\right).
    \end{equation}
  \end{itemize}
\end{rhp}

The expression of $J_{Z}(z)$ has been given in \eqref{jrassto0}, which indeed holds for $|z| = \delta$. As $i \tau \to 0^+$, the above RH problem is again a small-norm one. Finally, a direct residue computation gives us
\begin{equation}\label{rassto0}
  Z(z)=\begin{cases}
   I + \ds \frac{\gamma e^{\pi i \alpha}\cos \beta \pi}{\pi i 2^{2\alpha+1}(2\alpha+1)}\cdot\frac{|\tau|^{2\alpha+1}}{z}
   \mathcal{X}(0)  +O\left(\tau^{2\alpha+2}\right),    \ &  |z|>\delta,   \\
   I + \ds \frac{\gamma e^{\pi i \alpha}\cos \beta \pi}{\pi i 2^{2\alpha+1}(2\alpha+1)}\cdot\frac{|\tau|^{2\alpha+1}}{z} \left(\mathcal{X}(0) - \mathcal{X}(z)\right)
   +O\left(\tau^{2\alpha+2}\right), \ & |z|<\delta,
  \end{cases}
\end{equation}
where
\begin{equation}
    \mathcal{X}(z) = X^{(\infty,0)}(z)\begin{pmatrix} 0 & 1 \\ 0 & 0 \end{pmatrix}\left(X^{(\infty,0)}(z)\right)^{-1}.
\end{equation}

\subsection{Proof of Theorem \ref{huvasymptotics1}}
Based on Propositions \ref{laxpairprop} and \ref{ydprop} and the above RH analysis, we are ready to derive the small-$\tau$ asymptotics in Theorem \ref{huvasymptotics1}.

First let us consider the Hamiltonian $H(\tau)$. Denote the large-$z$ behavior for $X(z;\tau)$ in \eqref{x} as
\begin{equation}
  X(z;\tau)=\left(I + \frac{X_1(\tau)}{z} + O(z^{-2})\right)z^{-\beta \sigma_3}e^{-\frac{i}{2}z\sigma_3}, \qquad z \to \infty.
\end{equation}
Recalling the transformations in \eqref{r_}, we have from the above formula and \eqref{infinitybehaviorforpsi} that
\begin{equation}
 \left(\Psi_1(\tau)\right)_{11}=\frac{2\left(X_1(\tau)\right)_{11}}{|\tau|} = \frac{2}{|\tau|}\left((Z_1(\tau))_{11}+(X_{1}^{(\infty)}(\tau))_{11}\right).
\end{equation}
With the relation in \eqref{h}, we obtain the asymptotics of $H(\tau)$ in \eqref{hassto0} from \eqref{pinfinity1} and \eqref{rassto0}.

Like in the large-$\tau$ case, we make use of $y(\tau)$ and $d(\tau)$ to derive the asymptotics of $u_k(\tau),v_k(\tau)$, $k = 1,2$. Due to the relations \eqref{ypsi0} and \eqref{dpsi0}, let us investigate the local behavior of $\Psi(z;\tau)$ as $z \to 0$. As mentioned in Proposition \ref{prop: small-s-para} and Remark \ref{rmk: small-s-para},  we only consider the case $2\alpha \notin \mathbb{N}$, while the case for $2\alpha \in \mathbb{N}$ is similar. From \eqref{x} and \eqref{r_}, we have
\begin{equation} \label{eq: Psi-near0-small-s}
\Psi\left(\frac{2z}{|\tau|};\tau\right)= \left(\frac{|\tau|}{2}\right)^{\beta \sigma_3} Z(z) X^{(0)}(z), \qquad z \in U(0,\delta).
\end{equation}
With the explicit expression of $X^{(0)}(z)$ in \eqref {p0assto0} and the approximation \eqref{rassto0}, we obtain from \eqref{localbehaviorforpsinear0} that
\begin{equation}
 \Psi_0^{(0)}(\tau)=
 \begin{pmatrix}
  \left(\frac{|\tau|}{2}\right)^{\alpha+\beta}e^{-\frac{\pi i (\alpha+\beta)}{2}}\frac{\Gamma(1+\alpha-\beta)}{\Gamma(1+2\alpha)}  & * \\
  \left(\frac{|\tau|}{2}\right)^{\alpha-\beta}e^{-\frac{\pi i (\alpha-\beta)}{2}}\frac{\Gamma(1+\alpha+\beta)}{\Gamma(1+2\alpha)} & *
 \end{pmatrix}\left(I+O(\tau^{2\alpha+1})\right), \qquad i \tau \to 0^+.
\end{equation}
It then follows from \eqref{ypsi0} and \eqref{dpsi0} that
\begin{eqnarray}
 y(\tau) & = & \frac{\Gamma(1+\alpha-\beta)}{\Gamma(1+\alpha+\beta)}e^{-\pi i \beta}
 \left(\frac{|\tau|}{2}\right)^{2\beta}\left(1+O(\tau^{2\alpha+1})\right), \quad i \tau \to 0^+, \label{yassto0}  \\
 d(\tau) & = & 2\alpha \frac{\Gamma(1+\alpha-\beta)\Gamma(1+\alpha+\beta)}{(\Gamma(1+2\alpha))^2}
 e^{-\pi i \alpha}\left(\frac{|\tau|}{2}\right)^{2\alpha}\left(1+O(\tau^{2\alpha+1})\right), \quad i \tau \to 0^+.\label{dassto0}
\end{eqnarray}

Now, we are ready to derive asymptotics of $u_1(\tau)$ and $v_1(\tau)$. Combining \eqref{a1intermsofpsi} and \eqref{eq: v1-u1-relation},  we have
\begin{equation}\label{u1v1}
  v_1(\tau) = \frac{\left(\Psi_0^{(1)}(\tau)\right)_{21}}{\left(\Psi_0^{(1)}(\tau)\right)_{11}}y(\tau), \qquad
  u_1(\tau) = -\frac{\gamma e^{\pi i (\alpha-\beta)}}{v_1(\tau)}\left(\Psi_0^{(1)}(\tau)\right)_{11}\left(\Psi_0^{(1)}(\tau)\right)_{21},
\end{equation}
where $\Psi_0^{(1)}(\tau)$ arises in  the local behavior of $\Psi(z;\tau)$ at $z=1$ in \eqref{localbehaviorforpsinear1}. Due to the relation \eqref{eq: Psi-near0-small-s}, we take $\frac{2z}{|\tau|} = 1 $. With the explicit expression of $X^{(0)}(z)$ in \eqref {p0assto0} and the following approximation
\begin{equation}
     k(\zeta) \sim -\frac{\gamma e^{\pi i (\alpha-\beta)}}{2 \pi i}2^{-2\alpha} \ln(\zeta - \frac{1}{2}), \qquad \zeta \to \frac{1}{2},\label{k1/2}
   \end{equation}
a straightforward computation yields
\begin{equation}
 \begin{pmatrix}
  \left(\Psi_0^{(1)}(\tau)\right)_{11} & * \\
  \left(\Psi_0^{(1)}(\tau)\right)_{21} & *
 \end{pmatrix}
 =
 \begin{pmatrix}
  \left(\frac{|\tau|}{2}\right)^{\alpha+\beta}e^{-\frac{\pi i (\alpha+\beta)}{2}}\frac{\Gamma(1+\alpha-\beta)}{\Gamma(1+2\alpha)} & * \\
  \left(\frac{|\tau|}{2}\right)^{\alpha-\beta}e^{-\frac{\pi i (\alpha-\beta)}{2}}\frac{\Gamma(1+\alpha+\beta)}{\Gamma(1+2\alpha)} & *
 \end{pmatrix}\left(I+O(\tau^{2\alpha+1})+O(\tau) \right),
\end{equation}
where the $O(\tau)$-term come from the fact that
\begin{equation}\label{pinfinity0aszto0}
     X^{(\infty,0)}(\pm |\tau|/2) = X^{(\infty,0)}(0)(I+O(\tau)), \qquad \quad i \tau \to 0^+.
 \end{equation}
Therefore, the small-$\tau$ asymptotics of $u_1(\tau)$ and $v_1(\tau)$ in \eqref{u1assto0} and \eqref{v1assto0} follow from the above equation and \eqref{u1v1}.

Following the analogues procedure, we can obtain the small-$\tau$ asymptotics of $u_2(\tau)$ and $v_2(\tau)$. In this case, we need to take $\frac{2z}{|\tau|} = -1 $, and make use of the following approximations
   \begin{equation}
     k(\zeta) \sim \frac{\gamma e^{\pi i (\alpha+\beta)}}{2 \pi i}2^{-2\alpha} \ln(\zeta + \frac{1}{2}), \qquad \zeta \to -\frac{1}{2}, \label{k-1/2}
   \end{equation}
and
\begin{equation}
 \begin{pmatrix}
  \left(\Psi_0^{(-1)}(\tau)\right)_{11} & * \\
  \left(\Psi_0^{(-1)}(\tau)\right)_{21} & *
 \end{pmatrix}
 =
 \begin{pmatrix}
  \left(\frac{|\tau|}{2}\right)^{\alpha+\beta}e^{\frac{\pi i (\alpha-\beta)}{2}}\frac{\Gamma(1+\alpha-\beta)}{\Gamma(1+2\alpha)} & * \\
  \left(\frac{|\tau|}{2}\right)^{\alpha-\beta}e^{\frac{\pi i (\alpha+\beta)}{2}}\frac{\Gamma(1+\alpha+\beta)}{\Gamma(1+2\alpha)} & *
 \end{pmatrix}\left(I+O(\tau^{2\alpha+1})+O(\tau) \right),
 \end{equation}
which comes from  \eqref{localbehaviorforpsinear-1}, \eqref{p0assto0}, \eqref{rassto0} and \eqref{pinfinity0aszto0}.

This completes the proof of Theorem \ref{huvasymptotics1}.  \hfill \qed

\begin{remark}
  Via a further computation, one can show that the coefficients of $O(\tau^{2\alpha+1})$-term in \eqref{v1assto0} and \eqref{v2assto0} are indeed 0. Hence, the asymptotics can be improved to be
  \begin{eqnarray}
    v_1(\tau) & = & 1+O(\tau)+O(\tau^{2\alpha+2}),\\
    v_2(\tau) & = & 1+O(\tau)+O(\tau^{2\alpha+2}).
  \end{eqnarray}
\end{remark}

\section{Proof of main results} \label{sec: main-proof}

In the last section, we first establish the integral representation \eqref{eq:F-TW-formula}  for the deformed Fredholm determinant with $\gamma \in [0,1]$. Then, with this expression, we derive the large gap asymptotics in Theorem \ref{mainresult}.

\subsection{Integral expression for the Fredholm determinant}

First, we show that the integral representation \eqref{eq:F-TW-formula} holds for all $\gamma \in [0,1]$.

\begin{lemma} \label{lem: Integral-H}
  Let $\mathcal{K}^{(\alpha,\beta)}_s$  be the operator acting on $L^2(-s,s)$ with the confluent hypergeometric kernel given in \eqref{chgkernel}, and $H(\tau;\alpha,\beta)$ be the Hamiltonian for the Painlev\'e V equation given in \eqref{sh}. Then,  we have
  \begin{equation}\label{integralforfredholm}
      \det (I - \gamma \mathcal{K}^{(\alpha,\beta)}_s) = \exp\left( \int_0^{-4is} H(\tau; \alpha, \beta) d\tau \right), \qquad s>0, \ \gamma \in [0,1],
  \end{equation}
  for $\alpha > -\frac{1}{2}$ and $\beta \in i \mathbb{R}$.
\end{lemma}

\begin{proof}
As the integral representation for $\gamma = 1$ has been established in \cite[Theorem 3]{Xu:Zhao2020}, we focus on the case $0\leq \gamma < 1$.

To achieve the integral expression, we recall \eqref{differentialfordn} and study the limit in \eqref{constantl}. With \eqref{eq: Psi-near0-small-s}, after a direct computation of $\Psi(\frac{2z}{|\tau|},\tau)^{-1}\Psi^{'}(\frac{2z}{|\tau|},\tau)$, we have
  \begin{eqnarray}
    \left(\Psi_1^{(1)}(\tau)\right)_{21} & = & O(\tau^{2\alpha+1}) + O(\tau), \qquad i \tau \to 0^+, \label{psi01assto0}\\
    \left(\Psi_1^{(-1)}(\tau)\right)_{21} & = & O(\tau^{2\alpha+1}) + O(\tau), \qquad i \tau \to 0^+. \label{psi0-1assto0}
  \end{eqnarray}
  With the approximation \eqref{hassto0}, the above two formulas yield
  \begin{equation}
    \lim_{i \tau \to 0^+} \left[ \tau H(\tau) + \frac{\gamma}{2 \pi i }  \left( e^{\pi i (\alpha-\beta)}
     \left(\Psi_1^{(1)}(\tau)\right)_{21}+e^{-\pi i (\alpha-\beta)}\left(\Psi_1^{(-1)}(\tau)\right)_{21} \right) \right] = 0.
  \end{equation}
  Hence, an application of L'Hospital's rule gives us
  \begin{equation}
    \mathcal{L} = 2i  \lim_{i \tau \to 0^+} \left[  \frac{d}{d \tau}(\tau H) + \frac{\gamma}{2 \pi i }  \left( e^{\pi i (\alpha-\beta)}
     \frac{d}{d \tau}\left(\Psi_1^{(1)}(\tau)\right)_{21}+e^{-\pi i (\alpha-\beta)}\frac{d}{d \tau}\left(\Psi_1^{(-1)}(\tau)\right)_{21} \right) \right].
  \end{equation}
Recalling \eqref{sH+ic}, it is immediate to see that $\mathcal{L} = 0$. Then, integrating both sides of \eqref{differentialfordn} with respect to $t$, we have
  \begin{equation}\label{integralexpressionfordn}
   \frac{D_n(t)}{D_n(0)}=\exp \left(\int_0^{-2int} H(\tau)d \tau +O\left(\frac{1}{n}\right) \right), \qquad n \to \infty,
  \end{equation}
  uniformly for $nt$ bounded, where we also use the property that $H(\tau)$ is pole-free for $\tau \in -i(0, +\infty)$; see Proposition \ref{Prop:Psi-exist}. Substituting the above formula into \eqref{fredholmandtoeplitz}, we obtain  \eqref{integralforfredholm}. 

  This completes the proof of the lemma.
\end{proof}

\subsection{Proof of Theorem \ref{mainresult}}


The differential identities in Proposition \ref{prop: diff-iden} play a significant role in our proof. We first integrate \eqref{inth} to get
\begin{equation}\label{1}
 \begin{aligned}
  \int_0^s H(\tau)d\tau
   & =\int_0^s\left(u_1(\tau)\frac{dv_1(\tau)}{d\tau}+u_2(\tau)\frac{v_2(\tau)}{d\tau}-H(\tau)\right)d\tau \\
   & \qquad
   +\left(\tau H(\tau)+\alpha \ln d(\tau) -\beta \ln y(\tau) - 2(\alpha^2-\beta^2)\ln \tau \right)\big|_{\tau=0}^s.
 \end{aligned}
\end{equation}
To tackle integral on the right-hand side of the above formula, we integrate the next differential identity  \eqref{differentialidentitylambda}  on the both sides about $s$, as well as about $\gamma$. This gives us
\begin{equation*}\label{3}
 \begin{aligned}
  & \int_0^s\left(u_1(\tau)\frac{dv_1(\tau)}{d\tau}+u_2(\tau)\frac{v_2(\tau)}{d\tau}-H(\tau)\right)d\tau  - \int_0^s \left(u_1(\tau)\frac{dv_1(\tau)}{d \tau}+u_2(\tau)\frac{v_2(\tau)}{d \tau}-H(\tau)\right)\bigg|_{\gamma=0} d \tau \\
  &  = \int_0^{\gamma}\left(u_1(s)\frac{\partial v_1(s)}{\partial \tilde{\gamma}}+u_2(s)\frac{\partial v_2(s)}{\partial \tilde{\gamma}}\right) d \tilde{\gamma}
  -\lim_{\tau \to 0} \int_0^{\gamma}\left(u_1(\tau) \frac{\partial v_1(\tau)}{\partial \tilde{\gamma}}+u_2(\tau)\frac{\partial v_2(\tau)}{\partial \tilde{\gamma}}\right)d \tilde{\gamma}.
 \end{aligned}
\end{equation*}
It is obvious to see from \eqref{a1intermsofpsi} and \eqref{a2intermsofpsi} that $A_1(\tau)=A_2(\tau)=0$ when $\gamma=0$. This, together with the asymptotics in \eqref{v1assto0},  \eqref{v2assto0}, \eqref{a1} and \eqref{a2} implies $u_1(\tau)=u_2(\tau)\equiv 0$ when $\gamma  = 0$. From the definition of $H(\tau)$ in \eqref{sh}, we also have $H(\tau) \equiv 0$ when $\gamma  = 0$. This shows that the second term on the left-hand side of the above formula vanishes. In addition, as the asymptotics in Theorem \ref{huvasymptotics1} is uniform in $\gamma$, it is easy to check that the second integral on the right-hand side of the above formula tends to 0 as $i \tau \to 0^+$. Combining the above two formulas, we obtain
\begin{equation}\label{4}
 \begin{aligned}
  \int_0^s H(\tau)d\tau
   & =\int_0^{\gamma}\left(u_1(s)\frac{\partial v_1(s)}{\partial \tilde{\gamma}}+u_2(s)\frac{\partial v_2(s)}{\partial \tilde{\gamma}}\right) d \tilde{\gamma}            \\
   & \qquad +\left(\tau H(\tau)+\alpha \ln d(\tau) -\beta \ln y(\tau) - 2(\alpha^2-\beta^2)\ln \tau \right)\big|_{\tau=0}^s.
 \end{aligned}
\end{equation}

To compute asymptotics of the  first term on the  right-hand side of \eqref{4}, we apply the large-$s$ asymptotics for $u_1,u_2,v_1,v_2$. It follows from \eqref{u1infinity} and \eqref{v1infinity} that
\begin{eqnarray}
u_1(s)\frac{\partial v_1(s)}{\partial \gamma} & =& u_1(s) v_1(s) \frac{\partial \ln v_1(s)}{\partial \gamma}  \nonumber \\
& =& ic\left(-2i \ln s +\pi  + \frac{i\Gamma'(1+ic)}{\Gamma(1+ic)}+\frac{i\Gamma'(1-ic)}{\Gamma(1-ic)}+ O(s^{-1}) \right) \frac{\partial c}{\partial \gamma} .  \label{5}
\end{eqnarray}
Similarly, we have from \eqref{u2infinity} and \eqref{v2infinity} that
\begin{equation}\label{8}
 u_2(s)\frac{\partial v_2(s)}{\partial \gamma}
 =ic\left(-2i \ln s +\pi + \frac{i\Gamma'(1+ic)}{\Gamma(1+ic)}+\frac{i\Gamma'(1-ic)}{\Gamma(1-ic)}  + O(s^{-1})\right)\frac{\partial c}{\partial \gamma} .
\end{equation}
Combining the above two formulas, we get
\begin{eqnarray}
 && \int_0^{\gamma}\left(u_1(s)\frac{\partial v_1(s)}{\partial \tilde{\gamma}}+u_2(s)\frac{\partial v_2(s)}{\partial \tilde{\gamma}}\right) d \tilde{\gamma} \nonumber \\
 && = \int_0^{c} 2 i \tilde{c} \left(-2i \ln s +\pi  + \frac{i\Gamma'(1+i\tilde{c})}{\Gamma(1+i\tilde{c})}+\frac{i\Gamma'(1-i\tilde{c})}{\Gamma(1-ic)}+ O(s^{-1}) \right)  d \tilde{c}
\end{eqnarray}
via a change of variables. With the integral expression of the Barnes $G$-function in \cite[Eq. (5.17.4)]{NIST}, we have from the above formula
\begin{equation}\label{9}
 \begin{aligned}
   \int_0^{\gamma}\left(u_1\frac{\partial v_1}{\partial \tilde{\gamma}}+u_2\frac{\partial v_2}{\partial \tilde{\gamma}}\right) d \tilde{\gamma}
   = c^2(2\ln s +\pi i) +2 \ln (G(1+ic)G(1-ic)) - 2c^2+O\left(s^{-1}\right).
 \end{aligned}
\end{equation}

Last, we drive asymptotics of the last term in the  right-hand of \eqref{4} by straightforward computations. From the large-$\tau$ asymptotics in  \eqref{hinfinity}, \eqref{y} and \eqref{d}, we have, as $i \tau \to +\infty$,
\begin{multline} \label{10}
\tau H(\tau)+\alpha \ln d(\tau) -\beta \ln y(\tau) -2(\alpha^2 - \beta^2)\ln \tau =  -ic \tau+2c^2  \\
   +\alpha\ln\left(2\alpha\frac{\Gamma(1+\alpha-\beta)\Gamma(1+\alpha+\beta)}{(\Gamma(1+2\alpha))^2 2^{2\alpha}}e^{2\pi c}\right)
   -\beta\ln\left(\frac{\Gamma(1+\alpha-\beta)}{\Gamma(1+\alpha+\beta) 2^{2\beta}}\right)
   +O\left(\tau^{-1}\right).
\end{multline}
Similarly,  with \eqref{hassto0}, \eqref{yassto0} and \eqref{dassto0}, we get, as $i \tau \to 0^+$,
\begin{equation}\label{11}
 \begin{aligned}
    & \tau H(\tau)+\alpha \ln d(\tau) -\beta \ln y(\tau) -2(\alpha^2 - \beta^2)\ln \tau  \\
  = & \alpha\ln\left(2\alpha\frac{\Gamma(1+\alpha-\beta)\Gamma(1+\alpha+\beta)}{(\Gamma(1+2\alpha))^2 2^{2\alpha}}\right)
  -\beta\ln\left(\frac{\Gamma(1+\alpha-\beta)}{\Gamma(1+\alpha+\beta) 2^{2\beta}}\right)
  +O\left(\tau^{2\alpha+1}\right).
 \end{aligned}
\end{equation}
Finally, substituting \eqref{9}-\eqref{11} into \eqref{4} and using the integral representation \eqref{integralforfredholm}, we obtain \eqref{fdeterminant}.

This finishes the proof of Theorem \ref{mainresult}. \hfill \qed

\section*{Acknowledgments}

The authors thank the anonymous referees for their careful reading and constructive suggestions. We also thank Professors Shuai-Xia Xu and Yu-Qiu Zhao for helpful discussions and comments.

Dan Dai was partially supported by grants from the City University of Hong Kong (Project No. 7005252 and 7005597), and a grant
from the Research Grants Council of the Hong Kong Special Administrative Region, China (Project No.  CityU 11300520).

\begin{appendices}

\section{Confluent hypergeometric parametrix}\label{sec:Whittaker}


This parametrix is given in \cite[Sec. 4.3]{Its:Kra2008} for the case $a=0$ and \cite[Sec. 8.2]{Cha2019} for general parameter $a$. With $a$ and $b$ be two constants, the $2 \times 2$ matrix-valued function $\Phi(\zeta;a,b)$ satisfies the following RH problem.

\begin{rhp}
 \hfill
 \begin{itemize}
  \item [(a)] $\Phi(\zeta)= \Phi(\zeta;a,b)$ is analytic for $\zeta \in \mathbb{C} \setminus \{\cup_{i=1}^8 \Gamma_i \}$,  where the  contours are defined below
        \begin{equation*}
         \begin{aligned}
           & \Gamma_1=e^{\frac{\pi i}{2}}\mathbb{R}^{+}, \qquad \Gamma_2=e^{\frac{3\pi i}{4}}\mathbb{R}^{+}, \qquad
          \Gamma_3=(-\infty,0), \qquad  \Gamma_4=-e^{-\frac{3 \pi i}{4}}\mathbb{R}^{+}, \qquad  \\
           & \Gamma_5=-e^{-\frac{\pi i}{2}}\mathbb{R}^{+}, \qquad \Gamma_6=-e^{-\frac{\pi i}{4}}\mathbb{R}^{+}, \qquad \Gamma_7=(0,\infty), \qquad \Gamma_8=e^{\frac{\pi i}{4}}\mathbb{R}^{+};
         \end{aligned}
       \end{equation*}
       see Figure \ref{figure7} for an illustration.

        \begin{figure}[h]
          \centering
          \includegraphics[width=3.5in]{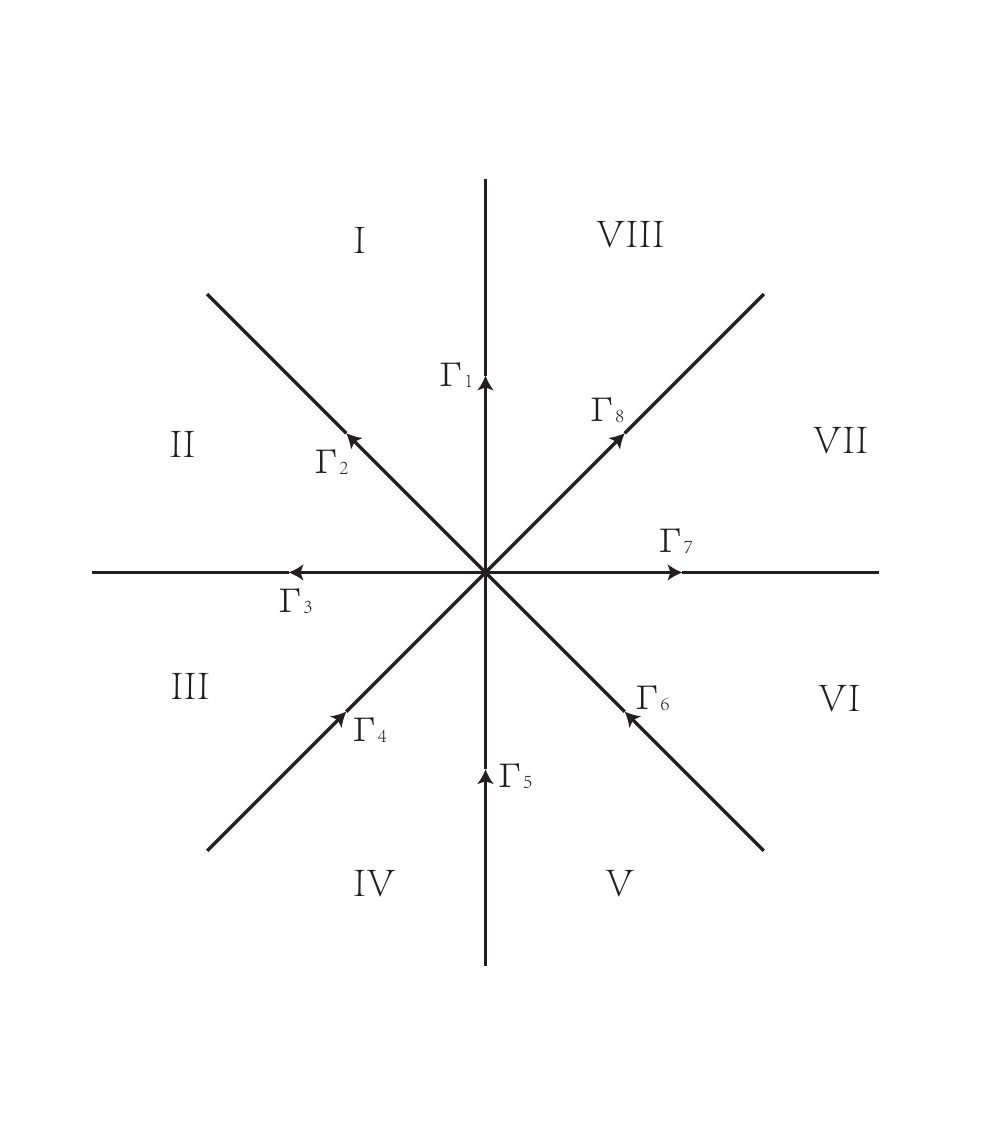}
          \caption{Contours for the RH problem for $\Phi(\zeta)$. Regions I-VIII are also depicted.}
          \label{figure7}
        \end{figure}

  \item [(b)] $\Phi(\zeta)$ satisfies the following jump conditions:
        \begin{equation}\label{jumpforphi}
         \Phi_+(\zeta)=\Phi_-(\zeta)J_{\Phi,i}(\zeta), \qquad \zeta \in \Gamma_i,
       \end{equation}
       where
       \begin{eqnarray}
           && J_{\Phi,1}(\zeta) = \begin{pmatrix} 0 & e^{-\pi i b} \\ -e^{\pi i b} & 0 \end{pmatrix}, \
           J_{\Phi,2}(\zeta) = \begin{pmatrix} 1 & 0 \\ e^{-2\pi i a}e^{\pi i b} & 1 \end{pmatrix}, \
           J_{\Phi,3}(\zeta) = \begin{pmatrix} e^{\pi i a} & 0 \\ 0 & e^{-\pi i a} \end{pmatrix}, \nonumber \\
           && J_{\Phi,4}(\zeta) = \begin{pmatrix} 1 & 0 \\ e^{2\pi i a}e^{-\pi i b} & 1 \end{pmatrix},\
           J_{\Phi,5}(\zeta) = \begin{pmatrix} 0 & e^{\pi i b} \\ -e^{-\pi i b} & 0 \end{pmatrix}, \
           J_{\Phi,6}(\zeta) = \begin{pmatrix} 1 & 0 \\ e^{-2\pi i a}e^{-\pi i b} & 1 \end{pmatrix}, \nonumber  \\
           && J_{\Phi,7}(\zeta) = \begin{pmatrix} e^{\pi i a} & 0 \\ 0 & e^{-\pi i a} \end{pmatrix}, \
           J_{\Phi,8}(\zeta) = \begin{pmatrix} 1 & 0 \\ e^{2\pi i a}e^{\pi i b} & 1 \end{pmatrix}.
       \end{eqnarray}

  \item [(c)] As $\zeta \to \infty$, we have
        \begin{equation}\label{infinitybehaviorforphi}
         \Phi(\zeta)=\left(I + \frac{a^2-b^2}{\zeta} \begin{pmatrix} 1 & \frac{\Gamma(a-b)}{\Gamma(a+b+1)} \\ -\frac{\Gamma(a+b)}{\Gamma(a-b+1)} & -1 \end{pmatrix} + O\left(\frac{1}{\zeta^2}\right) \right)
         \zeta^{-b\sigma_3}e^{-\frac{1}{2}\zeta\sigma_3}C(\zeta),
        \end{equation}
        where the branch cut of $\zeta ^{-b}$ is taken along the negative imaginary axis such that $\arg \zeta \in (-\frac{\pi}{2},\frac{3\pi}{2})$. And, $C(\zeta)$ is the following constant matrix
        \begin{equation}
         C(\zeta)=
         \begin{cases}
           e^{-\frac{1}{2}\pi i a \sigma_3}e^{\pi i b \sigma_3}, \qquad & \zeta \in \mathrm{I} \cup \mathrm{II},     \\
           e^{\frac{1}{2}\pi i a \sigma_3}e^{\pi i b \sigma_3}, \qquad  & \zeta \in \mathrm{III} \cup \mathrm{IV},   \\
           \begin{pmatrix} 0 & -1 \\ 1 & 0 \end{pmatrix} e^{-\frac{1}{2}\pi i a \sigma_3}, \qquad & \zeta \in \mathrm{V} \cup \mathrm{VI},  \\
           \begin{pmatrix} 0 & -1 \\ 1 & 0 \end{pmatrix} e^{\frac{1}{2}\pi i a \sigma_3}, \qquad  & \zeta \in \mathrm{VII} \cup \mathrm{VIII}.
         \end{cases}
        \end{equation}

  \item [(d)] As $\zeta \to 0$, $\Phi(\zeta)$ has the following local behavior:
        \begin{equation} \label{eq: Phi0-local-0}
        \begin{aligned}
          & \mathrm{for} \ a>0, \quad \Phi(\zeta) =
          \begin{cases}
            O\begin{pmatrix} \zeta^{a} & \zeta^{-a} \\ \zeta^{a} & \zeta^{-a} \end{pmatrix}, \qquad & \zeta \in \mathrm{II} \cup \mathrm{III} \cup \mathrm{VI} \cup \mathrm{VII}, \\
            O\begin{pmatrix} \zeta^{a} & \zeta^{a} \\ \zeta^{a} & \zeta^{a} \end{pmatrix}, \qquad & \zeta \in \mathrm{I} \cup \mathrm{IV} \cup \mathrm{V} \cup \mathrm{VIII}.
          \end{cases}
          \\
           & \mathrm{for} \ a=0, \quad \Phi(\zeta) =
           \begin{cases}
             O\begin{pmatrix} 1 & \ln \zeta \\ 1 & \ln \zeta \end{pmatrix}, \qquad & \zeta \in \mathrm{II} \cup \mathrm{III} \cup \mathrm{VI} \cup \mathrm{VII}, \\
             O\begin{pmatrix} \ln \zeta & \ln \zeta \\ \ln \zeta & \ln \zeta \end{pmatrix}, \qquad & \zeta \in \mathrm{I} \cup \mathrm{IV} \cup \mathrm{V} \cup \mathrm{VIII},
           \end{cases}
          \\
         & \mathrm{for} \ a<0, \quad \Phi(\zeta) = O \begin{pmatrix} \zeta^{a} & \zeta^{a} \\ \zeta^{a} & \zeta^{a} \end{pmatrix}, \qquad \qquad \zeta \in \mathbb{C} \setminus \{\cup_{i=1}^8 \Gamma_i \}.
        \end{aligned}
        \end{equation}
 \end{itemize}
\end{rhp}

Let $G$ and $H$ be functions defined in term of the standard   Whittaker functions M($\cdot$) and  W($\cdot$) as follows:
\begin{equation}
 G(s,t,\zeta):=\zeta^{-\frac{1}{2}}\mathrm{M}_{\frac{1}{2}+\frac{t}{2}-s,\frac{t}{2}}(\zeta), \qquad
 H(s,t,\zeta):=\zeta^{-\frac{1}{2}}\mathrm{W}_{\frac{1}{2}+\frac{t}{2}-s,\frac{t}{2}}(\zeta).
\end{equation}
Then, the solution to the above RH problem is given explicitly as, for $\zeta \in \mathrm{II}$,
\begin{equation}\label{phi}
 \Phi_0(\zeta) =
 \begin{pmatrix}
  \frac{\Gamma(1+a-b)}{\Gamma(1+2a)}G(a+b,2a,\zeta)e^{-\frac{3}{2}\pi i a}   &
  -\frac{\Gamma(1+a-b)}{\Gamma(a+b)}H(1+a-b,2a,e^{-\pi i}\zeta) e^{\frac{1}{2}\pi i a}      \\
  \frac{\Gamma(1+a+b)}{\Gamma(1+2a)}G(1+a+b,2a,\zeta)e^{-\frac{3}{2}\pi i a} &
  H(a-b,2a,e^{-\pi i}\zeta)e^{\frac{1}{2}\pi i a}
 \end{pmatrix}.
\end{equation}
The solution in other sectors can be obtained from the above formula and the jump condition in \eqref{jumpforphi}. One can obtain more details about the local behavior as $\zeta \to 0$ in \eqref{eq: Phi0-local-0}. More precisely, we have
\begin{small}
\begin{equation}\label{localbehaviorforphinear0}
  \Phi_0(\zeta)=
  \begin{cases}
  \begin{pmatrix}
    \frac{\Gamma(1+a-b)}{\Gamma(1+2a)} & -\frac{\Gamma(2a)}{\Gamma(a+b)} \\
    \frac{\Gamma(1+a+b)}{\Gamma(1+2a)} & \frac{\Gamma(2a)}{\Gamma(a-b)}
  \end{pmatrix}\left(I+O(\zeta)\right)
  e^{-\frac{3}{2}\pi i a \sigma_3}\zeta^{a \sigma_3},
   & a>\frac{1}{2},
  \\
  \begin{pmatrix}
    \frac{\Gamma(1+a-b)}{\Gamma(1+2a)} & -\frac{\Gamma(2a)}{\Gamma(a+b)} \\
    \frac{\Gamma(1+a+b)}{\Gamma(1+2a)} & \frac{\Gamma(2a)}{\Gamma(a-b)}
  \end{pmatrix}\left(I+O(\zeta\ln\zeta)\right)
  e^{-\frac{3}{2}\pi i a \sigma_3}\zeta^{a \sigma_3},
   & a=\frac{1}{2},
  \\
  \begin{pmatrix}
    \frac{\Gamma(1+a-b)}{\Gamma(1+2a)} & -\frac{\Gamma(2a)}{\Gamma(a+b)} \\
    \frac{\Gamma(1+a+b)}{\Gamma(1+2a)} & \frac{\Gamma(2a)}{\Gamma(a-b)}
  \end{pmatrix}\left(I+O(\zeta^{2a})\right)
  e^{-\frac{3}{2}\pi i a \sigma_3}\zeta^{a \sigma_3},
   & 0<a<\frac{1}{2},
  \\
  \begin{pmatrix}
    \Gamma(1-b) & \frac{\psi(1-b)+2\gamma_E-\pi i}{\Gamma(b)} \\
    \Gamma(1+b) & -\frac{\psi(-b)+2\gamma_E-\pi i}{\Gamma(-b)}
  \end{pmatrix}\left(I+O(\zeta) \right)
  \begin{pmatrix}
    1 & \frac{1}{\Gamma(b)\Gamma(1-b)}\ln \zeta \\ 0 & 1
  \end{pmatrix},
   & a=0,
  \\
  \begin{pmatrix}
    \frac{\Gamma(1+a-b)}{\Gamma(1+2a)}e^{-\frac{3}{2}\pi i a}\zeta^a (1+O(\zeta))
    & -\frac{\Gamma(1+a-b)\Gamma(-2a)}{\Gamma(1-a-b)\Gamma(a+b)}e^{-\frac{1}{2}\pi i a}\zeta^a \left(1+O(\zeta^{-2a})\right) \\
    \frac{\Gamma(1+a+b)}{\Gamma(1+2a)}e^{-\frac{3}{2}\pi i a}\zeta^a (1+O(\zeta))
    & \frac{\Gamma(-2a)}{\Gamma(-a-b)}e^{-\frac{1}{2}\pi i a}\zeta^a
    \left(1+O(\zeta^{-2a})\right)
  \end{pmatrix},
  \ & -\frac{1}{2}<a<0,
 \end{cases}
\end{equation}
\end{small}
where $\gamma_E$ is Euler's constant and $\psi$ is the polygamma function of order 0. In addition, $\Phi(\zeta)$ satisfies a differential equation as follows
\begin{equation}\label{differentialrelationofphi}
 \frac{d \Phi(\zeta)}{d\zeta} \Phi^{-1}(\zeta) = -\frac{1}{2}\sigma_3 + \frac{1}{\zeta}
 \begin{pmatrix} -b & \frac{\Gamma(1+a-b)}{\Gamma(a+b)}e^{4 \pi i a} \\ \frac{\Gamma(1+a+b)}{\Gamma(a-b)}e^{-4\pi i a} & b \end{pmatrix}.
\end{equation}

\end{appendices}

\end{document}